\documentclass[11pt]{amsart}
\usepackage[margin=1in]{geometry}

\usepackage{amssymb}
\usepackage{mathtools}
\usepackage{upgreek}
\usepackage{enumitem}
\usepackage{booktabs}
\usepackage[noadjust]{cite}
\usepackage[hidelinks]{hyperref}
\usepackage{cleveref}
\usepackage{aliascnt}

\numberwithin{equation}{section}
\newtheorem{theorem}{Theorem}[section]

\newaliascnt{proposition}{theorem}
\newtheorem{proposition}[proposition]{Proposition}
\aliascntresetthe{proposition}

\newaliascnt{lemma}{theorem}
\newtheorem{lemma}[lemma]{Lemma}
\aliascntresetthe{lemma}

\newaliascnt{corollary}{theorem}
\newtheorem{corollary}[corollary]{Corollary}
\aliascntresetthe{corollary}

\theoremstyle{definition}
\newaliascnt{definition}{theorem}
\newtheorem{definition}[definition]{Definition}
\aliascntresetthe{definition}

\theoremstyle{remark}
\newaliascnt{remark}{theorem}
\newtheorem{remark}[remark]{Remark}
\aliascntresetthe{remark}

\newaliascnt{notation}{theorem}
\newtheorem{notation}[notation]{Notation}
\aliascntresetthe{notation}
\crefname{appendix}{Appendix}{Appendices}
\Crefname{appendix}{Appendix}{Appendices}

\newcommand{\E}{\mathbb{E}}
\newcommand{\R}{\mathbb{R}}
\newcommand{\N}{\mathbb{N}}

\newcommand{\cC}{\mathcal{C}}
\newcommand{\cF}{\mathcal{F}}
\newcommand{\cN}{\mathcal{N}}
\newcommand{\cS}{\mathcal{S}}

\newcommand{\Var}{\operatorname{Var}}
\newcommand{\Cov}{\operatorname{Cov}}
\newcommand{\supp}{\operatorname{supp}}

\title[Constructive neural network field theory]{Constructive Neural Network Field Theory:\\ $\phi_2^4$ in Finite Volume}
\author{Samuel Frank}
\date{\today}
\address{Department of Physics, Northeastern University, Boston, MA 02115, USA}
\address{The NSF Institute for Artificial Intelligence and Fundamental Interactions (IAIFI)}
\email{frank.sam@northeastern.edu}

\begin{document}

\begin{abstract}
Neural network field theory specifies a Euclidean field theory by a network architecture and a density on its parameters, so that correlation functions are integrals over those parameters.  We introduce constructive neural network field theory and use it to realize finite-volume scalar $\phi^4_2$ as a limit of network measures.  For each dyadic momentum shell of the free propagator, the field is given a separate set of neurons with Gaussian coefficients, uniform phases, and frequencies drawn from the shell's spectral density.  Once the frequencies and phases are fixed, integrating out one shell's coefficients is a finite-dimensional integral against a Gaussian density reweighted by the Wick-ordered quartic interaction.  For widths that grow fast enough with the shell index and almost every draw of the frequencies and phases, the reweighted density is uniformly log-concave for every shell above a threshold, at any coupling in two and three dimensions, and the fluctuations of those shell fields are held to the free-field scale.  In two dimensions, an induction over scales shows that shell $j$ contributes $O(2^{-2j}j^3)$ to the log partition function, and the total-variation distance between marginals at successive cutoffs is of the same order.  Because these contributions are summable, the ultraviolet cutoff can be removed for any suitable fixed sequence of widths.  Almost surely, the ultraviolet and infinite-width limits commute, and in either order the joint laws of finitely many smeared fields converge to those of the standard finite-volume $\phi^4_2$ measure.  The method of parametrizing a field via a network organized into momentum shells and integrating out one shell at a time is not specific to $\phi^4_2$, so we expect it to extend to other theories.
\end{abstract}

\maketitle

\tableofcontents

\section{Introduction}\label{sec:intro}

\subsection{Constructive neural network field theory}
\label{sec:cnnft}

Neural network field theory (NN-FT) defines a Euclidean field theory in terms of a network architecture and a density on the network's parameters \cite{HalversonNN1,HalversonNN2}. At finite width, correlation functions are ordinary finite-dimensional expectations,
\begin{equation}\label{eq:nnft-correlator}
    G_n(x_1,\ldots,x_n)
    \;=\; \int d\theta\,p(\theta)\,
    \phi_\theta(x_1)\cdots\phi_\theta(x_n)\;,
\end{equation}
where $p(\theta)$ is the density and $\phi_\theta(x)$ is the field parametrized by $\theta$. A growing area of research has asked which structures of field theory can be built into an architecture and a density.  Among those realized so far are Euclidean invariance \cite{HalversonNN2}, locality \cite{DemirtasNonGaussian}, conformal symmetry \cite{HalversonConformal,DograConformalSpinning,LeutheusserNaskarAlgebraic}, fermions and supersymmetry \cite{FrankFermions}, and topological sectors \cite{FerkoTopological}.  The approach applies to quantum mechanics as well \cite{FerkoHalversonQM}, and has been used to compute string amplitudes \cite{FrankString}, to derive Ward identities and anomalies \cite{FerkoAnomalies}, and to optimize numerical sampling methods for scalar fields \cite{Zhang2026}.  See \cite{HalversonLectures} for a survey.

This paper asks instead whether the theory specified by a network exists as a measure on distributions.  A Euclidean (scalar) field theory is a probability measure on the space $\cS'(\R^d)$ of tempered distributions \cite[\S6.1]{GlimmJaffe1987}, and in $d\ge2$ the fields are not functions. A network with finitely many parameters defines a measure on a finite-dimensional space, and its fields are ordinary functions.  The theory can appear only once the width is sent to infinity and the ultraviolet cutoff is removed, and whether it appears at all is a question about those limits.

We introduce \emph{constructive neural network field theory} to build field theories as limits of network measures. The closest existing statement in the NN-FT literature is a representation theorem: every probability measure on $\cS'(\R^d)$ admits a network description with countably many parameters \cite{FerkoUniversality}, so any scalar theory that exists can be represented. There the measure is given, and the architecture is shown to exist. A construction must go the other way: fix the architecture and the density, and show that the measures they define converge.  As the first case of constructive NN-FT, we construct the finite-volume Euclidean $\phi^4_2$ measure. Existence of this theory was settled in the 1970s \cite{Simon1974, GlimmJaffe1987}, so this work contributes the method rather than the result.

\subsection{The construction}\label{sec:construct}

For each dyadic momentum shell of the free propagator, the field is given its own set of neurons. Fix a mass $m>0$, and let $Q=(-\Delta+m^2)^{-1}$ be the free covariance. We decompose $Q$ into dyadic shells as $Q=\sum_{j\ge0}Q_j$, with $Q_j$ carrying momenta of order $2^j$.  The field is then
\begin{equation}\label{eq:nnft-idea}
\begin{gathered}
    \phi_\Lambda(x) \;=\; \sum_{j=0}^{\Lambda}\varphi_j(x)\;,
    \qquad
    \varphi_j(x) \;=\; \sum_{\ell=1}^{N_j}
    a_{j,\ell}\,U_{j,\ell}(x)\;,
    \\
    U_{j,\ell}(x) \;=\; \sqrt{\frac{2q_j}{N_j\sigma^2}}\,
    \cos\!\big(b_{j,\ell}\cdot x+c_{j,\ell}\big)\;.
\end{gathered}
\end{equation}
Here $\phi_\Lambda$ is the \emph{scalar field} at cutoff $\Lambda$, $\varphi_j$ is the $j$-th \emph{shell field}, and $U_{j,\ell}$ is its $\ell$-th \emph{neuron}. All parameters are independent, with Gaussian coefficients $a_{j,\ell}$ of variance $\sigma^2$, uniform phases $c_{j,\ell}$, and frequencies $b_{j,\ell}$ drawn from the $j$-th shell's spectral density. These choices make the sampled covariance reproduce $Q_j$ on average (\Cref{lem:cov-mean}).  Equation \eqref{eq:nnft-idea} is a modification of the cos-net scalar architecture introduced in \cite{HalversonNN2}, itself the random Fourier features construction of Rahimi and Recht \cite{RahimiRecht2007}.  The cutoff $\Lambda\in\N$ is the index of the highest shell retained, and we denote the typical ultraviolet cutoff in momentum as $\Lambda_{\mathrm{UV}}=2^\Lambda m$, up to a constant factor.  The width at shell $j$ is $N_j$, and $q_j=Q_j(0)$ is the $j$-th shell's Wick constant, finite because each shell's spectral density has compact support. The interaction in two dimensions is the Wick-ordered quartic,
\begin{equation}\label{eq:intro-interaction}
    V_\Lambda \;=\; \frac{\lambda}{4!}\int g(x)\,
    :\!\phi_\Lambda(x)^4\!:\,dx\;,
\end{equation}
where $\lambda>0$ is the coupling, $g$ is a smooth, nonnegative function whose compact support $K$ is the finite volume on which the theory is defined, and the colons denote Wick ordering at the variance of $\phi_\Lambda(x)$.

The two regulators, $\Lambda$ and $N_j$, are not of the same kind.  The shell cutoff $\Lambda$ fixes which momenta are present, and removing it is the ultraviolet (UV) limit.  The width leaves the momenta untouched but fixes how closely the sampled covariance tracks $Q_j$, and it too must be taken to infinity. Averaged over the frequencies and phases, a finite width also generates interactions suppressed by powers of the inverse width \cite{DemirtasNonGaussian}, and these must be absent from the theory being constructed. The object produced at fixed width is therefore a regularized measure, and the field theory is the double limit.

\subsection{Results}\label{sec:results}

The main results hold in $d=2$ once $N_j\ge c_0(j{+}1)2^{4j}$ for a suitable constant $c_0$, and for almost every draw of the frequencies and phases (\Cref{thm:good-geom}):

\begin{enumerate}[label=(\roman*)]
\item \textbf{Existence at fixed width} (\Cref{thm:d2-existence}).  The partition functions converge, $Z_\Lambda\to Z_\infty\in[1,\infty)$, with shell $j$ contributing $O(2^{-2j}j^3)$.  The finite-shell marginals of the interacting Gibbs measures~\eqref{eq:mu-def} are Cauchy in total variation and define a limit measure $\mu_\infty$ on the full coefficient space.  The smeared field $\langle\phi_\Lambda,f\rangle$ converges in distribution for every $f\in\cS(\R^2)$, and its shell series converges absolutely almost surely.  Wick ordering is the only renormalization, and no smallness of the coupling is assumed.  A larger coupling only raises the threshold above which the estimates hold.

\item \textbf{The two limits commute} (\Cref{thm:interchange}).  Widening the network at fixed cutoff also has a limit, the Gibbs measure of the truncated free field (\Cref{thm:infwidth}), and the ultraviolet and infinite-width limits may be taken in either order without changing the result.

\item \textbf{The limit is the standard measure} (\Cref{thm:LP-standard}, \Cref{cor:d2-construction}).  The common value of the two limits is the finite-volume Wick-ordered $\phi^4_2$ measure built from the full propagator, a probability measure on $\cS'(\R^2)$ that is absolutely continuous with respect to the Gaussian free field.
\end{enumerate}

\subsection{Outline of the proof}\label{sec:outline}

The frequencies and phases are drawn once and then held fixed, leaving the shell field a linear function of its coefficients.  Since different shells have independent coefficients under the Gaussian prior, the log partition function $\log Z_\Lambda$ splits into a sum of one-shell contributions (\Cref{prop:telescope}).  Each comes from integrating the Gibbs factor of the shell interaction against the Gaussian density on one shell's coefficients, with the lower shells held fixed.  The logarithm of this integral is the shell's log-Laplace transform~\eqref{eq:Psi-def}, and the integral provides the normalization for the density of the \emph{tilted measure} with respect to the prior.  In two dimensions, the analysis of a single shell establishes three facts:

\begin{enumerate}
    \item Above a threshold $j_1$, the tilted measures are uniformly log-concave at any coupling (\Cref{cor:BL}).  The Gaussian prior gives the shell's log-density uniform curvature $\sigma^{-2}$, and the quartic can lower that curvature only through its Wick subtraction \eqref{eq:Tzz4}.  The Wick subtraction reaches the Hessian only through the Gram matrix $\Sigma_j$ of the shell's neurons, whose operator norm is $O(2^{-2j})$ in every dimension (\Cref{lem:Schur}).  The curvature is therefore lowered by at most a \emph{one-shell perturbation} of order $\lambda\,D_j\,2^{-2j}$, with $D_j=\sum_{i\le j}q_i$ the cumulative Wick constant (\Cref{prop:hessian}).  Since $D_j$ grows like $j$ in $d=2$ and $2^{(d-2)j}$ in $d\ge3$, the one-shell perturbation decays with increasing $j$ in $d<4$, is constant in $d=4$, and grows in $d>4$.  This is the familiar power counting of the quartic interaction, and dimensions greater than two are discussed in \Cref{sec:higher-dim}.  In two dimensions the perturbation is $O(\lambda\,2^{-2j}j)$, so eventually it is at most half the curvature of the prior.  By the Bakry--\'Emery criterion and the Herbst argument, the centered tilted moments of the shell field are then bounded at the free-field scale (\Cref{cor:tilted-moments}), and the Brascamp--Lieb inequality bounds the variance of any linear functional of the shell's coefficients~\eqref{eq:var-w}.

    \item Above a threshold $j_2\ge j_1$, the tilted mean is bounded on backgrounds that are not too large. On backgrounds bounded on $K$ by a constant times $2^{2j/3}$, the tilted mean is at most $\tfrac12$ (\Cref{lem:mean-shift,lem:d2-good-field}).  The log-Laplace transform is then at most $2^{-2j}$ times a sextic in the background (\Cref{thm:d2-one-shell}).  The proof integrates the interpolation derivative (\Cref{app:interpolation}).

    \item Above a third threshold $j_3$, larger backgrounds are super-exponentially rare under the free measure.  A chaining bound gives a Gaussian tail to the supremum of the background on $K$, and these backgrounds have probability at most $C\exp(-c\,2^{2j/3}/j)$ (\Cref{lem:d2-sup-tail} and \Cref{cor:d2-bad-prob}).  The lower bound $-6D^2$ on the Wick-ordered quartic (\Cref{lem:V-lower}) bounds the Gibbs factor by $e^{C(j+1)^2}$, and the rarity outweighs that bound.
\end{enumerate}

Uniform log-concavity, the bound on the tilted mean, and the rarity of large backgrounds all hold above $j_4=\max\{j_2,j_3\}$.  Together they bound each shell's contribution to $\log Z_\Lambda$ in terms of the sixth moment of the background, and an induction over shells keeps that moment of order $j^3$ under the interacting measure (\Cref{thm:d2-moments}).  Each contribution is therefore $O(2^{-2j}j^3)$, and since these contributions are summable, $\log Z_\Lambda$ converges as $\Lambda\to\infty$ (\Cref{thm:d2-convergence}).  The same rate bounds the total-variation distance between the marginals at successive cutoffs and the expected size of each shell's pairing with a Schwartz function. Consequently, the marginals converge, and the shell series for the smeared field converges absolutely almost surely (\Cref{thm:d2-existence}(b),(c)).

These rates hold uniformly in the width.  The estimates use the features only through the concentration bounds of \Cref{thm:good-geom} and \Cref{lem:unif-gg}, and through deterministic bounds such as $q^\omega_j\le2q_j$.  Every constant in the rates is therefore free of the width (\Cref{prop:uniform-N}), and the Moore--Osgood theorem makes the infinite-width and ultraviolet limits commute (\Cref{thm:interchange}).  The identification with the standard measure uses the $L^4$ convergence of the truncated covariance to $Q$, together with the partition-function bound at twice the coupling, which bounds the exponential of the interaction in $L^2$ uniformly in the cutoff (\Cref{thm:LP-standard}).

\subsection{Relation to existing constructions}
\label{sec:relation}

Existence of the finite-volume $\phi^4_2$ measure is long established.  A lower bound on the interaction at a momentum cutoff, combined with Gaussian moment bounds showing that the part above the cutoff is rarely large, makes the exponential of the Wick-ordered interaction integrable \cite[Thm.~8.6.2]{GlimmJaffe1987}, and hypercontractivity of the free field \cite{Nelson1973} gives those moment bounds directly.  The phase space cell and cluster expansions that reached the infinite-volume $\phi^4_3$ measure at weak coupling integrate out the field one momentum scale at a time \cite{FeldmanOsterwalder1976,MagnenSeneor1976}.  The renormalization group does the same in continuous form, controlling at each scale a field-space effective potential whose evolution in the scale is governed by Polchinski's equation \cite{Polchinski1984}.  Working with a scale decomposition, Barashkov and Gubinelli write the log partition function as the value of a stochastic optimal control problem and use this variational formula to construct finite-volume $\phi^4_3$ \cite{BarashkovGubinelli2020}.  Bauerschmidt and Bodineau turn the Polchinski flow into a criterion for a logarithmic Sobolev inequality.  The criterion requires the Hessian of the effective potential, conjugated by the heat semigroup, to be bounded below at each scale by a number that may be negative, subject to an integrability condition on how those numbers accumulate \cite[Thm.~1.2]{BauerschmidtBodineau2021}.  Their method recovers Bakry--\'Emery \cite{BakryEmery1985} when those numbers are nonnegative, but allowing them to be negative generalizes it to measures that are not log-concave.  Verifying this criterion for $\phi^4$ is difficult, because the effective potential is defined by a Gaussian integral, and bounding its Hessian at a given scale means controlling the variance of the fluctuation measure \cite[Lem.~3.12]{BauerschmidtBodineauDagallier2024}.  Bauerschmidt and Dagallier nonetheless accomplished this for the $\phi^4_2$ and $\phi^4_3$ measures uniformly in the lattice regularization using correlation inequalities for Ising models, the Perron--Frobenius theorem, and bounds on the susceptibility obtained from skeleton inequalities \cite{BauerschmidtDagallier2024}.

Like the renormalization group, the NN-FT construction presented here proceeds one scale at a time.  Integrating out shell $j$ gives the shell log-Laplace transform~\eqref{eq:Psi-def}, a functional of the lower shells.  Unlike the Bauerschmidt--Bodineau criterion, however, the construction never needs the convexity of that functional.  The only convexity it uses is that of the tilted measure, the Gaussian prior on one shell's coefficients reweighted by the shell interaction.  This measure has an explicit polynomial log-density whose Hessian can be bounded directly (\Cref{sec:outline}), and above a fixed shell index it is uniformly log-concave, independent of the lower shells (\Cref{cor:BL}).  Where convexity in the work above yields functional inequalities for measures already constructed, here it yields the construction.

The setting of this paper is simple (two dimensions, finite volume, and a fixed mass), but the construction's appeal lies in the elementarity of its tools.  Hoeffding's inequality, the strong law of large numbers, and Borel--Cantelli control the sampled features, the largest row sum controls the Gram matrix, and a chaining bound on the supremum of a Gaussian field controls large backgrounds.  The rest of the proof uses only standard probability, inequalities of H\"older and Jensen type, Gaussian integration by parts, Hermite orthogonality, and three facts about a uniformly log-concave measure on finitely many coefficients: the Bakry--\'Emery criterion \cite{BakryEmery1985}, the Herbst argument \cite[Thm.~5.3]{ledoux2001concentration}, and the Brascamp--Lieb variance inequality \cite[Thm.~4.1]{BrascampLieb1976}.

\subsection{Organization and conventions}\label{sec:scope-intro}

We treat only the finite-volume ultraviolet problem, not the thermodynamic limit or the Osterwalder--Schrader axioms.  Because the measure constructed is the standard one (\Cref{thm:LP-standard}), both questions may be approached by any method.  Whether they are amenable to NN-FT is discussed in \Cref{sec:open}.

\Cref{sec:arch,sec:good-geom,sec:one-shell} hold in every dimension $d\ge2$.  They set up the architecture and the Wick-ordered quartic interaction, show that the sampled covariance concentrates, and analyze the integral over a single shell.  \Cref{sec:d2-setup,sec:d2-converge} specialize to two dimensions and remove the ultraviolet cutoff at fixed width, and \Cref{sec:d2-standard} takes the width to infinity and identifies the limit.  \Cref{sec:open} collects open problems.  \Cref{app:free-field} shows that the architecture reproduces the Gaussian free field in the infinite-width limit, in every dimension and with no renormalization.  \Cref{app:interpolation} derives the interpolation derivative, and \Cref{app:sup-tail} proves the tail bound for the supremum of the background field.

Throughout, $C$, $C'$, $c$, and $c'$ denote positive constants whose values may change from line to line.  The numbered constants $C_1,C_2,\ldots$ are fixed within the proof in which they appear, and every other named constant, such as $c_0$, $c_q$, or $C_T$, keeps one value throughout the paper.  The \emph{fixed data} are $d$, $m$, $\chi$, $g$, $\sigma$, and $\lambda$.

For positive quantities $f$ and $f'$ we write $f\asymp f'$ when the ratio $f/f'$ is bounded above and below by positive constants depending only on the fixed data.  For a probability measure $\mu$ and an event $A$ we write $\mu(A)$ for the probability of $A$ under $\mu$, and the symbol $\Pr$ is reserved for probabilities over the draw of the frequencies and phases (as in \Cref{thm:good-geom}).  For symmetric matrices, $A\ge B$ means that $A-B$ is positive semidefinite, that is, $v^TAv\ge v^TBv$ for every vector $v$.

\section{The architecture and the interaction}\label{sec:arch}

\subsection{Dyadic decomposition of the propagator}
\label{sec:dyadic}

We work in Euclidean space $\R^d$, $d\ge2$, with mass $m>0$. The free massive scalar field has covariance (propagator) $Q=(-\Delta+m^2)^{-1}$, whose kernel is
\begin{equation}\label{eq:propagator}
    Q(x-y) \;=\; \int_{\R^d}\frac{d^dp}{(2\pi)^d}\,\frac{\cos(p\cdot(x-y))}{p^2+m^2}\,.
\end{equation}
Fix a smooth, radial function $\chi$ taking values in $[0,1]$, not identically zero, supported in the annulus $\{\,p:\tfrac{1}{2}\le|p|\le 2\,\}$, and such that $\chi_0(p):=1-\sum_{j\ge 1}\chi(2^{-j}p)$ is smooth, nonnegative, and supported in $\{|p|\le 2\}$ (such a $\chi$ exists, as in \cite[Prop.~2.10]{BahouriCheminDanchin2011}). Then,
\begin{equation}\label{eq:dyadic}
    \frac{1}{p^2+m^2}
    \;=\; \sum_{j=0}^{\infty} S_j(p)\,,\qquad S_j(p) \;:=\;
    \begin{cases}
    \chi_0(p)/(p^2+m^2) & j=0\;,\\[4pt]
    \chi(2^{-j}p)/(p^2+m^2) & j\ge 1\;.
    \end{cases}
\end{equation}
Decomposing the propagator by scale is a basic tool of the multiscale renormalization group \cite[\S III]{Gallavotti1985}, and \eqref{eq:dyadic} is the Littlewood--Paley decomposition of the propagator \cite[\S2.2]{BahouriCheminDanchin2011}.

For $j\ge 1$, the shell spectral function $S_j$ is supported in the annulus $2^{j-1}\le|p|\le 2^{j+1}$ and satisfies the pointwise bound
\begin{equation}\label{eq:Sj-bound}
    S_j(p) \;=\; \frac{\chi(2^{-j}p)}{p^2+m^2}
    \;\le\; \frac{1}{2^{2(j-1)}}
    \;=\; 4\cdot 2^{-2j}\,,
\end{equation}
since $\chi\le 1$ and $|p|\ge 2^{j-1}$ on the support of $S_j$.  For each $j$, the shell covariance kernel and its diagonal value (the \emph{shell Wick constant}) are
\begin{subequations}\label{eq:Qj}
\begin{align}
    Q_j(x-y) &\;:=\; \int_{\R^d}\frac{d^dp}{(2\pi)^d}\,
    S_j(p)\,\cos\!\big(p\cdot(x-y)\big)\;,
    \label{eq:Qj-kernel}\\[5 pt]
    q_j &\;:=\; Q_j(0)
    \;=\; \int_{\R^d}\frac{d^dp}{(2\pi)^d}\,S_j(p)\,.
    \label{eq:qj-def}
\end{align}
\end{subequations}
Because $S_j$ is smooth and compactly supported, both integrals converge absolutely.  The coincident-point singularity of $Q$ is therefore absent from every individual shell and reappears only in the sum $Q=\sum_{j\ge0}Q_j$ (\Cref{rem:coincident}).  We regularize by replacing $Q$ with its \emph{Littlewood--Paley truncation} $Q_{\le\Lambda}:=\sum_{j\le\Lambda}Q_j$ at level $\Lambda$.

\begin{lemma}[Shell dimensional analysis]
\label{lem:shell-dim}
For $j\ge 1$, the following bounds hold with constants depending only on $d$ and $\chi$, and in (d) also on $\|g\|_\infty$.  The mass $m$ enters only through the lower bounds in (a) and (e) and, when $d=2$, through $q_0$ in the upper bound in (e).
\begin{enumerate}[label=(\alph*)]
\item $q_j\asymp2^{(d-2)j}$.
\item $\|Q_j\|_{L^1(\R^d)}=O(2^{-2j})$.
\item $\|Q_j\|_{L^\infty(\R^d)}=q_j$.
\item For a smooth, nonnegative, compactly supported $g$ with $K:=\supp g$, the row integral $A_j:=\sup_{x\in K}\int g(y)\,|Q_j(x-y)|\,dy$ satisfies $A_j=O(2^{-2j})$.
\item For $J\ge1$, the cumulative Wick constant $D_J:=\sum_{j=0}^{J}q_j$ satisfies $D_J\asymp J$ if $d=2$ and $D_J\asymp2^{(d-2)J}$ if $d\ge3$.
\end{enumerate}
\end{lemma}

\begin{proof}
\textit{Part (a).} We first prove that $q_j$ is bounded above. Since $S_j$ is supported on the annulus $2^{j-1}\le|p|\le 2^{j+1}$ and satisfies $S_j(p)\le 4\cdot 2^{-2j}$, and the volume of this annulus in $\R^d$ is $O(2^{dj})$:
\begin{equation}\label{eq:q_j-upper}
    q_j
    \;\le\; \frac{4\cdot 2^{-2j}\cdot O(2^{dj})}{(2\pi)^d}
    \;=\; C\,2^{(d-2)j}\;.
\end{equation}
For the lower bound, substitute $u=2^{-j}p$ into \eqref{eq:qj-def}:
\begin{equation*}
    q_j \;=\; \int_{\R^d}\frac{2^{dj}\,d^du}{(2\pi)^d}\,
    \frac{\chi(u)}{2^{2j}u^2+m^2}\,.
\end{equation*}
On the support of $\chi$, we have $u^2\le 4$, so $2^{2j}u^2+m^2\le(4+m^2)\,2^{2j}$. Therefore,
\begin{equation}
    q_j \;\ge\;
    \frac{2^{(d-2)j}}{(4+m^2)\,(2\pi)^d}
    \int_{\R^d}\chi(u)\,d^du
    \;=\; c\,2^{(d-2)j}\;,
\end{equation}
with $c>0$ since $\chi\ge 0$ and $\chi\not\equiv 0$.

\medskip\noindent\textit{Part (b).}
Substituting $u=2^{-j}p$ in~\eqref{eq:Qj-kernel} gives
\begin{equation*}
    Q_j(z) \;=\; 2^{(d-2)j}\,\hat{G}_j\big(2^jz\big)\;,
\end{equation*}
where
\begin{equation*}
    G_j(u) \;:=\; \frac{\chi(u)}{u^2+m^2\,2^{-2j}}\;,
    \qquad
    \hat{G}_j(w) \;:=\; \int_{\R^d}\frac{d^du}{(2\pi)^d}\,G_j(u)\,\cos(u\cdot w)\;.
\end{equation*}
Changing variables in the $z$ integral gives $\|Q_j\|_{L^1}=2^{(d-2)j}\cdot2^{-dj}\|\hat{G}_j\|_{L^1}=2^{-2j}\|\hat{G}_j\|_{L^1}$.  It therefore suffices to bound $\|\hat{G}_j\|_{L^1}$ uniformly in $j$.

Two properties of $G_j$ give this uniform bound.  First, it is supported in the fixed annulus $\{\tfrac12\le|u|\le2\}$, where $1/4\leq u^2\leq4$ bounds for every $j$ both the denominators and the powers of $u$ that result from differentiation. Thus, $G_j$ and its derivatives are bounded, along with their $L^1$ norms, by constants free of $j$ and~$m$.  Second, $\hat{G}_j$ decays faster than any power of $|w|$.  Fix $N$ with $2N>d$.  Since $\Delta_u\cos(u\cdot w)=-|w|^2\cos(u\cdot w)$, each Laplacian moved off the cosine gives a factor $|w|^{-2}$, and because $G_j$ is smooth and compactly supported, moving $N$ of them onto $G_j$ leaves no boundary terms:
\begin{equation}\label{eq:Ghat-decay}
    \big|\hat{G}_j(w)\big|
    \;=\; \frac{1}{|w|^{2N}}
    \bigg|\int_{\R^d}\frac{d^du}{(2\pi)^d}\,
    \big(\Delta^NG_j\big)(u)\,\cos(u\cdot w)\bigg|
    \;\le\; \frac{\big\|\Delta^NG_j\big\|_{L^1}}
    {(2\pi)^d\,|w|^{2N}}\;.
\end{equation}
Bounding $|\hat{G}_j|\le(2\pi)^{-d}\|G_j\|_{L^1}$ on $|w|\le1$ and using~\eqref{eq:Ghat-decay} on $|w|>1$, where $\int_{|w|>1}|w|^{-2N}d^dw<\infty$ because $2N>d$, gives $\|\hat{G}_j\|_{L^1}\le C$ with $C$ free of $j$.

\medskip\noindent\textit{Part (c).}
Since $S_j\ge0$ and $|\cos|\le 1$ in \eqref{eq:Qj-kernel}, we have $|Q_j(z)|\le Q_j(0)=q_j$ for all $z$, with equality at $z=0$.

\medskip\noindent\textit{Part (d).}
By (b), $A_j\le\|g\|_\infty\,\|Q_j\|_{L^1}=O(2^{-2j})$.

\medskip\noindent\textit{Part (e).}
Part (a) bounds $q_j$ for $j\ge1$, and it remains to bound $q_0$. The zeroth shell has $0\le\chi_0\le1$ and is supported in $\{|p|\le2\}$, so
\begin{equation*}
    q_0 \;\le\; \int_{|p|\le2}\frac{d^dp}{(2\pi)^d}\,\frac{1}{p^2+m^2}\;,
\end{equation*}
which equals $\frac{1}{4\pi}\log(1+4/m^2)$ at $d=2$, and, because $|p|^{-2}$ is integrable at the origin when $d\ge3$, is at most a constant depending only on $d$ otherwise.

At $d=2$ each $q_j$ with $j\ge1$ lies between two constants $c$ and $C$, so for $J\ge1$
\begin{equation*}
    cJ\;\le\;D_J\;\le\;q_0+CJ\;\le\;(q_0+C)\,J\;.
\end{equation*}
At $d\ge3$ the bounds $q_j\le C\,2^{(d-2)j}$ form a geometric series with ratio at least $2$, whose sum up to $J$ is at most twice its last term, so
\begin{equation*}
    c\,2^{(d-2)J}\;\le\;q_J\;\le\;D_J\;\le\;q_0+2C\,2^{(d-2)J}\;\le\;(q_0+2C)\,2^{(d-2)J}\;.
\end{equation*}
In both cases, $D_J$ lies between constant multiples of the claimed rate, which is (e).
\end{proof}

\begin{remark}[Coincident-point singularity]
\label{rem:coincident}
The cumulative Wick constant $D_\Lambda=\sum_{j=0}^\Lambda q_j$ controls the coincident-point singularity of the free field truncated at shell $\Lambda$.  For a fixed $\Lambda\in\N$, related to the frequency cutoff by $\Lambda=\log_2(\Lambda_{\mathrm{UV}}/m)$ up to an additive constant,
\begin{equation}\label{eq:D_scaling}
    D_\Lambda \;\asymp\;
    \begin{cases}
    \log(\Lambda_{\mathrm{UV}}/m) & d=2
    \quad\text{(logarithmic)}\;,\\
    \Lambda_{\mathrm{UV}} & d=3
    \quad\text{(linear)}\;,\\
    \Lambda_{\mathrm{UV}}^2 & d=4
    \quad\text{(quadratic)}\;,
    \end{cases}
\end{equation}
all of which are the standard scaling of the Wick constant with the upper frequency cutoff $\Lambda_{\mathrm{UV}}$. Wick ordering removes this divergence from the interaction.
\end{remark}

\subsection{Random feature shell fields}\label{sec:shell-fields}

Fix shell $j\ge 0$ and network width $N_j\in\N$.  We sample $N_j$ independent \emph{features} $\omega_j=\{(b_{j,\ell},c_{j,\ell})\}_{\ell=1}^{N_j}$, where
\begin{equation}\label{eq:features}
    b_{j,\ell} \;\sim\; \pi_j(p)
    \;:=\; \frac{S_j(p)}{(2\pi)^d\,q_j}\,,
    \qquad
    c_{j,\ell} \;\sim\; \mathrm{Unif}([-\pi,\pi])\;,
\end{equation}
independently for each $\ell$.  The density $\pi_j$ is a probability distribution on $\R^d$ supported in the $j$-th frequency annulus (or in the disk of radius $2$ for $j=0$). The constant $q_0$ is positive because $\chi_0(0)=1$.

\begin{definition}[Neurons and shell field]
\label{def:shell-field}
Define the $\ell$-th \emph{neuron} of shell $j$ at spatial point $x$ (conditioned on features $\omega_j$) by
\begin{equation}\label{eq:feature-vec}
    U_{j,\ell}(x)
    \;:=\; \sqrt{\frac{2q_j}{N_j\sigma^2}}\;
    \cos(b_{j,\ell}\cdot x + c_{j,\ell})\;,
\end{equation}
and write $U_j(x):=(U_{j,1}(x),\ldots, U_{j,N_j}(x))^T\in\R^{N_j}$ for the vector of the shell's neurons at~$x$.  The $j$-th \emph{shell field} is
\begin{equation}\label{eq:shell-field}
    \varphi_j(x)
    \;:=\; \sum_{\ell=1}^{N_j}
    a_{j,\ell}\,U_{j,\ell}(x)\;,
\end{equation}
where $a_j=(a_{j,1},\ldots,a_{j,N_j})^T \sim\cN(0,\sigma^2 I_{N_j})$ is a vector of Gaussian coefficients (``weights'') independent of the features and of all other shells.
\end{definition}

\begin{definition}[Quenched shell covariance]
\label{def:quenched-cov}
Conditioned on features $\omega_j$, the shell field $\varphi_j(x)$ is a centered Gaussian process in $a_j$ with covariance
\begin{equation}\label{eq:quenched-cov}
    \Gamma_j^\omega(x,y)
    \;:=\; \sigma^2\,U_j(x)^T\,U_j(y)
    \;=\; \frac{2q_j}{N_j}\sum_{\ell=1}^{N_j}
    \cos(b_{j,\ell}\cdot x + c_{j,\ell})\,
    \cos(b_{j,\ell}\cdot y + c_{j,\ell})\;.
\end{equation}
Its diagonal value, the \emph{quenched Wick constant}, is $q_j^\omega(x):=\Gamma_j^\omega(x,x)$. Here and below, \emph{quenched} is used as in the theory of disordered systems: the features are frozen disorder, drawn once and held fixed, while the coefficients $a_j$ remain the degrees of freedom being integrated.
\end{definition}

\begin{lemma}[Quenched covariance reproduces the shell kernel]
\label{lem:cov-mean}
The expectation of the quenched covariance over the features recovers the deterministic shell kernel:
\begin{equation}\label{eq:cov-mean}
    \E_{\omega_j}\!\big[\Gamma_j^\omega(x,y)\big]
    \;=\; Q_j(x-y)\;.
\end{equation}
\end{lemma}

\begin{proof}
Each summand of~\eqref{eq:quenched-cov} is an i.i.d.\ copy of $2q_j\cos(b_j\cdot x+c_j)\cos(b_j\cdot y+c_j)$.  By the product-to-sum formula for cosines and the independence of $b$ and $c$:
\begin{align}
    \E\big[2q_j\cos(b\cdot x+c)\cos(b\cdot y+c)\big]
    &\;=\; q_j\,\E_b\!\big[\cos(b\cdot(x-y))\big]
    + q_j\,\E_b\!\big[\E_c[\cos(b\cdot(x+y)+2c)]\big]
    \notag\\
    &\;=\; q_j\,\E_b\!\big[\cos(b\cdot(x-y))\big]\;,
    \label{eq:cos-mean}
\end{align}
since $\E_c[\cos(\theta+2c)]=0$ for any $\theta$.  The remaining expectation over $b\sim\pi_j$ is
\begin{equation}
    q_j\int_{\R^d}\pi_j(b)\,\cos(b\cdot(x-y))\,d^db
    \;=\; \int_{\R^d}\frac{d^db}{(2\pi)^d}\,S_j(b)\,
    \cos(b\cdot(x-y))
    \;=\; Q_j(x-y)\;,
\end{equation}
by the definition \eqref{eq:features} of $\pi_j$ and the definition \eqref{eq:Qj-kernel} of $Q_j$.
\end{proof}

The deterministic kernel $Q_j$ is a function of $x-y$ and therefore translation invariant. In contrast, the quenched kernel $\Gamma^\omega_j(x,y)$ and its diagonal value $q^\omega_j(x)$ are not. Averaging over the phases $c$, as in the proof of \Cref{lem:cov-mean}, eliminates the dependence on $x+y$ and hence restores translation invariance.

Fix $\Lambda\in\N$ to be the upper UV cutoff on shell indices.  For $0\le J\le\Lambda$, define the cumulative field, the quenched cumulative Wick constant, and the sigma-algebra of the first $J{+}1$ shells:
\begin{align}
    B_J(x) &\;:=\; \sum_{j=0}^{J}\varphi_j(x),
    &
    D^\omega_J(x) &\;:=\; \sum_{j=0}^{J}q_j^\omega(x),
    &
    \cF_J &\;:=\;
    \sigma\big(\omega_0,a_0,\ldots,\omega_J,a_J\big)\;.
    \label{eq:cumulative}
\end{align}
The cumulative field $B_\Lambda$ is the scalar field $\phi_\Lambda$ of~\eqref{eq:nnft-idea}, and we write $\phi_\Lambda$ where it plays that role. The coefficient vectors $a_0,\ldots,a_J$ are mutually independent, and conditioned on all features $\omega:=(\omega_0,\ldots,\omega_J)$, $B_J(x)$ is a sum of $J+1$ independent centered Gaussians with total variance
\begin{equation}\label{eq:sigmaJ}
    \Var(B_J(x)\mid\omega)
    \;=\; \sum_{j=0}^{J}q_j^\omega(x) \;=\; D^\omega_J(x).
\end{equation}
For $j\ge 1$, we refer to $B_{j-1}$ as the \emph{background field} for shell~$j$. Throughout, the superscript $\omega$ marks only the quenched quantities that have a deterministic counterpart, namely $\Gamma_j^\omega$, $q_j^\omega$, $A_j^\omega$, and $D_J^\omega$.  The superscript is the full realization $\omega$ even when a quantity depends only on $\omega_j$ (e.g., $\Gamma^\omega_j$).

The supremum bound of \Cref{lem:d2-sup-tail} and the width limit of \Cref{lem:CK-conv} also need a bound on the increments of $B_J$.

\begin{lemma}[Increment variance]\label{lem:incr-var}
For every $J\ge0$ and all $x,y\in\R^d$, conditioned on the features,
\begin{equation}\label{eq:incr-var}
    \Var\big(B_J(x)-B_J(y)\mid\omega\big)
    \;\le\; \Upsilon_J\,|x-y|^2 ,
    \qquad
    \Upsilon_J \;:=\; \tfrac{32}{3}\Big(\max_{j\le J}q_j\Big)4^{J}\;,
\end{equation}
and $\Upsilon_J$ does not depend on the widths $(N_j)_{j\leq J}$.
\end{lemma}

\begin{proof}
Conditioned on features, $B_J(x)-B_J(y)$ is a linear combination of the weights with fixed coefficients, and the weights are independent and centered with variance $\sigma^2$:
\begin{equation*}
    \Var\big(B_J(x)-B_J(y)\mid\omega\big)
    \;=\;\sigma^2\sum_{j\le J}\sum_{\ell=1}^{N_j}
    \big(U_{j,\ell}(x)-U_{j,\ell}(y)\big)^2 .
\end{equation*}
Since $\cos$ is $1$-Lipschitz, and $|b_{j,\ell}|\leq2^{j+1}$ on the support of $\pi_j$,
\begin{equation*}
   |U_{j,\ell}(x)-U_{j,\ell}(y)|\le\sqrt{2q_j/(N_j\sigma^2)}\,2^{j+1}\,|x-y|\;. 
\end{equation*}
Squaring and summing over the $N_j$ features of shell $j$ cancels the width,
\begin{equation*}
    \sigma^2\sum_{\ell=1}^{N_j}\big(U_{j,\ell}(x)-U_{j,\ell}(y)\big)^2
    \;\le\;2q_j4^{j+1}|x-y|^2\;,
\end{equation*}
and summing over $j\le J$ gives $\sum_{j\le J}2q_j4^{j+1}\le8\big(\max_{j\le J}q_j\big)\tfrac{4^{J+1}-1}{3}\le\Upsilon_J$.
\end{proof}

\subsection{Interaction, partition function, and Gibbs measure}\label{sec:interaction}

The Wick-ordered polynomials relative to Wick constant $C$ are the (probabilist's) Hermite polynomials \cite[Thm.~3.19]{Janson1997}:
\begin{subequations}
\begin{align}
    H_0(u;C) &\;:=\; 1\;, \label{eq:H0}\\
    H_1(u;C) &\;:=\; u\;, \label{eq:H1}\\
    H_2(u;C) &\;:=\; u^2 - C\;, \label{eq:H2}\\
    H_3(u;C) &\;:=\; u^3 - 3Cu\;, \label{eq:H3}\\
    H_4(u;C) &\;:=\; u^4 - 6Cu^2 + 3C^2\;.
    \label{eq:H4}
\end{align}
\end{subequations}
If $u\sim\cN(0,C)$, then $\E_u[H_k(u;C)]=0$ for all $k\ge 1$. We use the derivative identities
\begin{equation}\label{eq:H4-derivs}
    \frac{\partial}{\partial u}H_4(u;C)
    \;=\; 4\,H_3(u;C)\;,
    \qquad
    \frac{\partial}{\partial C}H_4(u;C)
    \;=\; -6\,H_2(u;C)\;,
\end{equation}
for all $u\in\R$ and $C\ge0$.

Fix a smooth, nonneg\-ative, compactly supported spatial cutoff $g:\R^d\to[0,\infty)$, not identically zero, with $K:=\supp g$.  Set $t:=\lambda/24>0$ to be the rescaled quartic coupling.  The \emph{renormalized interaction} at shell index $J$, for $0\le J\le\Lambda$, is
\begin{equation}\label{eq:V-def}
    V_J \;:=\;
    t\int g(x)\,H_4\big(B_J(x);\,D^\omega_J(x)\big)
    \,dx\;.
\end{equation}

\begin{remark}[No counterterm beyond Wick ordering]
\label{rem:no-counterterms}
Wick ordering the interaction removes the coincident-point divergence, and in $d=2$ no further counterterms are needed because the shell contributions to $\log Z_\Lambda$ are $O(2^{-2j}j^3)$ (\Cref{thm:d2-convergence}).  Above two dimensions, Wick ordering alone no longer suffices.  In $d=3$ a mass counterterm is needed because the propagator diverges like $|x-y|^{-1}$ rather than logarithmically, so products of three or more propagators between distinct points are no longer integrable. Developing that structure is outside the scope of this paper (\Cref{sec:higher-dim}).
\end{remark}

The \emph{partition function} and \emph{Gibbs measure} are
\begin{align}
    Z_J &\;:=\;
    \E_{\gamma_J}\!\big[e^{-V_J}\mid\omega\big],
    \label{eq:Z-def}\\[5 pt]
    d\mu_J &\;:=\;
    \frac{e^{-V_J}}{Z_J}\,d\gamma_J,
    \label{eq:mu-def}
\end{align}
where $d\gamma_J:=\prod_{j=0}^J d\gamma_\sigma(a_j)$ is the product Gaussian prior, with $d\gamma_\sigma(a_j)$ the $\cN(0,\sigma^2 I_{N_j})$ distribution on $\R^{N_j}$.  All expectations unless denoted otherwise are conditioned on the features $\omega$, which we suppress when clear from context (\Cref{not:quench}).

\begin{lemma}[Deterministic bounds on $V_J$ and $Z_J$]
\label{lem:V-lower}
The Wick-ordered quartic obeys
\begin{equation}\label{eq:H4-lower}
    H_4(u;D)\;\ge\;-6D^2\qquad\text{for all }u\in\R,\ D\ge0\;.
\end{equation}
Consequently, for all $0\leq J\leq\Lambda$ and all coefficients $a_0,\dots,a_J$,
\begin{equation}\label{eq:V-det-lower}
    V_J \;\ge\;
    -6t\int g\,\big(D^\omega_J\big)^2dx\;,
\end{equation}
and in particular $e^{-V_J}\le e^{C_V(J)}$ with
\begin{equation}\label{eq:CV-def}
    C_V(J) \;:=\; 6t\|g\|_1\big\|D^\omega_J\big\|_\infty^2
    \;<\;\infty \;.
\end{equation}
Moreover $1\le Z_J\le e^{C_V(J)}$.
\end{lemma}

\begin{proof}
Completing the square, $H_4(u;D)=(u^2-3D)^2-6D^2\ge-6D^2$, which is~\eqref{eq:H4-lower}; substituted into~\eqref{eq:V-def}, it gives~\eqref{eq:V-det-lower}.  Bounding that integral by $\|g\|_1\|D^\omega_J\|_\infty^2$ makes $-V_J\le C_V(J)$ uniformly in the coefficients.  The constant $C_V(J)$ is finite because $\cos^2\le1$ in~\eqref{eq:quenched-cov} gives $q^\omega_j\le2q_j$ pointwise, and hence $\|D^\omega_J\|_\infty\le2D_J$.  Thus $e^{-V_J}\le e^{C_V(J)}$, and $Z_J\le e^{C_V(J)}$ follows.  For the lower bound, Jensen's inequality for the convex function $u\mapsto e^{-u}$ gives $Z_J\ge e^{-\E_{\gamma_J}[V_J]}=1$, since the field is distributed as $B_J(x)\sim\cN(0,D^\omega_J(x))$ under $\gamma_J$.
\end{proof}

\subsection{Integrating out one shell}

Define
\begin{equation}\label{eq:T4-def}
    T(B,D,q;\,z) \;:=\;
    H_4(B+z;\,D+q) - H_4(B;\,D)\;,
\end{equation}
which records the change in the Wick-ordered quartic when the shell field $z:=\varphi_j(x)$ with Wick constant $q:=q_j^\omega(x)$ is added to the background $B:=B_{j-1}(x)$ with Wick constant $D:=D^\omega_{j-1}(x)$. By direct expansion,
\begin{equation}\label{eq:T4-expand}
    T(B,D,q;\,z) \;=\;
    H_4(z;q)+4B\,H_3(z;q)+6\,H_2(B;D)\,H_2(z;q)+4\,H_3(B;D)\,H_1(z;q)\;,
\end{equation}
from which it follows that for a centered Gaussian $z\sim\cN(0,q)$ independent of~$B$,
\begin{equation}\label{eq:shell-mean-zero}
    \E_z\big[T(B,D,q;\,z)\big] \;=\; 0\;,
\end{equation}
since every term on the right of \eqref{eq:T4-expand} carries a factor $H_k(z;q)$ with $k\ge1$.

Define the \emph{shell interaction},
\begin{equation}\label{eq:Xj}
    X_j \;:=\; t\int g(x)\,
    T\big(B_{j-1}(x),D^\omega_{j-1}(x),
    q_j^\omega(x);\,\varphi_j(x)\big)\,dx\;,
\end{equation}
so that $V_j = V_{j-1} + X_j$.  The \emph{shell log-Laplace transform} is
\begin{equation}
    \Psi_j \;=\; \Psi_j(B_{j-1})
    \;:=\; \log\E_{a_j}\!\big[e^{-X_j}
    \mid\cF_{j-1},\omega_j\big]\;.
    \label{eq:Psi-def}
\end{equation}

\begin{lemma}[Positivity]\label{lem:positivity}
Conditionally on $\cF_{j-1}$ and $\omega_j$, the shell interaction has mean zero, $\E_{a_j}[X_j\mid\cF_{j-1},\omega_j]=0$, and consequently $\Psi_j\ge0$.
\end{lemma}

\begin{proof}
Conditionally on $\cF_{j-1}$ and $\omega_j$, the background is fixed, and the shell field $\varphi_j(x)$ is a centered Gaussian of variance $q_j^\omega(x)$ (\Cref{def:quenched-cov}).  That variance is exactly the value at which $X_j$ evaluates the Wick constant of $T$, so \eqref{eq:shell-mean-zero} applies pointwise in $x$ with $z=\varphi_j(x)$ and $q=q_j^\omega(x)$:
\begin{equation}
    \E_{a_j}\big[T\big(B_{j-1}(x),\,D^\omega_{j-1}(x),\,q_j^\omega(x);\,\varphi_j(x)\big)
    \mid \cF_{j-1},\omega_j\big] = 0
    \qquad\text{for each }x.
\end{equation}
Integrating against $g(x)\,dx$ gives $\E_{a_j}[X_j\mid\cF_{j-1},\omega_j]=0$.  Jensen's inequality for the convex function $u\mapsto e^{-u}$ then gives $e^{\Psi_j}=\E_{a_j}[e^{-X_j}]\ge e^{-\E_{a_j}[X_j]}=1$.
\end{proof}

The following identity reduces the convergence of the partition function to the analysis of individual shells.

\begin{proposition}[Telescoping]
\label{prop:telescope}
For $1\le j\le\Lambda$,
\begin{equation}\label{eq:telescope-step}
    Z_j \;=\; Z_{j-1}\,\E_{\mu_{j-1}}\big[e^{\Psi_j}\big]\;,
\end{equation}
and consequently
\begin{equation}\label{eq:telescope}
    Z_\Lambda \;=\; Z_0\,\prod_{j=1}^{\Lambda}
    \E_{\mu_{j-1}}\big[e^{\Psi_j}\big]\;.
\end{equation}
\end{proposition}

\begin{proof}
Write $V_j = V_{j-1}+X_j$ and $\gamma_j = \gamma_{j-1}\otimes\gamma_\sigma(a_j)$.  Since $V_{j-1}$ is $\cF_{j-1}$-measurable and hence independent of $a_j$ under the product prior, Fubini gives
\begin{equation*}
    Z_j \;=\; \int e^{-V_{j-1}-X_j}\,d\gamma_j
    \;=\; \int e^{-V_{j-1}}\,
    \bigg(\int e^{-X_j}\,
    d\gamma_\sigma(a_j)\bigg)\,
    d\gamma_{j-1}
    \;=\; Z_{j-1}\;\E_{\mu_{j-1}}\big[e^{\Psi_j}\big]\;,
\end{equation*}
where the inner integral is $e^{\Psi_j}$ by~\eqref{eq:Psi-def}, and the last step is the definition of $\mu_{j-1}$.  Iterating from $j=1$ to $\Lambda$ yields~\eqref{eq:telescope}.
\end{proof}

\section{The sampled covariance}\label{sec:good-geom}

The bounds in \Cref{lem:shell-dim} hold for the shell kernel $Q_j$, while the architecture~\eqref{eq:shell-field} produces the quenched covariance $\Gamma_j^\omega$.  This section shows that, at large enough width, $\Gamma_j^\omega$ is uniformly close to $Q_j$ on $K\times K$ for all but finitely many shells.

\begin{theorem}[Concentration]\label{thm:good-geom}
There exist constants $c_0>0$, depending only on $d$, and $\kappa>0$, depending in addition on $\chi,g$ (neither depends on $j$ or $N_j$), such that the following holds. If
\begin{equation}\label{eq:width-gen}
    N_j \;\ge\; c_0\,(j{+}1)\,2^{2dj},
\end{equation}
then almost surely in $\omega$ there is a finite random index $j_0=j_0(\omega)\geq1$ such that for all $j\ge j_0$:
\begin{align}
    \sup_{x,y\in K}
    \big|\Gamma_j^\omega(x,y)-Q_j(x-y)\big|
    &\;\le\; 2^{-2j}\;,
    \label{eq:gg-approx}\\
    A_j^\omega \;:=\;
    \sup_{x\in K}\int g(y)\,
    |\Gamma_j^\omega(x,y)|\,dy
    &\;\le\; \kappa\,2^{-2j}\;.
    \label{eq:gg-off}
\end{align}
\end{theorem}

\begin{proof}
Fix $(x,y)\in K\times K$ and $j\geq1$. The quenched covariance~\eqref{eq:quenched-cov} is a sample mean of $N_j$ independent, identically distributed, bounded random variables:
\begin{equation}\label{eq:Gamma-sample}
    \Gamma_j^\omega(x,y)
    \;=\; \frac{1}{N_j}\sum_{\ell=1}^{N_j} Z_\ell(x,y),
    \qquad
    Z_\ell(x,y) \;:=\; 2q_j\,
    \cos(b_{j,\ell}\cdot x+c_{j,\ell})\,
    \cos(b_{j,\ell}\cdot y+c_{j,\ell})\;.
\end{equation}
Each summand satisfies $|Z_\ell|\le 2q_j\le 2C\,2^{(d-2)j}$, and its expectation is $Q_j(x-y)$ by \Cref{lem:cov-mean}. By Hoeffding's inequality for bounded i.i.d.\ random variables, for each $\epsilon>0$:
\begin{equation}\label{eq:hoeffding}
    \Pr\!\big(|\Gamma_j^\omega(x,y)-Q_j(x-y)|>\epsilon\big)
    \;\le\; 2\exp\!\bigg(
    -\frac{N_j\,\epsilon^2}{8\,q_j^2}\bigg)
    \;\le\; 2\exp\!\bigg(
    -\frac{C\,N_j\,\epsilon^2}{2^{2(d-2)j}}\bigg)\;.
\end{equation}
The constant $C>0$ depends only on $d$ because the upper bound on $q_j$ in \eqref{eq:q_j-upper} rests only on $\chi\le1$, on the volume of the fixed annulus, and on $p^2+m^2\ge p^2$ in~\eqref{eq:Sj-bound}.

We upgrade the pointwise bound to a uniform bound over $K\times K$.  Differentiating either kernel with respect to $x$ or $y$ brings down a frequency bounded by $|b_{j,\ell}|\le 2^{j+1}$ for the cosines of~\eqref{eq:Gamma-sample}, and by $|p|\le 2^{j+1}$ on the support of $S_j$ for the integral \eqref{eq:Qj-kernel}.  Hence $\Gamma_j^\omega(x,y)$ and $Q_j(x-y)$, both viewed as functions of $(x,y)\in K\times K$, are Lipschitz in each spatial argument with constant
\begin{equation}\label{eq:lip}
    L \;=\; O(2^j\cdot q_j) \;=\; O(2^{(d-1)j})\;.
\end{equation}
Define $f:=\Gamma_j^\omega-Q_j$, which is Lipschitz in each argument with constant $2L$.  If $(x',y')$ is within Euclidean distance $\delta$ of $(x,y)$ in $\R^{2d}$, then, moving one argument at a time,
\begin{equation}\label{eq:lip_estimate}
    |f(x,y)| \;\le\; |f(x',y')|
    + 2L\,\big(|x-x'|+|y-y'|\big)
    \;\le\; |f(x',y')| + 4L\delta\;.
\end{equation}
Set
\begin{equation}\label{eq:delta-choice}
    \delta \;:=\; \frac{2^{-2j}}{8L}\;,
\end{equation}
and choose a finite set of points in $K\times K$ whose $\delta$-balls cover it (a \emph{$\delta$-net}). If $|f|\le 2^{-2j}/2$ at every net point, then the Lipschitz estimate \eqref{eq:lip_estimate} gives $|f|\le 2^{-2j}/2+4L\delta=2^{-2j}$ everywhere on $K\times K$.

It remains to establish the bound at each net point. Since $K\times K$ is a bounded subset of $\R^{2d}$, such a net has $O(\delta^{-2d})=O(2^{2d(d+1)j})$ points.  Applying the Hoeffding bound~\eqref{eq:hoeffding} at each net point with $\epsilon=2^{-2j}/2$ and taking a union bound, the probability that the uniform bound fails is at most
\begin{equation}\label{eq:union}
    O\big(2^{2d(d+1)j}\big)\cdot
    \exp\!\big(-C\,N_j\cdot 2^{-2dj}\big)\;.
\end{equation}
With the width requirement~\eqref{eq:width-gen}, the bound~\eqref{eq:union} becomes
\begin{equation}
    O\big(2^{2d(d+1)j}\big)\cdot e^{-C\,c_0\,j}
    \;=\; O\big(\exp(C_{d}\,j - C\,c_0\,j)\big)\;,
\end{equation}
where $C_{d}:=2d(d{+}1)\ln 2$.  For $c_0$ large enough that $C\,c_0>C_d+1$, with $C$ the constant in~\eqref{eq:union}, this is $O(e^{-j})$ for all $j\ge1$. Since $\sum_{j\ge 1}e^{-j}<\infty$, by the Borel--Cantelli lemma there exists a.s.\ a finite $j_0$ such that \eqref{eq:gg-approx} holds for all $j\geq j_0$.

For the row integral bound, \eqref{eq:gg-approx} gives $|\Gamma_j^\omega(x,y)|\le|Q_j(x-y)|+2^{-2j}$ for $j\ge j_0$.  Integrating against $g(y)$ and taking the supremum over $x\in K$,
\begin{equation*}
    A_j^\omega
    \;\le\; A_j+2^{-2j}\,\|g\|_1
    \;\le\; \kappa\,2^{-2j}\;,
\end{equation*}
with $\kappa:=C+\|g\|_1$, where $C$ is the constant in $A_j\le C\,2^{-2j}$ (\Cref{lem:shell-dim}(d)).
\end{proof}

The interchange of limits in \Cref{thm:interchange} uses \Cref{thm:good-geom} for a family of networks of increasing width, with one threshold valid for every network in the family.  We fix one infinite feature sequence per shell and obtain each finite-width network by truncating it.  Widening a network then adds features without redrawing the ones it already has.

\begin{definition}[Nested width family]\label{def:nested}
For each shell $j$, fix a single infinite sequence $\omega_j^\infty=\big((b_{j,\ell},c_{j,\ell})\big)_{\ell\ge1}$ drawn i.i.d.\ from $\pi_j\otimes\mathrm{Unif}([-\pi,\pi])$ as in \eqref{eq:features}, all shells independent, and write $\omega^\infty:=(\omega_j^\infty)_{j\ge0}$.  For $n\in\N$ set
\begin{equation}\label{eq:width-schedule}
    N_j^{(n)} \;:=\; \big\lceil n\,c_0\,(j{+}1)\,2^{2dj}\big\rceil\;,
\end{equation}
with $c_0$ as in \Cref{thm:good-geom}, and let $\omega^{(n)}$ denote the truncation of $\omega^\infty$ in which shell $j$ uses the first $N_j^{(n)}$ features of $\omega_j^\infty$.  We write $\Gamma_j^{\omega^{(n)}}$, $q_j^{\omega^{(n)}}$, $D_\Lambda^{\omega^{(n)}}$, $\phi_\Lambda^{(n)}$, $\gamma_\Lambda^{(n)}$, $Z_\Lambda^{(n)}$, $V_\Lambda^{(n)}$, $\mu_\Lambda^{(n)}$ for objects built from $\omega^{(n)}$.
\end{definition}

\begin{lemma}[Concentration, uniformly over the width family]
\label{lem:unif-gg}
Almost surely in $\omega^\infty$ there is a finite index $j_0^*=j_0^*(\omega^\infty)\geq1$ such that \eqref{eq:gg-approx} and \eqref{eq:gg-off} hold for every $j\ge j_0^*$ and every $n\ge1$ simultaneously, with the same constant $\kappa$.
\end{lemma}

\begin{proof}
Call the pair $(j,n)$ a \emph{failure} if \eqref{eq:gg-approx} is violated at shell $j$ for the network of width $N_j^{(n)}$.  It is enough to control failures because the proof of \Cref{thm:good-geom} derives \eqref{eq:gg-off} from \eqref{eq:gg-approx}.  By the Borel--Cantelli lemma, it suffices that
\begin{equation}\label{eq:failure-sum}
    \sum_{j\ge1}\sum_{n\ge1}\Pr\big((j,n)\text{ is a failure}\big)\;<\;\infty\;,
\end{equation}
for then almost surely only finitely many pairs fail, and any $j_0^*$ exceeding every shell index in that finite set is as claimed.

The proof of \Cref{thm:good-geom} bounds the failure probability of shell $j$ at width $N$ by the cardinality of the $\delta$-net times the Hoeffding tail, and by the choice of $c_0$ there this is at most $O\big(2^{2d(d+1)j}\big)\exp\big(-(C_{d}+1)\,N\,2^{-2dj}/c_0\big)$.  From the width schedule~\eqref{eq:width-schedule},
\begin{equation}\label{eq:joint-BC}
    \Pr\big((j,n)\text{ is a failure}\big)
    \;\le\; C\,e^{C_{d}j}\,e^{-(C_{d}+1)\,n\,(j+1)}\;.
\end{equation}
At a fixed shell $j$, the bound~\eqref{eq:joint-BC} is geometric in $n$ with ratio $e^{-(C_{d}+1)(j+1)}\le e^{-1}$, and its sum over $n$ is therefore at most twice the $n=1$ term, $2C\,e^{-(C_{d}+1)}e^{-j}$.  The remaining sum over $j$ is geometric with ratio $e^{-1}$. Hence~\eqref{eq:failure-sum} holds.

Finally, the constant $\kappa$ of \Cref{thm:good-geom} depends on neither the realization nor the width, so the same value holds for every $n$.
\end{proof}

\Cref{lem:unif-gg} bounds $|\Gamma_j^{\omega^{(n)}}-Q_j|$ by $2^{-2j}$ uniformly in $n$, but the bound does not decrease with $n$ and says nothing about points outside $K$.  Convergence as $n\to\infty$ follows instead from the nesting in \Cref{def:nested} and the strong law of large numbers.

\begin{proposition}[Kernel convergence]\label{prop:kernel-conv}
Almost surely, for every shell $j$ and every compact set $K'\subset\R^d$,
\begin{equation}\label{eq:kernel-lim}
    \sup_{x,y\in K'}
    \big|\Gamma_j^{\omega^{(n)}}(x,y)-Q_j(x-y)\big|
    \;\longrightarrow\; 0
    \qquad(n\to\infty)\;,
\end{equation}
and consequently $\sup_{K'}|q_j^{\omega^{(n)}}-q_j|\to0$ and $\sup_{K'}|D_\Lambda^{\omega^{(n)}}-D_\Lambda|\to0$ for every $\Lambda$.
\end{proposition}

\begin{proof}
Fix $j$ and a pair $(x,y)\in\R^d\times\R^d$.  By \Cref{def:nested}, $\Gamma_j^{\omega^{(n)}}(x,y)$ is the average of the first $N_j^{(n)}$ terms of the sequence $\big(Z_\ell(x,y)\big)_{\ell\ge1}$ of~\eqref{eq:Gamma-sample}, and this sequence does not depend on $n$.  Its terms are i.i.d.\ and bounded, with mean $Q_j(x-y)$ by \Cref{lem:cov-mean}.  Because $N_j^{(n)}\to\infty$, the strong law of large numbers gives
\begin{equation}\label{eq:pointwise-lim}
    \Gamma_j^{\omega^{(n)}}(x,y)\;\longrightarrow\;Q_j(x-y)
    \qquad\text{almost surely.}
\end{equation}
The event on which~\eqref{eq:pointwise-lim} holds depends on $(x,y)$, and there are uncountably many pairs.  On a bounded set, a $\delta$-net reduces them to finitely many, as in the proof of \Cref{thm:good-geom}. Write $f_n:=\Gamma_j^{\omega^{(n)}}-Q_j$, and let $K_R\subset\R^d$ be the closed ball of radius $R$ about the origin.  Both kernels are Lipschitz in each argument on all of $\R^d$ with the $n$-independent constant $L=O(2^{(d-1)j})$ of~\eqref{eq:lip}, so~\eqref{eq:lip_estimate} holds for every $f_n$ with the same $L$.  Fix $\epsilon>0$ and $R\in\N$, let $P$ be a finite $\delta$-net of $K_R\times K_R$ with $\delta:=\epsilon/(8L)$, and let $E_{j,\epsilon,R}$ be the event that~\eqref{eq:pointwise-lim} holds at every point of $P$.  As a finite intersection of probability-one events, $E_{j,\epsilon,R}$ has probability one.  On $E_{j,\epsilon,R}$ there is an $n_0$ such that $|f_n|\le\epsilon/2$ on $P$ for all $n\ge n_0$, and~\eqref{eq:lip_estimate} then gives $\sup_{K_R\times K_R}|f_n|\le\epsilon/2+4L\delta=\epsilon$ for these $n$.

The events $E_{j,1/k,R}$ with $j\ge0$ and $k,R\in\N$ are countably many, so their intersection $E$ has probability one.  On $E$, for every shell $j$ and every $R$, the supremum of $|f_n|$ over $K_R\times K_R$ is eventually at most $1/k$ for each $k$ and therefore tends to zero.  Every compact $K'$ lies in some $K_R$, which gives~\eqref{eq:kernel-lim}.
\end{proof}

\begin{notation}\label{not:quench}
From this point on, we fix a network whose widths satisfy~\eqref{eq:width-gen} and condition on a realization of its features $\omega=(\omega_j)_{j\ge 0}$ satisfying~\eqref{eq:gg-approx} and~\eqref{eq:gg-off} for all $j\ge j_0$.  Each network of \Cref{def:nested} is such a network, and \Cref{prop:uniform-N} applies \Cref{sec:d2-setup,sec:d2-converge} to the entire family.  In expectations and probabilities, we suppress the conditioning on the features and write, e.g.,  $\E_{a_j}[\,\cdot\,]$ for $\E_{a_j}[\,\cdot\mid\omega,\cF_{j-1}]$.
\end{notation}

\section{One-shell analysis}\label{sec:one-shell}

\subsection{The tilted measure}
\label{sec:coercivity}

The tilted measure, the Gaussian prior on the shell coefficients re\-weight\-ed by $e^{-X_j}$, exists only if its normalization $\E_{a_j}[e^{-X_j}]=e^{\Psi_j}$ is finite and strictly positive for every background field. Finiteness comes from \emph{coercivity}: the $+z^4$ term of $T$ dominates all lower powers of $z$, including its sign-indefinite quadratic term.

\begin{lemma}[Quartic coercivity]\label{lem:T-coercivity}
Let $C\ge0$, $u\in\mathbb R$.  The Wick-ordered quartic obeys
\begin{equation}\label{eq:H4-coerce}
    H_4(u;C)\;\ge\;\tfrac12\,u^4 - 15\,C^2\;.
\end{equation}
Consequently, for all $B,z\in\R$, $D\ge0$, and $0\le q\le 2q_j$, the function $T$ of~\eqref{eq:T4-def} satisfies
\begin{equation}\label{eq:T-coercivity}
    T(B,D,q;\,z)\;\ge\;\tfrac{1}{16}\,z^4
    - c_{\mathrm{coer}}\,(1+B^4+D^2)\;,
\end{equation}
where $c_{\mathrm{coer}}$ depends only on $q_j$.
\end{lemma}

\begin{proof}
For~\eqref{eq:H4-coerce}, apply Young's inequality $6Cu^2\le\tfrac12 u^4+18C^2$ to the middle term of $H_4(u;C)=u^4-6Cu^2+3C^2$.

For \eqref{eq:T-coercivity}, the first term of \eqref{eq:T4-def} is bounded below by \eqref{eq:H4-coerce} with $u=B+z$ and $C=D+q$, and the second by $H_4(B;D)\le B^4+3D^2$, obtained by dropping its nonpositive middle term.  We also use the elementary inequality
\begin{equation}\label{eq:power-sum}
    (x_1+\cdots+x_n)^p \;\le\; n^{p-1}\big(x_1^p+\cdots+x_n^p\big)\;,
    \qquad x_i\ge0,\ p\ge1\;,
\end{equation}
which is Jensen's inequality for the convex map $s\mapsto s^p$.  At $n=2$ and $p=4$, applied to $|z|\le|B+z|+|B|$, it gives $(B+z)^4\ge\tfrac18z^4-B^4$.  At $p=2$ it gives $(D+q)^2\le2D^2+8q_j^2$, since $q\le2q_j$.  Combining these bounds,
\begin{equation}
    T(B,D,q;z)\;\ge\;\tfrac{1}{16}z^4
    - \tfrac32 B^4 - 33D^2 - 120q_j^2\;,
\end{equation}
which is~\eqref{eq:T-coercivity} with $c_{\mathrm{coer}}:=\max\{33,\,120q_j^2\}$.
\end{proof}

\begin{lemma}[Lower bound on the shell interaction]
\label{lem:X-lower}
For each $j\ge1$ set
\begin{equation}\label{eq:CB-def}
    C_B \;:=\; t\,c_{\mathrm{coer}}\int g(x)\,
    \Big(1+B_{j-1}(x)^4+\big(D^\omega_{j-1}(x)\big)^2\Big)\,dx\;,
\end{equation}
with $c_{\mathrm{coer}}$ as in \Cref{lem:T-coercivity}.  Then $C_B$ is finite and $\cF_{j-1}$-measurable, and
\begin{equation}\label{eq:X-bound}
    X_j \;\ge\; \frac{t}{16}\int g\,\varphi_j^4\,dx
    - C_B \;\ge\; -C_B
\end{equation}
for every $a_j\in\R^{N_j}$.
\end{lemma}

\begin{proof}
At each $x\in K$, write $B$, $D$, and $z$ for $B_{j-1}(x)$, $D_{j-1}^\omega(x)$, and $\varphi_j(x)$, and note that $q_j^\omega(x)\le 2q_j$.  Then $T$ obeys the coercivity bound~\eqref{eq:T-coercivity}.  Multiplying by $t$ and integrating against $g$ gives the first inequality of~\eqref{eq:X-bound}, and the second follows by discarding the quartic term, which is nonnegative.

The constant $C_B$ is $\cF_{j-1}$-measurable because $B_{j-1}$ and $D^\omega_{j-1}$ are.  It is finite because both are continuous and hence bounded on the compact set $K$.
\end{proof}

\begin{lemma}[Finiteness of the shell log-Laplace transform]
\label{lem:Psi-finite}
For each $j\ge1$, the shell log-Laplace transform~\eqref{eq:Psi-def} satisfies $0\le\Psi_j\le C_B$.  Moreover, $\E_{a_j}\big[\,|P(a_j)|\,e^{-X_j}\big]<\infty$ for every polynomial~$P$.
\end{lemma}

\begin{proof}
The lower bound is \Cref{lem:positivity}.  By \Cref{lem:X-lower}, $e^{-X_j}\le e^{C_B}$ for every $a_j$, and $C_B$ does not depend on $a_j$.  Hence $\E_{a_j}[e^{-X_j}]\le e^{C_B}$, which is the upper bound, and $\E_{a_j}[\,|P(a_j)|\,e^{-X_j}]\le e^{C_B}\,\E_{a_j}[\,|P(a_j)|\,]$, which is finite because the Gaussian prior has moments of all orders.
\end{proof}

By \Cref{lem:Psi-finite}, the normalization $\E_{a_j}[e^{-X_j}]=e^{\Psi_j}$ is finite and strictly positive.  Set
\begin{equation}\label{eq:W-def}
    W_j(a_j) \;:=\;
    \frac{|a_j|^2}{2\sigma^2} + X_j(a_j)\;.
\end{equation}
The \emph{tilted measure} is the probability measure
\begin{equation}\label{eq:tilted-def}
    d\rho_j(a_j)
    \;\propto\;
    e^{-W_j(a_j)}\,da_j
    \;\propto\;
    e^{-X_j(a_j)}\,d\gamma_\sigma(a_j)\;.
\end{equation}
We write $\langle F\rangle :=\E_{\rho_j}[F]$ for expectations under this measure.

\subsection{Uniform log-concavity}
\label{sec:schur}

Several of the estimates that follow need the centered moments of the shell field under the tilted measure $\rho_j$ to be bounded at the free-field scale $q^\omega_j$.  \Cref{cor:tilted-moments} derives that bound from the uniform convexity of $W_j$, equivalently the uniform log-concavity of $\rho_j$, which this subsection establishes.

Because $\varphi_j(x)=\sum_\ell a_{j,\ell}\,U_{j,\ell}(x)$ is linear in the coefficients, for any smooth $h$,
\begin{equation*}
    \nabla^2_{a_j}\,h\big(\varphi_j(x)\big)
    \;=\; h''\big(\varphi_j(x)\big)\;U_j(x)\,U_j(x)^T \;.
\end{equation*}
The Hessian of $W_j$ is therefore $\sigma^{-2}I$ plus the $g$-weighted integral of terms of this form, each a scalar prefactor times the rank-one matrix $U_j(x)\,U_j(x)^T$.  A lower bound on the prefactor, uniform in $x$, comes out of the integral and leaves the \emph{Gram matrix} of the shell's neurons in $L^2(g\,dx)$,
\begin{equation}\label{eq:Sigma-def}
    \Sigma_j \;:=\; \int g(x)\,U_j(x)\,U_j(x)^T\,dx
    \;\in\;\R^{N_j\times N_j},
    \qquad
    (\Sigma_j)_{\ell m}
    \;=\; \langle U_{j,\ell},\,U_{j,m}
    \rangle_{L^2(g\,dx)}\,.
\end{equation}
\Cref{lem:Schur} bounds $\|\Sigma_j\|_{\mathrm{op}}$ by the row integral of the quenched shell kernel, and \Cref{prop:hessian} combines this with a lower bound on the prefactor.

\begin{lemma}[Operator norm of the Gram matrix]
\label{lem:Schur}
For all $j\ge j_0$ (\Cref{thm:good-geom}),
\begin{equation}\label{eq:Schur}
    \|\Sigma_j\|_{\mathrm{op}}
    \;\le\; \frac{A_j^\omega}{\sigma^2}
    \;\le\; \frac{\kappa\,2^{-2j}}{\sigma^2}\,.
\end{equation}
The first inequality holds for every realization of the features, and only the second uses~\eqref{eq:gg-off}.
\end{lemma}

\begin{proof}
The matrix $\Sigma_j$ is symmetric and positive semidefinite, since $v^T\Sigma_j v=\int g\,(U_j(x)^T\,v)^2\,dx\ge0$, so its operator norm is its largest eigenvalue $\lambda$.  Let $v$ be a corresponding eigenvector, and define the associated field $f(x):=U_j(x)^T\,v$.  Pairing the eigenvalue equation $\Sigma_j v=\lambda v$ with $U_j(x)$ and using $U_j(x)^T\,U_j(y)=\Gamma^\omega_j(x,y)/\sigma^2$ gives
\begin{equation}\label{eq:schur-eigen}
    \lambda\,f(x)
    \;=\; U_j(x)^T\,\Sigma_jv
    \;=\; \frac{1}{\sigma^2}\int g(y)\,\Gamma^\omega_j(x,y)\,f(y)\,dy\;.
\end{equation}
If $f$ vanishes identically on $K$, then $\lambda|v|^2=\int g\,f^2\,dx=0$, and the bound is trivial. Otherwise, $|f|$ is continuous on the compact set $K$ and attains its maximum at some $x_*\in K$ with $|f(x_*)|>0$.  Evaluating \eqref{eq:schur-eigen} at $x_*$ and bounding $|f(y)|\le|f(x_*)|$ under the integral,
\begin{equation*}
    \lambda\,|f(x_*)|
    \;\le\; \frac{|f(x_*)|}{\sigma^2}
    \sup_{x\in K}\int g(y)\,|\Gamma^\omega_j(x,y)|\,dy
    \;=\; |f(x_*)|\,\frac{A_j^\omega}{\sigma^2}\;.
\end{equation*}
Dividing by $|f(x_*)|$ gives $\|\Sigma_j\|_{\mathrm{op}}=\lambda\le A_j^\omega/\sigma^2$.

This argument is the continuum form of the fact that no eigenvalue of a matrix exceeds its largest absolute row sum.  For a symmetric kernel, it gives the same bound as the Schur test \cite[Thm.~5.2]{HalmosSunder1978}.
\end{proof}

\begin{proposition}[One-shell Hessian]
\label{prop:hessian}
Set $\mathcal{D}_j:=\sup_{x\in K}D^\omega_j(x)$.  For every $j\ge j_0$,
\begin{equation}\label{eq:hessian-combined}
    \nabla^2_{a_j}W_j
    \;\ge\; \frac{1}{\sigma^2}\Big(1
    - \underbrace{12\,t\,\kappa\,\mathcal{D}_j
    \cdot 2^{-2j}}_{\text{one-shell perturbation}}
    \Big)\,I_{N_j}\;.
\end{equation}
The one-shell perturbation depends on neither $\Lambda$ nor the coefficients $a_0,\dots,a_j$, and with constants free of the widths $(N_j)$,
\begin{equation}\label{eq:quartic-pert-size}
    12\,t\,\kappa\,\mathcal{D}_j\,2^{-2j}
    \;=\;
    \begin{cases}
    O\big(t\,2^{-2j}j\big) & d=2\;,\\[4pt]
    O\big(t\,2^{(d-4)j}\big) & d\ge3\;.
    \end{cases}
\end{equation}
\end{proposition}

\begin{proof}
Differentiating~\eqref{eq:W-def} twice in $a_j$:
\begin{equation}\label{eq:hessian}
    \nabla^2_{a_j}W_j
    \;=\; \frac{1}{\sigma^2}\,I_{N_j}
    + t\int g(x)\,
    T_{zz}\,U_j(x)\,U_j(x)^T\,dx\;,
\end{equation}
where $T_{zz}$ is the second $z$-derivative of $T$ \eqref{eq:T4-def}, evaluated at $B=B_{j-1}(x)$, $D=D^\omega_{j-1}(x)$, $q=q^\omega_j(x)$, and $z=\varphi_j(x)$. Direct differentiation gives
\begin{equation}\label{eq:Tzz4}
    T_{zz} \;=\; 12(B+z)^2 - 12(D+q)\;.
\end{equation}
The first term of~\eqref{eq:Tzz4} is nonnegative, and the second is at least $-12\big(D^\omega_{j-1}(x)+q^\omega_j(x)\big)=-12\,D^\omega_j(x)\ge-12\,\mathcal{D}_j$.  The quartic contribution to the Hessian is therefore bounded below by $-12\,t\,\mathcal{D}_j\,\Sigma_j$. Using \Cref{lem:Schur},
\begin{align}
    \nabla^2_{a_j}W_j
    &\;\ge\; \frac{1}{\sigma^2}\,I
    - 12\,t\,\mathcal{D}_j\,\|\Sigma_j\|_{\mathrm{op}} I
    \notag\\
    &\;\ge\; \frac{1}{\sigma^2}\Big(1
    - 12\,t\,\kappa\,\mathcal{D}_j\cdot 2^{-2j}
    \Big)\,I\;.
\end{align}

Since $q^\omega_j\le 2q_j$ deterministically, $\mathcal{D}_j\le 2D_j$, so \Cref{lem:shell-dim}(a),(e) give a constant $c_{\mathcal{D}}$, depending only on $d$, $\chi$, and $m$, with
\begin{equation}\label{eq:Dcal-bound}
    \mathcal{D}_j \;\le\;
    \begin{cases}
    c_{\mathcal{D}}\,(j{+}1) & d=2\;,\\[4pt]
    c_{\mathcal{D}}\,2^{(d-2)j} & d\ge3\;,
    \end{cases}
\end{equation}
and multiplying by $12\,t\,\kappa\,2^{-2j}$ gives~\eqref{eq:quartic-pert-size}.
\end{proof}

\begin{corollary}[Uniform convexity]
\label{cor:BL}
Let $2\le d<4$.  With no restriction on the coupling, there is a threshold $j_1$ with $j_0\le j_1<\infty$ such that for every $j\ge j_1$,
\begin{equation}\label{eq:BL-hessian}
    \nabla^2_{a_j}W_j
    \;\ge\; \frac{1}{2\sigma^2}\,I_{N_j}\;.
\end{equation}
The threshold depends only on $d$, $\chi$, $m$, $t$, $\kappa$, and the threshold $j_0$ of \Cref{thm:good-geom}, and in particular not on the widths $(N_j)$, on $\Lambda$, or on the coefficients $a_0,\dots,a_j$.
\end{corollary}

\begin{proof}
By~\eqref{eq:quartic-pert-size} the one-shell perturbation is at most $12\,t\,\kappa c_{\mathcal{D}}(j{+}1)2^{-2j}$ at $d=2$ and $12\,t\,\kappa c_{\mathcal{D}}2^{-j}$ at $d=3$.  Both bounds are free of the widths and tend to zero, so choose $j_1\ge j_0$ beyond which the perturbation is at most $1/2$.  \Cref{prop:hessian} then gives $\nabla^2_{a_j}W_j\ge\sigma^{-2}(1-\tfrac12)I=\tfrac{1}{2\sigma^2}I$.
\end{proof}

\subsection{Moments of the tilted measure}\label{sec:tilted-moments}

\begin{corollary}[Centered tilted moments]
\label{cor:tilted-moments}
Let $2\le d<4$.  Under $\rho_j$, write $\bar{\varphi}_j(x) :=\langle\varphi_j(x)\rangle$ for the tilted mean and $\tilde{\varphi}_j(x) :=\varphi_j(x)-\bar{\varphi}_j(x)$ for the centered part.  Then
\begin{equation}\label{eq:centered-moments}
    \big\langle\tilde{\varphi}_j(x)^{2k}
    \big\rangle
    \;\le\; \big(8q_j^\omega(x)\big)^k\,k!\,
\end{equation}
for all $j\ge j_1$, $k\ge 1$, and $x\in K$.
\end{corollary}

\begin{proof}
By the Hessian bound~\eqref{eq:BL-hessian} and the Bakry--\'Emery criterion (\cite{BakryEmery1985}; \cite[Theorem~5.2]{ledoux2001concentration}), $\rho_j$ satisfies the logarithmic Sobolev inequality
\begin{equation}\label{eq:LSI}
    \mathrm{Ent}_{\rho_j}(h^2)
    \;\le\; 4\sigma^2\int|\nabla h|^2\,
    d\rho_j\;,
\end{equation}
for every smooth $h:\mathbb R^{N_j}\to\mathbb R$, where $\mathrm{Ent}_\mu(h):=\int h\log h\,d\mu -\big(\!\int h\,d\mu\big)\log\!\big(\!\int h\,d\mu\big)$.

The shell field $\varphi_j(x)=U_j(x)^T\,a_j$ is a linear functional of $a_j$ with gradient $U_j(x)$, hence Lipschitz constant
\begin{equation}
    L:=|U_j(x)|\;,
    \qquad
    L^2=|U_j(x)|^2
    =\frac{1}{\sigma^2}\,q_j^\omega(x)\;.
\end{equation}
If $L=0$ (that is, $q_j^\omega(x)=0$) then $\varphi_j(x)=0$ almost surely, and the centered moments vanish. Assume henceforth $L>0$.  Applying the Herbst argument \cite[Theorem~5.3]{ledoux2001concentration} to the $1$-Lipschitz functionals $\pm\varphi_j(x)/L$ and taking a union bound over the two tails gives, for every $s\ge0$,
\begin{equation}\label{eq:herbst-tail}
    \rho_j\big(
    |\varphi_j(x)-\bar\varphi_j(x)|>s\big)
    \;\le\; 2\,e^{-s^2/(2\tau_j^2)}\;,
    \qquad
    \tau_j^2:=2\sigma^2L^2=2\,q_j^\omega(x)\;.
\end{equation}
With the tail bound~\eqref{eq:herbst-tail}, the layer-cake formula gives, for $k\ge1$,
\begin{equation}
    \big\langle\tilde\varphi_j(x)^{2k}\big\rangle
    = 2k\int_0^\infty \,s^{2k-1}\,
      \rho_j\big(|\tilde\varphi_j(x)|>s\big)\,ds
    \le 4k\!\int_0^\infty\! s^{2k-1}e^{-s^2/(2\tau_j^2)}\,ds
    = 2^{k+1}k!\,\tau_j^{2k}\;.
\end{equation}
Inserting $\tau_j^2=2\,q_j^\omega(x)$ gives $2^{k+1}k!\,\tau_j^{2k}=2\cdot4^k\,k!\,q^\omega_j(x)^k$, and $2\cdot4^k\le8^k$ for $k\ge1$ yields~\eqref{eq:centered-moments}.
\end{proof}

\begin{lemma}[Mean identity]\label{lem:mean-exact}
For every $j\ge1$,
\begin{equation}\label{eq:mean-exact}
    \bar{\varphi}_j(x)
    = -4t\int g(y)\,\Gamma_j^\omega(x,y)\,
    \big\langle H_3\big(B_j(y);D^\omega_j(y)\big)
    \big\rangle\,dy \;.
\end{equation}
\end{lemma}

\begin{proof}
Write $\partial_\ell:=\partial/\partial a_{j,\ell}$.  Since $X_j\ge-C_B$ by~\eqref{eq:X-bound}, $e^{-W_j}\le e^{C_B}e^{-|a_j|^2/(2\sigma^2)}$, and $\partial_\ell W_j$ is a polynomial in $a_j$.  Integration by parts therefore gives $\int\partial_\ell(e^{-W_j})\,da_j=0$, that is, $\langle\partial_\ell W_j\rangle=0$.
The gradient is
\begin{equation*}
    \partial_\ell W_j=\frac{a_{j,\ell}}{\sigma^2}
    +4t\int g(y)\,H_3\big(B_j(y);D^\omega_j(y)\big)\,U_{j,\ell}(y)\,dy\;,
\end{equation*}
so averaging and solving for $\bar a_{j,\ell}:=\langle a_{j,\ell} \rangle$ gives
\begin{equation*}
    \bar a_{j,\ell}
    =-4t\,\sigma^2\int g(y)\,\big\langle
    H_3\big(B_j(y);D^\omega_j(y)\big)\big\rangle\,
    U_{j,\ell}(y)\,dy\;.
\end{equation*}
The average passes inside the $y$-integral by Fubini, because on $K$ the cubic $H_3\big(B_j(y);D^\omega_j(y)\big)$ is bounded by a polynomial in $|a_j|$ and $\rho_j$ has moments of all orders (\Cref{lem:Psi-finite}).  Since $\varphi_j(x)$ is linear in the coefficients, $\bar\varphi_j(x)=\sum_\ell U_{j,\ell}(x)\,\bar a_{j,\ell}$.  Substituting $\bar a_{j,\ell}$ and $\sigma^2\sum_\ell U_{j,\ell}(x)\,U_{j,\ell}(y)=\Gamma_j^\omega(x,y)$ yields~\eqref{eq:mean-exact}.
\end{proof}

\section{The log-Laplace bound in two dimensions}\label{sec:d2-setup}

We now specialize to $d=2$ and bound the shell log-Laplace transform on backgrounds that are not too large (\Cref{thm:d2-one-shell}).

\subsection{Two dimensions and the interpolation}\label{sec:d2-special}

In two dimensions, \Cref{lem:shell-dim} gives $q_j\asymp1$ and $D_J\asymp J$.  We write $q_{\max}:=\sup_jq_j$, which is finite because $q_j\asymp1$.  Since $q_j^\omega(x)\le2q_j$ pointwise, we fix the constant
\begin{equation}\label{eq:cq-def}
    c_q \;:=\; \max\big\{1,\ 2q_{\max}\big\}\;,
\end{equation}
so that $q^\omega_j\le c_q$ and $D^\omega_J\le c_q(J{+}1)$ at every shell and for every feature realization.  The floor of $1$ lets later steps use $c_q\ge1$.  Consequently, the constant of \Cref{lem:V-lower} obeys $C_V(J)\le c_V(J{+}1)^2$, where
\begin{equation}\label{eq:d2-C1}
    c_V \;:=\; 6\,t\,\|g\|_1\,c_q^2\;,
\end{equation}
so that $e^{-V_J}\le e^{c_V(J+1)^2}$ pointwise.

The proof of the bound on $\Psi_j$ interpolates between the Gaussian prior and the tilted measure.  For $\alpha\in[0,1]$, define the interpolated shell field, shell interaction, and log-Laplace transform:
\begin{align}
    \varphi_{j,\alpha}(x)&\;:=\;\sqrt{\alpha}\varphi_j(x)\;,\\
    X_{j,\alpha} &\;:=\; t\int g\,
    T\big(B_{j-1},D^\omega_{j-1},\alpha\,q^\omega_j;\,
    \sqrt{\alpha}\,\varphi_j\big)\,dx\;,
    \label{eq:Xj-alpha}\\
    \Psi_j(\alpha) &\;:=\;
    \log\E_{a_j}\big[e^{-X_{j,\alpha}}\mid\cF_{j-1},
    \omega_j\big]\;.
    \label{eq:Psi-alpha}
\end{align}
Then $X_{j,0}=0$, $X_{j,1}=X_j$, $\Psi_j(0)=0$, and $\Psi_j(1)=\Psi_j$.  We write $\rho_{j,\alpha}$ and $\langle\cdot\rangle_\alpha$ for the tilted measure and expectation built from $X_{j,\alpha}$, and $\bar\varphi_{j,\alpha}$ and $\tilde\varphi_{j,\alpha}$ for the tilted mean and centered part of $\varphi_{j,\alpha}$.  Multiplying every neuron $U_{j,\ell}$ of shell $j$ by $\sqrt\alpha$ leaves the lower shells unchanged and replaces $\varphi_j$, $\Gamma^\omega_j$, and $q^\omega_j$ by $\varphi_{j,\alpha}$, $\alpha\,\Gamma^\omega_j$, and $\alpha\,q^\omega_j$, and hence $X_j$ by $X_{j,\alpha}$.  The proofs of \Cref{sec:one-shell} use the neurons of shell $j$ only through upper bounds on $q^\omega_j$ and $A^\omega_j$, and the rescaling only lowers these.  The results of \Cref{sec:one-shell} therefore hold for $X_{j,\alpha}$ at every $\alpha$, with the same constants and thresholds.  In particular, $X_{j,\alpha}\ge-C_B$ with $C_B$ as in \Cref{lem:X-lower}, and the mean identity~\eqref{eq:mean-exact} reads
\begin{equation}\label{eq:mean-alpha}
    \bar\varphi_{j,\alpha}(x)
    = -4t\,\alpha\int g(y)\,\Gamma_j^\omega(x,y)\,
    \big\langle H_3^{(\alpha)}(y)\big\rangle_\alpha\,dy\;,
\end{equation}
where
\begin{equation*}
    H_3^{(\alpha)}(y) \;:=\;
    H_3\big(B_{j-1}(y)+\varphi_{j,\alpha}(y);\,
    D^\omega_{j-1}(y)+\alpha\,q^\omega_j(y)\big)\;.
\end{equation*}
The bound $q^\omega_j\le c_q$ turns the centered moments of \Cref{cor:tilted-moments} into
\begin{equation}\label{eq:d2-cent}
    \big\langle\tilde{\varphi}_{j,\alpha}(x)^{2k}
    \big\rangle_\alpha
    \;\le\; (8c_q)^k\,k!
\end{equation}
for every $k\ge1$, $\alpha\in[0,1]$, $x\in K$, and $j\geq j_1$ (\Cref{cor:BL}).

\subsection{The mean shift}\label{sec:d2-mean-shift}
The log-Laplace bound needs a bound on the tilted mean $\bar\varphi_{j,\alpha}$ of the shell field.  The mean identity~\eqref{eq:mean-alpha} determines it only implicitly because the tilted average of $H_3^{(\alpha)}$ on its right side depends on the mean itself.  \Cref{lem:mean-shift} therefore solves the identity under a smallness condition on the background field.

\begin{lemma}[Mean shift bound]\label{lem:mean-shift}
Set
\begin{equation*}
    R_j \;:=\; 1+\sup_{y\in K}\Big(B_{j-1}(y)^2
    +D_j^\omega(y)\Big)\;,
    \qquad
    \epsilon_j \;:=\; 4t\,A_j^\omega\;,
\end{equation*}
so that $\epsilon_j\le4t\,\kappa\,2^{-2j}$ by~\eqref{eq:gg-off} once $j\geq j_0$.  For $j\ge j_1$, if the smallness condition
\begin{equation}\label{eq:smallness-cond}
    \epsilon_j\,R_j^{3/2}\le\tfrac{1}{512}
\end{equation}
holds, then
\begin{equation}\label{eq:mean-shift}
    \sup_{x\in K}\big|\bar\varphi_{j,\alpha}(x)\big|\;\le\;\tfrac12\;,
\end{equation}
for all $\alpha\in[0,1]$.
\end{lemma}

\begin{proof}
Abbreviate $\bar\varphi:=\bar\varphi_{j,\alpha}$, $\tilde\varphi:=\tilde\varphi_{j,\alpha}$, $\delta:=D_{j-1}^\omega+\alpha q_j^\omega \le D_j^\omega$, and $s:=\sup_{x\in K}|\bar\varphi(x)|$.

Taking absolute values in the mean identity~\eqref{eq:mean-alpha}, bounding $\int g(y)\,|\Gamma^\omega_j(x,y)|\,dy$ by $A^\omega_j$, and using $\alpha\le1$,
\begin{equation}\label{eq:S-ineq}
    s\le 4t\,A_j^\omega\;\sup_{y\in K}
    \big|\big\langle H_3^{(\alpha)}(y)
    \big\rangle_\alpha\big|\;.
\end{equation}
Split the field into its mean and centered parts, $B_{j-1}+\varphi_{j,\alpha}=\beta+\tilde\varphi$, with $\beta:=B_{j-1}+\bar\varphi$ and $\langle\tilde\varphi\rangle_\alpha=0$. Expanding $H_3^{(\alpha)}$ about $\beta$ and averaging,
\begin{equation*}
    \big\langle H_3^{(\alpha)}\big\rangle_\alpha
    =H_3(\beta;\delta)+3\beta\,\langle\tilde\varphi^2
    \rangle_\alpha+\langle\tilde\varphi^3\rangle_\alpha
    =\mathcal{S}(y)+\mathcal{R}(y)\;,
\end{equation*}
where $\mathcal S(y)$ is the value at $\bar\varphi=0$,
\begin{equation*}
    \mathcal{S}(y):=H_3(B_{j-1};\delta)
    +3B_{j-1}\,\langle\tilde\varphi^2\rangle_\alpha
    +\langle\tilde\varphi^3\rangle_\alpha\;,
\end{equation*}
and $\mathcal R(y)$ is the remainder.  Since $H_3(B{+}\bar\varphi;\delta)-H_3(B;\delta) =3H_2(B;\delta)\bar\varphi+3B\bar\varphi^2+\bar\varphi^3$,
\begin{equation*}
    \mathcal{R}(y)=3H_2(B_{j-1};\delta)\,\bar\varphi
    +3B_{j-1}\,\bar\varphi^2+\bar\varphi^3
    +3\langle\tilde\varphi^2\rangle_\alpha\,\bar\varphi\;.
\end{equation*}

Set $P_j:=\sup_{y\in K}\big(1+|B_{j-1}(y)|^3+D^\omega_j(y)\,|B_{j-1}(y)|+(q^\omega_j(y))^{3/2}\big)$.  For $\mathcal S(y)$, \Cref{cor:tilted-moments} gives $\langle\tilde\varphi^2\rangle_\alpha\le 8q_j^\omega$ and $|\langle\tilde\varphi^3\rangle_\alpha| \le\langle\tilde\varphi^2\rangle_\alpha^{1/2} \langle\tilde\varphi^4\rangle_\alpha^{1/2} \le32\,(q_j^\omega)^{3/2}$. These bounds, together with $|H_3(B;\delta)|\le|B|^3+3\delta|B|$ and $\delta,q^\omega_j\le D_j^\omega$, give
\begin{equation}\label{eq:source-bd}
    |\mathcal{S}(y)|\le 32\,P_j\;.
\end{equation}
The remainder is a polynomial in $\bar\varphi$ with no constant term, and its coefficients are bounded by multiples of $R_j$, namely $|3H_2(B;\delta)|\le3(B^2{+}\delta)\le3R_j$, $3\langle\tilde\varphi^2\rangle_\alpha\le24R_j$, and $|B_{j-1}|\le R_j^{1/2}\le R_j$.  To bound the powers of $\bar\varphi$ by $s$ itself, we assume for now that $s\le1$, and remove this assumption at the end of the proof.  Then $\bar\varphi^2,|\bar\varphi|^3\le s$, and
\begin{equation}\label{eq:feedback-bd}
    |\mathcal{R}(y)|\le 31\,R_j\,s\;.
\end{equation}

Insert \eqref{eq:source-bd}--\eqref{eq:feedback-bd} into~\eqref{eq:S-ineq}:
\begin{equation*}
    s\le\epsilon_j\big(32\,P_j+31\,R_j\,s\big)\;.
\end{equation*}
Assume that \eqref{eq:smallness-cond} holds.  Since $R_j\ge1$, the condition gives $31\,\epsilon_jR_j\le31\,\epsilon_jR_j^{3/2}\le\tfrac{31}{512}<\tfrac12$, and hence $s\le64\,\epsilon_jP_j$.  Since also $|B|\le R_j^{1/2}$ and $q_j^\omega\le D_j^\omega\le R_j$, each of the four terms of $P_j$ is at most $R_j^{3/2}$, so $P_j\le4R_j^{3/2}$, and the smallness condition gives
\begin{equation}\label{eq:S-half}
    s\;\le\;256\,\epsilon_jR_j^{3/2}\;\le\;\tfrac{256}{512}\;=\;\tfrac12\;.
\end{equation}
This holds for every $\alpha$ with $s(\alpha)\le1$, and it remains to show that $s(\alpha)\le1$ for every $\alpha\in[0,1]$.

The map $\alpha\mapsto s(\alpha)$ is continuous on $[0,1]$.  Each coefficient mean $\langle a_{j,\ell}\rangle_\alpha$ is a ratio of two integrals with positive denominator $e^{\Psi_j(\alpha)}$, and both integrals are continuous in $\alpha$ by dominated convergence, with a dominating function uniform on $[0,1]$ because $X_{j,\alpha}\ge-C_B$.  Hence $\bar\varphi_{j,\alpha}=\sqrt{\alpha}\sum_\ell U_{j,\ell}\,\langle a_{j,\ell}\rangle_\alpha$, a finite sum of functions bounded on $K$, is continuous in $\alpha$ in the supremum norm on $K$, and so is its norm $s(\alpha)$.

Let $\mathcal{I}:=\{\alpha\in[0,1]:s(\alpha)\le1\}$.  It is nonempty because $s(0)=0$, and closed in $[0,1]$ as the preimage of $(-\infty,1]$ under the continuous map $s$.  By~\eqref{eq:S-half}, $s\le\tfrac12$ on $\mathcal{I}$, so $\mathcal{I}$ coincides with $s^{-1}((-\infty,1))$ and is open in $[0,1]$ as well.  Since $[0,1]$ is connected and $\mathcal I$ is a nonempty clopen subset of it, $\mathcal{I}=[0,1]$. Therefore, \eqref{eq:S-half} holds for every $\alpha\in[0,1]$, which is~\eqref{eq:mean-shift}.
\end{proof}

\subsection{The log-Laplace bound}\label{sec:d2-one-shell}

The bound on $\Psi_j$ integrates the following identity over $\alpha$.

\begin{proposition}[Interpolation derivative]
\label{prop:interp}
For every $\alpha\in[0,1]$,
\begin{equation}\label{eq:psi-prime}
    \Psi_j'(\alpha)
    \;=\; 8t^2\iint g(x)\,g(y)\,
    \Gamma_j^\omega(x,y)\,
    \big\langle H_3^{(\alpha)}(x)\,
    H_3^{(\alpha)}(y)\big\rangle_\alpha
    \,dx\,dy\;.
\end{equation}
\end{proposition}

The derivation proceeds via Gaussian integration by parts in $a_j$. The full computation is given in \Cref{app:interpolation}.

The interpolation derivative~\eqref{eq:psi-prime} pairs $H_3^{(\alpha)}$ with itself against the shell kernel, so a bound on $\Psi_j$ needs the second moment of $H_3^{(\alpha)}$ under the tilted measure.  That estimate uses the mean shift bound of \Cref{lem:mean-shift}, which holds only under the smallness condition~\eqref{eq:smallness-cond} on the background field.  We therefore work on a \emph{good-field event} on which that condition holds, and \Cref{sec:d2-converge} handles its complement by returning to the free measure.

\begin{lemma}[Good-field event]\label{lem:d2-good-field}
Set $c_1:=\big(6144\sqrt3\,t\,\kappa\big)^{-1/3}$, and define the $\cF_{j-1}$-measurable good-field event:
\begin{equation}\label{eq:good-field}
    G_j \;:=\; \Big\{(a_0,\dots,a_{j-1}):\,\sup_{x\in K}\big|B_{j-1}(x)\big|
    \;\le\; c_1\,2^{2j/3}\,\Big\}\;.
\end{equation}
There exists $j_2\ge j_1$, depending only on $t$, $\kappa$, $c_q$, and $j_1$, such that for all $j\ge j_2$ and every $(a_0,\dots,a_{j-1})\in G_j$, the smallness condition~\eqref{eq:smallness-cond} holds.  The mean shift bound~\eqref{eq:mean-shift} then applies.
\end{lemma}

\begin{proof}
On $G_j$, using $D^\omega_j\le c_q\,(j{+}1)$,
\begin{equation}
    R_j \;\le\; 1+c_1^22^{4j/3}+c_q\,(j{+}1)\;.
\end{equation}
By~\eqref{eq:power-sum} at $n=3$, $p=3/2$, and $\epsilon_j\le4t\kappa 2^{-2j}$,
\begin{equation}
    \epsilon_j\,R_j^{3/2}
    \;\le\; 4\sqrt3\,t\kappa\Big(2^{-2j}
    + c_1^3 + \big(c_q (j{+}1)\big)^{3/2}\,2^{-2j}\Big)\;.
\end{equation}
Choose $j_2\ge j_1$ so large that for $j\ge j_2$ each of the two $j$-dependent terms in the bracket is at most $c_1^3$.  Then $\epsilon_j R_j^{3/2}\le 12\sqrt3\,t\kappa\,c_1^3=\tfrac{1}{512}$.
\end{proof}

\begin{theorem}[Second moment of $H_3^{(\alpha)}$]\label{thm:d2-tilted}
There is a constant $C_H>0$ such that for all $j\ge j_2$, $\alpha\in[0,1]$, $(a_0,\dots,a_{j-1})\in G_j$, and $x\in K$,
\begin{equation}\label{eq:d2-tilted-bound}
    \big\langle H_3^{(\alpha)}(x)^2
    \big\rangle_\alpha
    \;\le\; C_H\,\big(1+B_{j-1}(x)^6
    +D^\omega_{j-1}(x)^3\big)\;.
\end{equation}
The constant $C_H$ depends only on $c_q$.
\end{theorem}

\begin{proof}
Expand via $\beta:=B_{j-1}(x)+\bar\varphi_{j,\alpha}(x)$, $\delta:=D^\omega_{j-1}(x)+\alpha q^\omega_j(x)$, and $\tilde\varphi:=\tilde\varphi_{j,\alpha}(x)$, as in the proof of \Cref{lem:mean-shift}:
\begin{equation}
    H_3^{(\alpha)}(x) \;=\;
    \tilde{\varphi}^3
    + 3\beta\,\tilde{\varphi}^2
    + (3\beta^2-3\delta)\,\tilde{\varphi}
    + (\beta^3-3\delta\beta).
\end{equation}
Squaring via~\eqref{eq:power-sum} at $n=4$, $p=2$, and taking $\langle\cdot\rangle_\alpha$:
\begin{equation}\label{eq:d2-H3-sq}
    \big\langle H_3^{(\alpha)}(x)^2
    \big\rangle_\alpha
    \;\le\; 4\Big(
    \langle\tilde{\varphi}^6\rangle_\alpha
    + 9\beta^2\,\langle\tilde{\varphi}^4\rangle_\alpha
    + 9(\beta^2{+}\delta)^2\,
    \langle\tilde{\varphi}^2\rangle_\alpha
    + (\beta^3{+}3\delta|\beta|)^2\Big)\;.
\end{equation}
Bounding the centered tilted moments by \eqref{eq:d2-cent} gives
\begin{equation}
    \big\langle H_3^{(\alpha)}(x)^2\big\rangle_\alpha
    \;\le\; C_1\,\big(1+\beta^2+\beta^4+\beta^2\delta
    +\delta^2+\beta^6+\beta^4\delta+\beta^2\delta^2\big)\;.
\end{equation}
Every monomial $\beta^a\delta^b$ on the right has $a$ even and $a+2b\le6$.  Young's inequality for three factors, with weights $\tfrac{a}{6}$, $\tfrac{b}{3}$, and $1-\tfrac{a}{6}-\tfrac{b}{3}\ge0$, then gives
\begin{equation}
    \beta^a\delta^b
    \;\le\; \tfrac{a}{6}\,\beta^6+\tfrac{b}{3}\,\delta^3
    +\Big(1-\tfrac{a}{6}-\tfrac{b}{3}\Big)
    \;\le\; 1+\beta^6+\delta^3\;.
\end{equation}
Therefore,
\begin{equation}
    \big\langle H_3^{(\alpha)}(x)^2\big\rangle_\alpha
    \;\le\; C_2\,\big(1+\beta^6+\delta^3\big)\;.
\end{equation}

It remains to replace $\beta$ by $B:=B_{j-1}(x)$ and $\delta$ by $D:=D^\omega_{j-1}(x)$.  On $G_j$, \Cref{lem:d2-good-field} gives $|\bar{\varphi}_{j,\alpha}(x)|\le\tfrac12$ for all $x$ and $\alpha$, so
\begin{equation*}
    |\beta|^6
    \;=\; \big|B+\bar{\varphi}_{j,\alpha}(x)\big|^6
    \;\le\; 2^5\big(|B|^6+1\big)\;,
\end{equation*}
while $\delta=D+\alpha\,q^\omega_j(x)\le D+c_q$ gives
\begin{equation*}
    \delta^3 \;\le\; (D+c_q)^3 \;\le\; 4\big(D^3+c_q^3\big)\;.
\end{equation*}
Therefore $1+\beta^6+\delta^3\le C_3\big(1+B^6+D^3\big)$, which is~\eqref{eq:d2-tilted-bound} with $C_H:=C_2C_3$.
\end{proof}

\begin{theorem}[Log-Laplace bound]\label{thm:d2-one-shell}
There is a constant $C_\star>0$ such that for all $j\ge j_2$ and $(a_0,\dots,a_{j-1})\in G_j$,
\begin{equation}\label{eq:d2-one-shell}
    0 \;\le\; \Psi_j(1) \;\le\;
    C_\star\,2^{-2j}\int g(x)\,
    \big(1+B_{j-1}(x)^6+D^\omega_{j-1}(x)^3\big)\,dx\;.
\end{equation}
The constant $C_\star$ depends only on $t$, $\kappa$, and $C_H$.
\end{theorem}

\begin{proof}
The lower bound is \Cref{lem:positivity}.  For the upper bound, integrate the interpolation derivative~\eqref{eq:psi-prime}:
\begin{equation}
    \Psi_j(1) \;=\; \int_0^1\!\Psi_j'(\alpha)\,d\alpha
    \;=\; 8t^2\int_0^1\!\iint g(x)\,g(y)\,
    \Gamma^\omega_j(x,y)\,\big\langle H_3^{(\alpha)}(x)\,
    H_3^{(\alpha)}(y)\big\rangle_\alpha
    \,dx\,dy\,d\alpha\;.
\end{equation}
By Cauchy--Schwarz under $\langle\cdot\rangle_\alpha$, followed by AM--GM:
\begin{equation}
    \big|\big\langle H_3^{(\alpha)}(x)\,
    H_3^{(\alpha)}(y)\big\rangle_\alpha\big|
    \;\le\; \sqrt{\big\langle H_3^{(\alpha)}(x)^2
    \big\rangle_\alpha\,
    \big\langle H_3^{(\alpha)}(y)^2
    \big\rangle_\alpha}
    \;\le\; \tfrac{1}{2}\Big(
    \big\langle H_3^{(\alpha)}(x)^2
    \big\rangle_\alpha
    +\big\langle H_3^{(\alpha)}(y)^2
    \big\rangle_\alpha\Big)\;.
\end{equation}
Substituting, and using the symmetry $\Gamma^\omega_j(x,y)=\Gamma^\omega_j(y,x)$,
\begin{align}
    \Psi_j(1)
    &\;\le\; 8t^2\int_0^1\!\int g(x)\,
    \big\langle H_3^{(\alpha)}(x)^2
    \big\rangle_\alpha\,
    \bigg(\int g(y)\,|\Gamma^\omega_j(x,y)|\,dy\bigg)
    dx\,d\alpha
    \notag\\
    &\;\le\; 8t^2\,A_j^\omega\int_0^1\!\int g(x)\,
    \big\langle H_3^{(\alpha)}(x)^2
    \big\rangle_\alpha\,dx\,d\alpha\;.
\end{align}
Applying the second-moment bound (\Cref{thm:d2-tilted}) and $A_j^\omega\le\kappa\,2^{-2j}$:
\begin{equation}
    \Psi_j(1) \;\le\;
    8t^2\kappa C_H\,2^{-2j}\int g(x)\,
    \big(1+B_{j-1}(x)^6+D^\omega_{j-1}(x)^3\big)\,dx\;,
\end{equation}
which is~\eqref{eq:d2-one-shell} with $C_\star:=8t^2\kappa C_H$.
\end{proof}

From here on we write, for $j\ge1$,
\begin{equation}\label{eq:fJ-def}
    f_j \;:=\; C_\star\,2^{-2j}\int g(y)\,
    \big(1+B_{j-1}(y)^6+D^\omega_{j-1}(y)^3\big)\,dy\;,
    \qquad
    \mathcal{G}_j \;:=\; G_j\cap\{f_j\le1\}\;,
\end{equation}
so that for $j\geq j_2$, \Cref{thm:d2-one-shell} reads $0\le\Psi_j\le f_j$ on $G_j$, and $0\le\Psi_j\le f_j\le1$ on $\mathcal{G}_j$.

\section{Moment stability and convergence}
\label{sec:d2-converge}

\Cref{thm:d2-one-shell} bounds $\Psi_j$ by $2^{-2j}$ times a sextic in the background field, so bounding the telescoping factors $\E_{\mu_{j-1}}[e^{\Psi_j}]$ of~\eqref{eq:telescope} requires control over the moments of the background under the interacting measure.  This section bounds those moments and removes the ultraviolet cutoff at fixed width.

\subsection{Large backgrounds}\label{sec:d2-free}

The estimates below return to the free measure on an event where the supremum of $|B_{j-1}|$ on $K$ is large, and the resulting bound is small only if that event is rare under the free measure.  The next lemma bounds the free-measure probability that this supremum exceeds a given value.

\begin{lemma}[Free-field supremum tail]\label{lem:d2-sup-tail}
There are constants $c_2,C_K>0$ such that for all $J\ge1$ and all $s>0$,
\begin{equation}\label{eq:d2-sup-tail}
    \gamma_{J-1}\Big(\sup_{x\in K}\big|B_{J-1}(x)\big|>s\Big)
    \;\le\; C_K\,16^{J}
    \exp\!\Big(-\frac{c_2\,s^2}{c_qJ}\Big)\;.
\end{equation}
The constant $c_2$ is universal, and $C_K$ depends only on $K$.
\end{lemma}
The proof is a chaining argument, given in \Cref{app:sup-tail}.

The free measure is always used on the complement of $\mathcal{G}_j$ from~\eqref{eq:fJ-def}, which is the union of $\{f_j>1\}$ and $G_j^c$.  Above a finite threshold in $j$, each of the two requires $\sup_K|B_{j-1}|$ to exceed a positive power of $2^{j}$, and \Cref{lem:d2-sup-tail} at the smaller of the two powers shows that the union is super-exponentially rare.

\begin{corollary}[Rarity of the bad event]\label{cor:d2-bad-prob}
There are a constant $c_3>0$ and a threshold $j_3\ge1$, depending only on $t$, $g$, $\kappa$, and $c_q$, such that for all $j\ge j_3$
\begin{equation}\label{eq:d2-bad-prob}
    \gamma_j\big(\mathcal{G}_j^c\big)
    \;\le\; C_K\exp\Big(-\frac{c_3\,2^{2j/3}}{j}\Big),
\end{equation}
with $\mathcal{G}_j$ as in~\eqref{eq:fJ-def} and $C_K$ as in \Cref{lem:d2-sup-tail}.
\end{corollary}

\begin{proof}
The event $\mathcal{G}_j^c=\{f_j>1\}\cup G_j^c$ is $\cF_{j-1}$-measurable, and because the marginal of $\gamma_j$ on the first $j$ shells is $\gamma_{j-1}$, we have $\gamma_j(\mathcal{G}_j^c)=\gamma_{j-1}(\mathcal{G}_j^c)$.

Take $j_3$ large enough that for $j\geq j_3$,
\begin{equation}\label{eq:d2-fj-const}
    C_\star\,\|g\|_1\big(1+(c_qj)^3\big)\,2^{-2j}\;\le\;\tfrac12\;.
\end{equation}
The field-independent part of $f_j$ is then at most $\tfrac12$, so $f_j>1$ forces $C_\star2^{-2j}\!\int g\,B_{j-1}^6\,dy>\tfrac12$. Bounding the integral by $\|g\|_1\big(\sup_K|B_{j-1}|\big)^6$, we have
\begin{equation*}
    \sup_K|B_{j-1}|\;>\;C_1\,2^{j/3},
    \qquad C_1:=\big(2C_\star\|g\|_1\big)^{-1/6},
\end{equation*}
and $G_j^c$ forces $\sup_K|B_{j-1}|>c_1\,2^{2j/3}$ by~\eqref{eq:good-field}.  After enlarging $j_3$ so that $c_1 2^{2j/3}\ge C_1\,2^{j/3}$, both parts lie in $\{\sup_K|B_{j-1}|>C_1\,2^{j/3}\}$, and \Cref{lem:d2-sup-tail} at index $j$ and threshold $s=C_1\,2^{j/3}$ gives
\begin{equation*}
    \gamma_{j-1}\big(\mathcal{G}_j^c\big)
    \;\le\; C_K\,16^{j}
    \exp\!\Big(-\frac{c_2C_1^2\,2^{2j/3}}{c_qj}\Big)\;.
\end{equation*}
Enlarging $j_3$ once more so that $16^{j}\le\exp\big(c_2C_1^22^{2j/3}/(2c_qj)\big)$, which holds for large $j$ because $2^{2j/3}$ outgrows $j^2$, gives~\eqref{eq:d2-bad-prob} with $c_3:=c_2C_1^2/(2c_q)$.
\end{proof}

The estimates that follow need both the log-Laplace bound (\Cref{thm:d2-one-shell}) and \Cref{cor:d2-bad-prob}, so we write
\begin{equation}\label{eq:j4}
    j_4 \;:=\; \max\{j_2,\,j_3\}
\end{equation}
for their common threshold.

\subsection{Moment stability}\label{sec:d2-moments}

The log-Laplace bound is sextic in the background field, so only the sixth moment of the background needs to be controlled for the construction.  The theorem nonetheless covers every moment order because its induction step at order $p$ uses moments of $B_{J-1}$ up to order $6(p+3)$.

\begin{theorem}[Moment stability]\label{thm:d2-moments}
There are constants $C_M,\mathcal{K}>0$ such that for all $J\ge0$, $p\ge1$, and $x\in K$,
\begin{equation}\label{eq:d2-moments}
    \E_{\mu_J}\big[B_J(x)^{2p}\big]
    \;\le\; \mathcal{K}\,C_M^{\,p}\,p!\,(J{+}1)^p\;.
\end{equation}
The growth constant $C_M$ depends only on $c_q$, and the prefactor $\mathcal{K}$ depends only on $t$, $g$, $c_q$, $\kappa$, and the threshold $j_4$ of~\eqref{eq:j4}.
\end{theorem}

\begin{proof}
Throughout the proof, we use $(p!)^{1/p}\in[p/e,\,p]$ repeatedly.

The constants are fixed in the following order.  First,
\begin{equation}\label{eq:d2-Cg}
    C_M \;:=\; 512\,e\,c_q\;,
\end{equation}
so that $C_M$ depends only on $c_q$.  Next, with $c_V$ as in~\eqref{eq:d2-C1}, let $J_\star\ge j_4$ be the finite threshold determined by conditions \eqref{eq:d2-Jstar-a}--\eqref{eq:d2-Jstar-c} of Step 5, which depends only on $t$, $g$, $\kappa$, $c_q$, and $j_4$. Finally,
\begin{equation}\label{eq:d2-K}
    \mathcal{K} \;:=\; e^{c_V(J_\star+1)^2}\;,
\end{equation}
so that $\mathcal{K}\ge1$.  With $A_p:=C_M^{\,p}\,p!$, the claim~\eqref{eq:d2-moments} reads $\E_{\mu_J}[B_J(x)^{2p}]\le\mathcal{K}A_p(J+1)^p$.

\medskip\noindent\textit{Step 0: direct bound.}
For every $J\ge0$, $p\ge1$, and $x\in K$,
\begin{equation}\label{eq:d2-direct}
    \E_{\mu_J}\big[B_J(x)^{2p}\big]
    \;\le\; e^{c_V(J+1)^2}\;p!\,
    \big(2c_q(J{+}1)\big)^p\;.
\end{equation}
Indeed, $\E_{\mu_J}[F]=\E_{\gamma_J}[e^{-V_J}F]/Z_J$ with $Z_J\ge1$ by \Cref{lem:V-lower}, and $e^{-V_J}\le e^{c_V(J+1)^2}$ pointwise by~\eqref{eq:d2-C1}.  Under $\gamma_J$ the variable $B_J(x)$ is $\cN(0,D^\omega_J(x))$, and $(2p{-}1)!!\le 2^p p!$.  Consequently \eqref{eq:d2-moments} holds:
\begin{enumerate}[label=(\roman*)]
\item for all $p\ge1$ whenever $J\le J_\star$, by the
choice~\eqref{eq:d2-K} of $\mathcal{K}$ and $C_M>2c_q$. This covers the base case $J=0$ and the finitely many shells below $J_\star$ without induction;
\item for $J>J_\star$ whenever
\begin{equation}\label{eq:d2-p1}
    p \;\ge\; p_1(J) \;:=\;
    \frac{c_V\,(J{+}1)^2}{\ln\!\big(C_M/(2c_q)\big)}\;,
\end{equation}
without using $\mathcal{K}$ at all.
\end{enumerate}
It therefore suffices to carry out the inductive step at pairs $(J,p)$ with $J>J_\star$ and $p<p_1(J)$, assuming \eqref{eq:d2-moments} at level $J-1$ for all orders in $p$.

\medskip\noindent\textit{Step 1: reduction and split.}
Fix $x\in K$.  Throughout Steps~1--4 we abbreviate $B(x):=B_{J-1}(x)$ and $D(x):=D^\omega_{J-1}(x)$.  Since $V_J=V_{J-1}+X_J$ and $\gamma_J=\gamma_{J-1}\otimes\gamma_\sigma(a_J)$, integrating over $a_J$ first and dividing by~\eqref{eq:telescope-step} gives
\begin{equation}\label{eq:d2-ratio}
    \E_{\mu_J}\big[B_J(x)^{2p}\big]
    \;=\; \frac{\E_{\mu_{J-1}}\big[e^{\Psi_J}\,
    \big\langle(B(x)+\varphi_J(x))^{2p}\big\rangle\big]}%
    {\E_{\mu_{J-1}}\big[e^{\Psi_J}\big]}\;,
\end{equation}
where $\langle\cdot\rangle$ denotes the tilted expectation over $a_J$.  Since $\Psi_J\ge 0$ (\Cref{lem:positivity}), the denominator is at least~$1$:
\begin{equation}\label{eq:d2-drop-denom}
    \E_{\mu_J}\big[B_J(x)^{2p}\big]
    \;\le\; \E_{\mu_{J-1}}\big[e^{\Psi_J}\,\mathcal{M}_p\big]\;,
    \qquad \mathcal{M}_p:=\big\langle(B(x)+\varphi_J(x))^{2p}
    \big\rangle\;\ge\;0\;.
\end{equation}
With $f_J$ and $\mathcal{G}_J$ as in~\eqref{eq:fJ-def}, we have $0\le\Psi_J\le f_J\le1$ on $\mathcal{G}_J$. Hence $e^{\Psi_J}\le1+2f_J$ there, and
\begin{equation}\label{eq:d2-split}
    \E_{\mu_{J-1}}\big[e^{\Psi_J}\mathcal{M}_p\big]
    \;\le\; \E_{\mu_{J-1}}\big[\mathcal{M}_p\,\mathbf{1}_{G_J}\big]
    \;+\; 2\,\E_{\mu_{J-1}}\big[f_J\,\mathcal{M}_p\,\mathbf{1}_{G_J}\big]
    \;+\; \E_{\mu_{J-1}}\big[e^{\Psi_J}\mathcal{M}_p\,
    \mathbf{1}_{\mathcal{G}_J^c}\big]\;.
\end{equation}
Steps 2--4 bound the three terms of \eqref{eq:d2-split} in order.

\medskip\noindent\textit{Step 2: main term.}
Split the shell field into its tilted mean $\bar\varphi:=\langle\varphi_J(x)\rangle$ and centered part $\tilde\varphi:=\varphi_J(x)-\bar\varphi$.  Then $\mathcal{M}_p=\big\langle(B(x)+\bar\varphi+\tilde\varphi)^{2p}\big\rangle$ and $\langle\tilde\varphi\rangle=0$.  The centered moments obey
\begin{equation}\label{eq:d2-phit-mom}
    \big|\langle\tilde{\varphi}^{\,k}\rangle\big|
    \;\le\; (8c_q)^{k/2}\,\lceil k/2\rceil!\qquad(k\ge2)\;.
\end{equation}
For even $k$ this is~\eqref{eq:d2-cent}.  For odd $k=2m{+}1$ it follows from Cauchy--Schwarz between the orders $2m$ and $2m{+}2$, since $\sqrt{m!\,(m+1)!}\le(m+1)!$.

On $G_J$, the mean identity~\eqref{eq:mean-exact} bounds $|\bar\varphi|$ by $4t$ times the integral of $|\langle H_3(y)\rangle|$ against a measure $g(y)\,|\Gamma^\omega_J(x,y)|\,dy$ of total mass at most $\kappa 2^{-2J}$.  Jensen's inequality for that normalized measure moves the power $2p$ inside the integral.  For the integrand, \Cref{thm:d2-tilted} and~\eqref{eq:power-sum} give $|\langle H_3(y)\rangle|^{2p}\le\langle H_3(y)^2\rangle^p\le3^{p-1}C_H^p\big(1+B(y)^{6p}+D(y)^{3p}\big)$.  Under $\mu_{J-1}$, the inductive hypothesis bounds $\E_{\mu_{J-1}}[B(y)^{6p}]$ by $\mathcal{K}A_{3p}J^{3p}$, and $1+D(y)^{3p}\le2A_{3p}J^{3p}$ because $D\le c_qJ$ and $C_M\ge c_q\ge1$.  All together,
\begin{equation}\label{eq:d2-mean-power}
    \E_{\mu_{J-1}}\big[\,|\bar{\varphi}|^{2p}\,
    \mathbf{1}_{G_J}\big]
    \;\le\; (4t\kappa)^{2p}\,2^{-4pJ}\,(3C_H)^p\,
    \mathcal{K}A_{3p}J^{3p}
    \;=:\;\mathcal{K}\,\Theta_p^{2p}.
\end{equation}
Two consequences follow: a bound on the size of the mean shift relative to that of the background, and the main-term estimate. First, using $(A_pJ^p)^{1/2p}\ge\sqrt{C_M Jp/e}$ and $\big((3p)!\big)^{1/2p}\le(3p)^{3/2}$,
\begin{equation}\label{eq:d2-theta}
    \theta \;:=\;
    \frac{\Theta_p}{(A_pJ^p)^{1/2p}}
    \;\le\; C_\theta\,C_M\,p\,J\,2^{-2J}\;,
    \qquad C_\theta:=36\,t\,\kappa\sqrt{C_He}\;.
\end{equation}
Second, expand $(B(x)+\bar\varphi)^{2p}$ binomially and apply H\"older's inequality with exponents $\big(\tfrac{2p}{2p-k},\tfrac{2p}{k}\big)$ to the $k$-th term.  By~\eqref{eq:d2-mean-power} and the inductive hypothesis, the factors of $\mathcal{K}$ then recombine to first power:
\begin{align}
    \mathcal{E} \;:=\;
    \E_{\mu_{J-1}}\big[\big(B(x)+\bar\varphi\big)^{2p}\,\mathbf{1}_{G_J}\big]
    \;&\le\; \sum_{k=0}^{2p}\binom{2p}{k}
    \big(\mathcal{K}A_pJ^p\big)^{\frac{2p-k}{2p}}\,
    \big(\mathcal{K}\Theta_p^{2p}\big)^{\frac{k}{2p}}\notag\\[10 pt]
    &\;=\; \mathcal{K}A_pJ^p\,(1+\theta)^{2p}
    \;\le\; \mathcal{K}A_pJ^p\,e^{2p\theta}\label{eq:d2-T-bound}\;.
\end{align}
By the binomial theorem, $\langle\tilde\varphi\rangle=0$, \eqref{eq:d2-phit-mom}, and Jensen ($\E_{\mu_{J-1}}[|B(x)+\bar\varphi|^{2p-k}\mathbf{1}_{G_J}]\le \mathcal{E}^{(2p-k)/2p}$),
\begin{equation}
    \E_{\mu_{J-1}}\big[\mathcal{M}_p\,\mathbf{1}_{G_J}\big]
    \;\le\; \mathcal{E}+\sum_{k=2}^{2p}\binom{2p}{k}
    (8c_q)^{k/2}\,\lceil k/2\rceil!\;\mathcal{E}^{\frac{2p-k}{2p}}.
\end{equation}
Each term is nondecreasing in $\mathcal{E}$ (its exponents are nonnegative), so $\mathcal{E}$ may be replaced throughout by its upper bound~\eqref{eq:d2-T-bound}, $\mathcal{E}_{\max}:=\mathcal{K}A_pJ^p\,e^{2p\theta}$:
\begin{equation}\label{eq:d2-bracket}
    \E_{\mu_{J-1}}\big[\mathcal{M}_p\,\mathbf{1}_{G_J}\big]
    \;\le\; \mathcal{E}_{\max}\,\bigg(1+\sum_{k=2}^{2p}\binom{2p}{k}\,
    (8c_q)^{k/2}\,\lceil k/2\rceil!\;\mathcal{E}_{\max}^{-k/2p}\bigg)\;.
\end{equation}
Now bound the summands.  Using $k!\ge\lceil k/2\rceil!\,\lfloor k/2\rfloor!$ in the first bound,
\begin{equation*}
    \binom{2p}{k}\lceil k/2\rceil!\;\le\;\frac{(2p)^k}{\lfloor k/2\rfloor!}\;,
    \qquad
    \mathcal{E}_{\max}^{-k/2p}\;\le\;(A_pJ^p)^{-k/2p}\;\le\;\Big(\frac{e}{C_MJp}\Big)^{k/2}\;,
\end{equation*}
so the $k$-th summand is at most $\zeta^{\,k}/\lfloor k/2\rfloor!$ with $\zeta:=\big(32\,e\,c_q\,p/(C_M J)\big)^{1/2}$.  Splitting into even $k=2m$ and odd $k=2m{+}1$ (both with $\lfloor k/2\rfloor=m$) and extending the sums to infinity,
\begin{equation*}
    \sum_{k=2}^{2p}\frac{\zeta^{\,k}}{\lfloor k/2\rfloor!}
    \;\le\; (1+\zeta)\sum_{m=1}^{\infty}\frac{\zeta^{2m}}{m!}
    \;=\; (1+\zeta)\big(e^{\zeta^2}-1\big)
    \;\le\; e^{2\zeta^2}-1\;,
\end{equation*}
using $e^{2\zeta^2}-1=(e^{\zeta^2}-1)(e^{\zeta^2}+1)$ and $\zeta\le e^{\zeta^2}$ in the last step.  The bracket in~\eqref{eq:d2-bracket} is therefore at most $e^{2\zeta^2}$, and hence
\begin{equation}\label{eq:d2-main-final}
    \E_{\mu_{J-1}}\big[\mathcal{M}_p\,\mathbf{1}_{G_J}\big]
    \;\le\; \mathcal{K}A_pJ^p\,
    e^{2p\theta+2\zeta^2}\;,
    \qquad
    2\zeta^2 \;=\;\frac{u}{8}
    \quad\text{by } C_M=512\,e\,c_q\;,
\end{equation}
with $u:=p/J$.

\medskip\noindent\textit{Step 3: middle term.}
Insert the definition of $f_J$ into the middle term of~\eqref{eq:d2-split} and bound $D(y)^3\le(c_qJ)^3$. The term requiring work is $\E_{\mu_{J-1}}[B(y)^6\,\mathcal{M}_p\,\mathbf{1}_{G_J}]$. H\"older with exponents $\big(\tfrac{p+3}{3},\tfrac{p+3}{p}\big)$, followed by Jensen's inequality under the tilted expectation, $\mathcal{M}_p^{(p+3)/p}\le\mathcal{M}_{p+3}$, gives two factors of order $p{+}3$:
\begin{equation*}
    \E_{\mu_{J-1}}\big[B(y)^6\,\mathcal{M}_p\,\mathbf{1}_{G_J}\big]
    \;\le\;
    \E_{\mu_{J-1}}\big[B(y)^{2(p+3)}\,\mathbf{1}_{G_J}\big]^{\frac{3}{p+3}}\;
    \E_{\mu_{J-1}}\big[\mathcal{M}_{p+3}\,\mathbf{1}_{G_J}\big]^{\frac{p}{p+3}}\;.
\end{equation*}
Bound the first factor by the level-$(J{-}1)$ hypothesis and the second by \eqref{eq:d2-main-final} at order $p{+}3$, writing $\theta',\zeta'$ for the order-$(p{+}3)$ instances of $\theta,\zeta$. The $\mathcal{K}$-powers, at exponents $\tfrac{3}{p+3}$ and $\tfrac{p}{p+3}$, recombine to first power:
\begin{align}
    \E_{\mu_{J-1}}\big[B(y)^6\,\mathcal{M}_p\,\mathbf{1}_{G_J}\big]
    &\;\le\;
    \big(\mathcal{K}A_{p+3}J^{p+3}\big)^{\frac{3}{p+3}}
    \Big(\mathcal{K}A_{p+3}J^{p+3}\,
    e^{2(p+3)\theta'+2\zeta'^2}\Big)^{\frac{p}{p+3}}\notag\\
    &\;\le\; \mathcal{K}A_{p+3}J^{p+3}\,
    e^{2(p+3)\theta'+2\zeta'^2}\;,\label{eq:step3-intermed}
\end{align}
dropping the exponent $\tfrac{p}{p+3}\le1$ from the exponential.  The other two pieces of $f_J$ give the main-term expectation $\E_{\mu_{J-1}}[\mathcal{M}_p\,\mathbf{1}_{G_J}]$ times $1$ and $(c_qJ)^3$.  By~\eqref{eq:d2-main-final} and $(c_qJ)^3\,A_pJ^p\le C_M^3A_p\,J^{p+3} \le A_{p+3}J^{p+3}$, they take the same form as~\eqref{eq:step3-intermed} with exponential $e^{2p\theta+2\zeta^2}$.  Since $\zeta\le\zeta'$, all three pieces are bounded using the combined exponent $2p\theta+2(p{+}3)\theta'+2\zeta'^2$.  Integrating against $g$ and collecting constants into $C_2:=6\,C_\star\|g\|_1$ gives
\begin{equation}\label{eq:d2-small}
    2\,\E_{\mu_{J-1}}\big[f_J\,\mathcal{M}_p\,\mathbf{1}_{G_J}\big]
    \;\le\; C_2\,2^{-2J}\,\mathcal{K}A_{p+3}\,
    J^{p+3}\,e^{2p\theta+2(p+3)\theta'+2\zeta'^2},
    \qquad
    2\zeta'^2 \;=\;\frac{p{+}3}{8J}\;.
\end{equation}

\medskip\noindent\textit{Step 4: complementary event.}
Since $e^{\Psi_J}\langle B_J(x)^{2p}\rangle=\E_{a_J}[e^{-X_J}B_J(x)^{2p}]$ and $Z_{J-1}\ge1$, the last term of \eqref{eq:d2-split} is at most an expectation under the free measure, with no factor of $\mathcal{K}$:
\begin{equation}\label{eq:d2-large}
    \E_{\mu_{J-1}}\big[e^{\Psi_J}\mathcal{M}_p\,\mathbf{1}_{\mathcal{G}_J^c}\big]
    \;\le\; \E_{\gamma_J}\big[e^{-V_J}\,B_J(x)^{2p}\,
    \mathbf{1}_{\mathcal{G}_J^c}\big]
    \;\le\; e^{c_V(J+1)^2}
    \big(4p\,c_q(J{+}1)\big)^p\,
    \gamma_J\big(\mathcal{G}_J^c\big)^{1/2}
\end{equation}
by \Cref{lem:V-lower}, Cauchy--Schwarz, the Gaussian moments of Step~0, and the bound $(4p-1)!!\leq(4p)^{2p}$.  \Cref{cor:d2-bad-prob}, available because $J>J_\star\ge j_4$, bounds the probability of the bad event:
\begin{equation*}
    \gamma_J\big(\mathcal{G}_J^c\big)
    \;\le\; C_K\exp\!\Big(
    -\frac{c_3\,2^{2J/3}}{J}\Big)\;.
\end{equation*}

\medskip\noindent\textit{Step 5: closing.}
The three terms of~\eqref{eq:d2-split} are bounded by \eqref{eq:d2-main-final}, \eqref{eq:d2-small}, and~\eqref{eq:d2-large}, and the target is~\eqref{eq:d2-moments} at level $J$.  It suffices that the main term \eqref{eq:d2-main-final} be at most $\mathcal{K}A_pJ^p\,e^{u/4}$ and that \eqref{eq:d2-small} and \eqref{eq:d2-large} each be at most $\tfrac{u}{8}\,e^{u/4}\,\mathcal{K}A_pJ^p$.  Define $J_\star$ as the smallest integer at least $j_4$ such that for all $J>J_\star$ and all $p<p_1(J)$:
\begin{subequations}\label{eq:d2-Jstar}
\begin{align}
    16\,C_\theta\,C_M\,(p{+}3)^2J^2\,2^{-2J}
    &\;\le\; 1\;,
    \label{eq:d2-Jstar-a}\\[5 pt]
    C_3\,C_M^3\,(p{+}3)^3J^4\,2^{-2J}
    &\;\le\; 1\;,
    \label{eq:d2-Jstar-b}\\[5 pt]
    C_4J\,e^{c_V(J+1)^2+p/J}
    \exp\!\big(-c_3\,2^{2J/3}/(2J)\big)
    &\;\le\; 1\;,
    \label{eq:d2-Jstar-c}
\end{align}
\end{subequations}
with $C_3:=8e^{1/2}C_2$ and $C_4:=8C_K^{1/2}$. Such a $J_\star$ exists.  Since $p<p_1(J)=O(J^2)$, the left sides of~\eqref{eq:d2-Jstar-a} and~\eqref{eq:d2-Jstar-b} are bounded by polynomials in $J$ times $2^{-2J}$, and in~\eqref{eq:d2-Jstar-c} the factor $e^{c_V(J+1)^2+p/J}=e^{O(J^2)}$ is outweighed by the super-exponential decay of $\exp\!\big(-c_3\,2^{2J/3}/(2J)\big)$.  All three left sides therefore tend to zero as $J\to\infty$, uniformly over $p<p_1(J)$.

For $J>J_\star$ and $p<p_1(J)$:
\begin{enumerate}[label=(\alph*)]
\item By~\eqref{eq:d2-Jstar-a} and~\eqref{eq:d2-theta}, $2p\theta\le1/(8J)\le u/8$, and with $2\zeta^2= u/8$ from~\eqref{eq:d2-main-final} the Step-2 term is at most $\mathcal{K}A_pJ^p\,e^{u/4}$.
\item The argument of~(a) at order $p+3$ gives $2(p+3)\theta'\leq 1/(8J)$.  With~(a) and $2\zeta'^2=(p+3)/(8J)$, the exponent of~\eqref{eq:d2-small} is at most $(p{+}5)/(8J)\le u/4+\tfrac12$.  The ratio of~\eqref{eq:d2-small} to $\tfrac{u}{8}e^{u/4}\mathcal{K}A_pJ^p$ is then at most the left side of~\eqref{eq:d2-Jstar-b}, so the Step-3 term is at most $\tfrac{u}{8}e^{u/4}\mathcal{K}A_pJ^p$.
\item Using $p^p\le e^p\,p!$ and $(1+1/J)^p\le e^{p/J}$, the ratio of~\eqref{eq:d2-large} to $\tfrac{u}{8}e^{u/4}\mathcal{K}A_pJ^p$ is at most the left side of~\eqref{eq:d2-Jstar-c} times $e^{-u/4}/(128^p\,p\,\mathcal{K})\le1$, so the Step-4 term is at most $\tfrac{u}{8}e^{u/4}\mathcal{K}A_pJ^p$.
\end{enumerate}
Collecting the three terms:
\begin{align}
    \E_{\mu_J}\big[B_J(x)^{2p}\big]
    &\;\le\; \mathcal{K}A_pJ^p\,e^{u/4}
    + \tfrac{u}{4}\,\mathcal{K}A_pJ^p\,e^{u/4}
    \;=\; \mathcal{K}A_pJ^p\,e^{u/4}\big(1+\tfrac{u}{4}\big)\notag\\[5 pt]
    &\;\le\; \mathcal{K}A_pJ^p\,e^{u/2}
    \;\le\; \mathcal{K}A_p(J{+}1)^p\;,
\end{align}
using $1+u/4\le e^{u/4}$ and then $u/2\le p\ln(1+1/J)$. Together with Step~0(i)--(ii), this establishes~\eqref{eq:d2-moments} for all $(J,p)$.
\end{proof}

\subsection{Convergence and existence}\label{sec:d2-existence}

\begin{proposition}[Telescoping factor]
\label{prop:d2-exp-int}
There is a constant $C_T>0$ such that for every $j\ge1$,
\begin{equation}\label{eq:d2-exp-int}
    \E_{\mu_{j-1}}\big[e^{\Psi_j}\big]
    \;\le\; 1 + C_T\,2^{-2j}\,j^3\;.
\end{equation}
The constant $C_T$ depends only on $t$, $g$, $c_q$, $\kappa$, and $j_4$.
\end{proposition}

\begin{proof}
First suppose $j\geq j_4$. Since $D^\omega_{j-1}\le c_q\,j$ and $\E_{\mu_{j-1}}[B_{j-1}(y)^6]\le6\,\mathcal{K}C_M^3\,j^3$ at every $y\in K$ (\Cref{thm:d2-moments}),
\begin{equation}\label{eq:d2-eta}
    \E_{\mu_{j-1}}[f_j]
    \;\le\; C_1\,2^{-2j}\,j^3\;,
    \qquad C_1\;:=\;C_\star\,\|g\|_1\big(1+6\,\mathcal{K}C_M^3+c_q^3\big)\;.
\end{equation}
On $\mathcal{G}_j$, $e^{\Psi_j}\le1+2f_j$, and hence $\E_{\mu_{j-1}}[e^{\Psi_j}\mathbf{1}_{\mathcal{G}_j}]\le1+2C_1\,2^{-2j}j^3$.  Because $Z_{j-1}\ge1$, the expectation over $\mathcal{G}_j^c$ is at most $\E_{\gamma_j}[e^{-V_j}\mathbf{1}_{\mathcal{G}_j^c}]$, and \Cref{lem:V-lower} and \Cref{cor:d2-bad-prob} bound this by $e^{c_V(j+1)^2}C_K\exp(-c_3\,2^{2j/3}/j)$, which is $o(2^{-2j}\,j^3)$.  Adding the two bounds gives~\eqref{eq:d2-exp-int} for $j\ge j_4$.  For the finitely many shells $1\le j<j_4$, equation \eqref{eq:telescope-step}, \Cref{lem:V-lower}, and~\eqref{eq:d2-C1} give $\E_{\mu_{j-1}}[e^{\Psi_j}]=Z_j/Z_{j-1}\le e^{c_V(j+1)^2}$, and enlarging $C_T$ covers these shells.
\end{proof}

\begin{theorem}[UV convergence]\label{thm:d2-convergence}
The log partition function is given by
\begin{equation*}
    \log Z_\Lambda \;=\; \log Z_0
    +\sum_{j=1}^\Lambda\log\E_{\mu_{j-1}}\big[e^{\Psi_j}\big]\;,
\end{equation*}
and for every $j\geq1$, the $j$-th summand satisfies
\begin{equation}\label{eq:d2-log-bound}
    0 \;\le\;
    \log\E_{\mu_{j-1}}\big[e^{\Psi_j}\big]
    \;\le\; C_T\,2^{-2j}\,j^3\;.
\end{equation}
Since $\sum_{j\ge1}2^{-2j}\,j^3<\infty$, the sequence $\log Z_\Lambda$ converges as $\Lambda\to\infty$.
\end{theorem}

\begin{proof}
The displayed sum is \Cref{prop:telescope} after taking logarithms.  The lower bound in~\eqref{eq:d2-log-bound} follows from $\Psi_j\ge 0$ (\Cref{lem:positivity}).  The upper bound is \Cref{prop:d2-exp-int} together with $\log(1+u)\le u$ for $u\ge 0$.
\end{proof}

We write $d_{\mathrm{TV}}(\mu,\mu'):=\sup_A|\mu(A)-\mu'(A)|$ for the total variation distance between probability measures, with the supremum over measurable sets $A$.  When $\mu$ has Radon--Nikodym derivative $\mathcal{R}$ with respect to $\mu'$, Scheff\'e's theorem \cite[Lem.~2.1]{Tsybakov2009} gives
\begin{equation}\label{eq:dTV-L1}
    d_{\mathrm{TV}}(\mu,\mu')
    \;=\;\frac{1}{2}\int\big|\mathcal{R}-1\big|\,d\mu'\;.
\end{equation}
We write $\langle \psi,f\rangle:=\int_{\R^2}\psi(x)f(x)\,dx$ for the spatial pairing of a field $\psi$ with a Schwartz test function $f\in\cS(\R^2)$, read as the duality pairing when $\psi\in\cS'(\R^2)$.

\begin{theorem}[Existence of the finite-width $\phi^4_2$ measure]
\label{thm:d2-existence}
Fix a spatial cutoff $g\in C_c^\infty(\R^2;[0,\infty))$, not identically zero, with $K=\supp g$, and a width sequence $(N_j)_{j\ge 0}$ with $N_j\ge c_0\,(j{+}1)\,2^{4j}$ as in \eqref{eq:width-gen}. For $\Lambda\in\N$, let $Z_\Lambda$ and $\mu_\Lambda$ be the partition function and Gibbs measure of~\eqref{eq:Z-def}--\eqref{eq:mu-def} at shell cutoff $\Lambda$.  Then the following hold for almost every draw of the features:
\begin{enumerate}[label=(\alph*)]
\item The partition function $Z_\Lambda$ converges to a limit $Z_\infty\in[1,\infty)$.

\item For integers $0\le L\le\Lambda$, let $\mu_{\Lambda}^{\le L}$ denote the marginal of $\mu_\Lambda$ on the shell coefficients $(a_0,\ldots,a_L)$.  For each $\Lambda'>\Lambda\geq L$,
\begin{equation}\label{eq:d_TV}
    d_{\mathrm{TV}}\big(\mu_{\Lambda'}^{\le L},
    \,\mu_{\Lambda}^{\le L}\big)
    \;\le\; C_T\!\sum_{j=\Lambda+1}^{\Lambda'}
    2^{-2j}\,j^3\;,
\end{equation}
and in particular, for each $L$, these marginals are Cauchy in total variation as $\Lambda\to\infty$.  There is a unique probability measure $\mu_\infty$ on $\prod_{j\ge 0}\R^{N_j}$ such that
\begin{equation}\label{eq:mu-infty}
    \mu_\infty^{\le L}
    \;=\;\lim_{\Lambda\to\infty}\mu_\Lambda^{\le L}
\end{equation}
in total variation for every $L$.

\item The field $\phi_\Lambda=\sum_{j=0}^\Lambda\varphi_j$ converges in distribution when smeared against test functions.  For any $f_1,\ldots,f_r\in\cS(\R^2)$ and any bounded continuous $F:\R^r\to\R$, as $\Lambda\to\infty$,
\begin{equation}\label{eq:joint-conv}
    \E_{\mu_\Lambda}\big[F\big(
    \langle\phi_\Lambda,f_1\rangle,\ldots,
    \langle\phi_\Lambda,f_r\rangle\big)\big]
    \;\longrightarrow\;
    \E_{\mu_\infty}\big[F\big(
    \langle\phi_\infty,f_1\rangle,\ldots,
    \langle\phi_\infty,f_r\rangle\big)\big]\;,
\end{equation}
where the series $\langle\phi_\infty,f\rangle:=\sum_{j\ge0}\langle\varphi_j,f\rangle$ converges absolutely $\mu_\infty$-almost surely.
\end{enumerate}
\end{theorem}

\begin{proof}
\textit{Part (a).} By \Cref{thm:d2-convergence}, $Z_\Lambda$ converges to a finite limit $Z_\infty$, and $Z_\infty\ge1$ because $Z_\Lambda\ge1$ (\Cref{lem:V-lower}).

\medskip\noindent\textit{Part (b).}
For any bounded $\eta=\eta(a_0,\ldots,a_J)$, integrating out $a_{J+1}$ as in the proof of \Cref{prop:telescope} and dividing by $Z_{J+1}=Z_J\,\E_{\mu_J}[e^{\Psi_{J+1}}]$ gives
\begin{equation*}
    \E_{\mu_{J+1}}[\eta]
    \;=\; \E_{\mu_J}\big[\eta\,\mathcal{R}_{J+1}\big]\;,
    \qquad
    \mathcal{R}_{J+1} \;:=\;
    \frac{e^{\Psi_{J+1}}}{\E_{\mu_J}[e^{\Psi_{J+1}}]}\;.
\end{equation*}
Hence $\mu_{J+1}^{\le J}$ has Radon--Nikodym derivative $\mathcal{R}_{J+1}$ with respect to $\mu_J$.  Since $\Psi_{J+1}\ge 0$ and $\E_{\mu_J}[e^{\Psi_{J+1}}]\ge 1$,
\begin{equation}\label{eq:RN_identity}
    \E_{\mu_J}\big|\mathcal{R}_{J+1}-1\big|
    \;\le\; 2\big(\E_{\mu_J}[e^{\Psi_{J+1}}]-1\big)\;.
\end{equation}
Combining \eqref{eq:dTV-L1}, \eqref{eq:RN_identity}, and \eqref{eq:d2-exp-int}:
\begin{equation*}
    d_{\mathrm{TV}}(\mu_{J+1}^{\le J},\mu_J) =\frac{1}{2}\E_{\mu_J}|\mathcal{R}_{J+1}{-}1| \le C_T\,2^{-2(J+1)}(J{+}1)^3\;.
\end{equation*}
For $L\leq J$, marginalizing to the first $L+1$ shells does not increase total variation distance because an event in those shells is also an event in the first $J+1$ shells.  The same bound therefore holds for $d_{\mathrm{TV}}\big(\mu_{J+1}^{\le L},\mu_J^{\le L}\big)$, and summing over $J=\Lambda,\ldots,\Lambda'-1$ with the triangle inequality gives~\eqref{eq:d_TV}. The right side of \eqref{eq:d_TV} is at most $C_T$ times the tail $\sum_{j>\Lambda}2^{-2j}j^3$, which vanishes as $\Lambda\to\infty$, so the marginals are Cauchy in total variation.

Fix $L$ and represent each marginal by its density against the free measure $\gamma_L$.  Integrating out the shells above $L$ gives
\begin{equation}\label{eq:marg-density}
    \frac{d\mu_\Lambda^{\le L}}{d\gamma_L}
    \;=\;\frac{1}{Z_\Lambda}\int e^{-V_\Lambda}\,
    d\gamma_\sigma(a_{L+1})\cdots d\gamma_\sigma(a_\Lambda)\;.
\end{equation}
By Scheff\'e's theorem,
\begin{equation}\label{eq:marg-L1}
    d_{\mathrm{TV}}\big(\mu_{\Lambda'}^{\le L},\mu_{\Lambda}^{\le L}\big)
    \;=\;\frac{1}{2}
    \bigg\|\frac{d\mu_{\Lambda'}^{\le L}}{d\gamma_L}
    -\frac{d\mu_{\Lambda}^{\le L}}{d\gamma_L}
    \bigg\|_{L^1(\gamma_L)}\;.
\end{equation}
Since the marginals are Cauchy in total variation, the densities are Cauchy in $L^1(\gamma_L)$, and by completeness \cite[Lem.~1.33]{Kallenberg2021} they converge in $L^1(\gamma_L)$.  Let $\hat\mu_L$ be the measure whose density against $\gamma_L$ is their limit.  Because nonnegativity and integration to one pass to $L^1$ limits, $\hat\mu_L$ is a probability measure.  Scheff\'e's theorem applied to $\mu_\Lambda^{\le L}$ and $\hat\mu_L$ then gives $\mu_\Lambda^{\le L}\to\hat\mu_L$ in total variation.

Fix $L'\ge L$.  At every cutoff, marginalizing $\mu_\Lambda^{\le L'}$ onto the first $L+1$ shells returns $\mu_\Lambda^{\le L}$, and marginalizing does not increase total variation distance, so
\begin{equation}\label{eq:marg-consist}
    d_{\mathrm{TV}}\big(\mu_\Lambda^{\le L},\,(\hat\mu_{L'})^{\le L}\big)
    \;\le\; d_{\mathrm{TV}}\big(\mu_\Lambda^{\le L'},\,\hat\mu_{L'}\big)
    \;\xrightarrow[\Lambda\to\infty]{}\;0\;.
\end{equation}
The marginals $\mu_\Lambda^{\le L}$ therefore converge to both $\hat\mu_L$ and $(\hat\mu_{L'})^{\le L}$.  Limits in total variation are unique, so $(\hat\mu_{L'})^{\le L}=\hat\mu_L$, and the family $(\hat\mu_L)_{L\ge0}$ is consistent.

Finally, by Daniell's theorem \cite[Thm.~8.21]{Kallenberg2021}, there is a probability measure $\mu_\infty$ on $\prod_{j\ge0}\R^{N_j}$ with $\mu_\infty^{\le L}=\hat\mu_L$ for every $L$, and these marginals determine $\mu_\infty$ uniquely \cite[Lem.~1.18]{Kallenberg2021}.

\medskip\noindent\textit{Part (c).}
Fix $f\in\cS(\R^2)$, write $w_j:=\langle\varphi_j,f\rangle$ for each shell, and let $\mathcal{X}_\Lambda(s):=\E_{\mu_\Lambda}[e^{is\langle\phi_\Lambda,f\rangle}]$ for $s\in\R$.  For bounded $\eta=\eta(a_0,\ldots,a_\Lambda)$, Scheff\'e's theorem and~\eqref{eq:d_TV} with $L=\Lambda$ and $\Lambda'\to\infty$ give
\begin{align}
    \big|\E_{\mu_\Lambda}[\eta]-\E_{\mu_\infty}[\eta]\big|
    &\;\le\;\sup|\eta|\,\bigg\|\frac{d\mu_\Lambda}{d\gamma_\Lambda}-\frac{d\mu_\infty^{\le\Lambda}}{d\gamma_\Lambda}\bigg\|_{L^1(\gamma_\Lambda)}\notag\\[5 pt]
    &\;=\;2\sup|\eta|\,d_{\mathrm{TV}}\big(\mu_\infty^{\le\Lambda},\mu_\Lambda\big)
    \;\le\;2C_T\sup|\eta|\sum_{j>\Lambda}2^{-2j}j^3\;.\label{eq:TV-tail}
\end{align}
As in~\eqref{eq:d2-ratio}, and because $\E_{\mu_{j-1}}[e^{\Psi_j}]\ge1$, every bounded $\eta=\eta(a_0,\ldots,a_j)\ge0$ satisfies
\begin{equation}\label{eq:top-shell}
    \E_{\mu_j}[\eta]\;\le\;\E_{\mu_{j-1}}\big[e^{\Psi_j}\langle \eta\rangle\big]\;,
\end{equation}
where $\langle\cdot\rangle$ is the tilted expectation over $a_j$.

Fix $j>j_4$, so that \Cref{cor:BL}, \Cref{thm:d2-tilted}, and \Cref{cor:d2-bad-prob} apply at shell $j$.  First, set $v_\ell:=\langle U_{j,\ell},f\rangle$, so that $w_j=v^Ta_j$.  Each $U_{j,\ell}$ is a cosine of frequency $|b_{j,\ell}|\ge2^{j-1}$ times the amplitude $(2q_j/(N_j\sigma^2))^{1/2}$, and the Laplacian acting on the cosine multiplies it by $-|b_{j,\ell}|^2$.  Moving the Laplacian onto $f$ by two integrations by parts, which leave no boundary terms because $f$ is a Schwartz function, gives
\begin{equation}\label{eq:d2-fourier}
    |v_\ell| \;\le\; \frac{C_f}{2^{2j}\sqrt{N_j}}\,,\qquad C_f:=4\sqrt{\frac{2q_{\max}}{\sigma^2}}\,\|\Delta f\|_{L^1}\;.
\end{equation}
Squaring and summing over the $N_j$ features gives $|v|\le C_f2^{-2j}$.

Second, $w_j$ is a linear functional of $a_j$, and by \Cref{cor:BL} the log-density has Hessian at least $\tfrac{1}{2\sigma^2}I_{N_j}$.  The Brascamp--Lieb inequality \cite[Thm.~4.1]{BrascampLieb1976} therefore gives
\begin{equation}\label{eq:var-w}
    \Var(w_j)\;\le\;2\sigma^2|v|^2\;\le\;2\sigma^2C_f^2\,2^{-4j}\;.
\end{equation}

Third, the mean identity~\eqref{eq:mean-exact}, together with $\Gamma^\omega_j(x,y)=\sigma^2\sum_\ell U_{j,\ell}(x)U_{j,\ell}(y)$, gives
\begin{equation}\label{eq:d2-w-mean}
    \langle w_j\rangle
    \;=\;\int f\,\bar{\varphi}_j\,dx
    \;=\; -4t\int g(y)\,\Big(\sigma^2\textstyle\sum_\ell
    v_\ell\,U_{j,\ell}(y)\Big)\,
    \Big\langle H_3\big(B_j(y); D^\omega_j(y)\big)\Big\rangle\,dy\;.
\end{equation}
Cauchy--Schwarz over $\ell$ bounds the parentheses,
\begin{equation*}
    \sigma^2\Big|\sum_\ell v_\ell\,U_{j,\ell}(y)\Big|
    \;\le\;\sigma|v|\sqrt{q^\omega_j(y)}
    \;\le\;\sigma\sqrt{c_q}\,C_f\,2^{-2j}\;.
\end{equation*}
The $\langle H_3\rangle$ factor in the integrand of~\eqref{eq:d2-w-mean} is at most $\langle H_3^2\rangle^{1/2}$ in absolute value, which \Cref{thm:d2-tilted} bounds on $G_j$.

Splitting $w_j$ into its tilted mean and fluctuation gives $\langle|w_j|\rangle\le|\langle w_j\rangle|+\Var(w_j)^{1/2}$.  Averaging under $\mu_{j-1}$, using $e^{\Psi_j}\le e$ on $\mathcal{G}_j$, Jensen's inequality for the square root, and the sixth moment bound $\E_{\mu_{j-1}}[B_{j-1}(y)^6]\le6\,\mathcal{K}C_M^3j^3$ (\Cref{thm:d2-moments}), gives
\begin{equation}\label{eq:d2-good-event}
    \E_{\mu_{j-1}}\big[e^{\Psi_j}\langle|w_j|\rangle\,\mathbf{1}_{\mathcal{G}_j}\big]
    \;\le\;C_f'\,2^{-2j}j^{3/2}\;.
\end{equation}
The proof of \Cref{prop:d2-exp-int} also gives $\E_{\mu_{j-1}}[e^{\Psi_j}\mathbf{1}_{\mathcal{G}_j^c}]\le C_K\,e^{c_V(j+1)^2}\exp(-c_3\,2^{2j/3}/j)$.

Taking $\eta=\min\{|w_j|,1\}$ in~\eqref{eq:top-shell} and bounding $\langle\eta\rangle$ by $\langle|w_j|\rangle$ on $\mathcal{G}_j$ and by $1$ on $\mathcal G_j^c$ gives
\begin{align*}
    \E_{\mu_j}\big[\min\{|w_j|,1\}\big]
    &\;\le\;\E_{\mu_{j-1}}\big[e^{\Psi_j}\langle|w_j|\rangle\,\mathbf{1}_{\mathcal{G}_j}\big]
    +\E_{\mu_{j-1}}\big[e^{\Psi_j}\mathbf{1}_{\mathcal{G}_j^c}\big]\\
    &\;\le\;C_f'\,2^{-2j}j^{3/2}+C_K\,e^{c_V(j+1)^2}\exp\!\big(-c_3\,2^{2j/3}/j\big)\;,
\end{align*}
and replacing $\mu_j$ by $\mu_\infty$ adds at most $2C_T\sum_{k>j}2^{-2k}k^3$, by~\eqref{eq:TV-tail} at $\Lambda=j$.  All three terms are $O(2^{-2j}j^3)$, so
\begin{equation}\label{eq:d2-w-first}
    \E_{\mu_\infty}\big[\min\{|w_j|,1\}\big]\;\le\;C_1\,2^{-2j}j^3\qquad(j>j_4)\;,
\end{equation}
where $C_1$ is built from $C_f'$, $C_K$, $c_V$, $c_3$, and $C_T$.

By monotone convergence and~\eqref{eq:d2-w-first}, $\sum_j\min\{|w_j|,1\}$ has finite expectation under $\mu_\infty$.  Almost surely, therefore, $|w_j|<1$ for all large $j$, and the shell series $\langle\phi_\infty,f\rangle=\sum_jw_j$ converges absolutely.

Fix $s\in\R$ and $\Lambda\ge j_4$, and insert $\E_{\mu_\infty}[e^{is\langle\phi_\Lambda,f\rangle}]$ between $\mathcal{X}_\Lambda(s)$ and $\E_{\mu_\infty}[e^{is\langle\phi_\infty,f\rangle}]$:
\begin{equation}\label{eq:char-split}
    \Big|\mathcal{X}_\Lambda(s)-\E_{\mu_\infty}\big[e^{is\langle\phi_\infty,f\rangle}\big]\Big|
    \le\Big|\E_{\mu_\Lambda}\big[e^{is\langle\phi_\Lambda,f\rangle}\big]-\E_{\mu_\infty}\big[e^{is\langle\phi_\Lambda,f\rangle}\big]\Big|
    +\E_{\mu_\infty}\Big|e^{is\langle\phi_\Lambda,f\rangle}-e^{is\langle\phi_\infty,f\rangle}\Big|\;.
\end{equation}
The exponential $e^{is\langle\phi_\Lambda,f\rangle}$ has modulus one and depends only on the first $\Lambda{+}1$ shells, so~\eqref{eq:TV-tail} bounds the first term by $2C_T\sum_{j>\Lambda}2^{-2j}j^3$.  For the second term, pointwise,
\begin{equation*}
    \big|e^{is\langle\phi_\Lambda,f\rangle}-e^{is\langle\phi_\infty,f\rangle}\big|
    \;\le\;(2+|s|)\sum_{j>\Lambda}\min\{|w_j|,1\}\;.
\end{equation*}
Indeed, the left side is at most $2$ and at most $|s|\,|\langle\phi_\infty,f\rangle-\langle\phi_\Lambda,f\rangle|\le|s|\sum_{j>\Lambda}|w_j|$, while the right side is at least $2$ if some $|w_j|$ exceeds $1$, and at least $|s|\sum_{j>\Lambda}|w_j|$ otherwise.  By~\eqref{eq:d2-w-first}, the second term of \eqref{eq:char-split} is therefore at most $(2+|s|)\,C_1\sum_{j>\Lambda}2^{-2j}j^3$.  Adding the two bounds,
\begin{equation}\label{eq:char-rate}
    \Big|\mathcal{X}_\Lambda(s)-\E_{\mu_\infty}\big[e^{is\langle\phi_\infty,f\rangle}\big]\Big|
    \;\le\;C_{\mathcal{X}}\,(1+|s|)\sum_{j>\Lambda}2^{-2j}j^3\;,
\end{equation}
with $C_{\mathcal{X}}:=2C_T+2C_1$.

Given $f_1,\ldots,f_r\in\cS(\R^2)$ and $s\in\R^r$, the function $h:=\sum_ks_kf_k$ is again Schwartz, and~\eqref{eq:char-rate} for $h$ at the point $1$ gives $\E_{\mu_\Lambda}[e^{i\langle\phi_\Lambda,h\rangle}]\to\E_{\mu_\infty}[e^{i\langle\phi_\infty,h\rangle}]$.  Since the pairing is linear and the shell series converge absolutely, $\langle\phi_\Lambda,h\rangle=\sum_ks_k\langle\phi_\Lambda,f_k\rangle$, and almost surely $\langle\phi_\infty,h\rangle=\sum_ks_k\langle\phi_\infty,f_k\rangle$.  The characteristic functions of $(\langle\phi_\Lambda,f_k\rangle)_{k\le r}$ under $\mu_\Lambda$ therefore converge pointwise to that of $(\langle\phi_\infty,f_k\rangle)_{k\le r}$ under $\mu_\infty$, and the L\'evy continuity theorem \cite[Thm.~6.23]{Kallenberg2021} gives~\eqref{eq:joint-conv}.
\end{proof}

\section{The infinite-width limit and the standard measure}
\label{sec:d2-standard}

The field theory requires both the ultraviolet limit and the infinite-width limit.  \Cref{thm:d2-existence} removes the ultraviolet cutoff at a fixed width sequence.  This section lets the width grow at a fixed cutoff, shows that the two limits commute, and identifies their common value with the standard finite-volume $\phi^4_2$ measure.

\subsection{The width limit at fixed cutoff}\label{sec:infwidth}

The argument of this subsection works in every dimension $d\ge2$, and we give it in that generality.

The measures $\mu_\Lambda^{(n)}$ live on coefficient spaces $\prod_{j\le\Lambda}\R^{N_j^{(n)}}$ of different dimensions for different $n$, so any statement comparing them must first place them on a common space.  We compare them in field space.  Let $\cC(\R^d)$ be the space of continuous real-valued functions on $\R^d$ with the topology of uniform convergence on compact sets, a Polish space, and let
\begin{equation}\label{eq:pushforward}
    \iota_\Lambda^{(n)}:\;(a_0,\dots,a_\Lambda)\;\longmapsto\;
    \phi_\Lambda^{(n)}\;\in\;\cC(\R^d)\;.
\end{equation}
For each $n$, write $P_\Lambda^{(n)}:=\mu_\Lambda^{(n)}\circ(\iota_\Lambda^{(n)})^{-1}$ and $P^{0,(n)}_\Lambda:=\gamma_\Lambda^{(n)}\circ(\iota_\Lambda^{(n)})^{-1}$ for the laws of the interacting and free fields on $\cC(\R^d)$.  Write $P^{0,(\infty)}_\Lambda$ for the centered Gaussian law on $\cC(\R^d)$ with covariance $Q_{\le\Lambda}$.

The law $P^{0,(\infty)}_\Lambda$ exists because a centered Gaussian process $Y$ on $\R^d$ with covariance $Q_{\le\Lambda}$ has a continuous version.  Writing the covariance in momentum space and using $1-\cos u\le\tfrac12u^2$,
\begin{align}
    \Var\big(Y(x)-Y(y)\big)
    &\;=\;2\int_{\R^d}\frac{d^dp}{(2\pi)^d}
    \Big(\sum_{j\le\Lambda}S_j(p)\Big)
    \big(1-\cos\big(p\cdot(x-y)\big)\big)
    \notag\\
    &\;\le\;|x-y|^2\int_{\R^d}\frac{d^dp}{(2\pi)^d}\,|p|^2
    \sum_{j\le\Lambda}S_j(p)\;,
    \label{eq:trunc-incr}
\end{align}
and the last integral is finite because the truncated spectral density is bounded and supported in $|p|\le2^{\Lambda+1}$. The increments are Gaussian, so the moment of order $2p$ of $Y(x)-Y(y)$ is at most $(2p-1)!!$ times the $p$-th power of the right side of~\eqref{eq:trunc-incr}.  For any $p$ with $2p>d$, the Kolmogorov--Chentsov criterion \cite[Thm.~4.23]{Kallenberg2021}, with exponents $a=2p$ and $b=2p-d$ in its notation, gives a continuous version of $Y$.

The interaction $V_\Lambda^{(n)}$ depends on the coefficients only through the values of the field $\phi_\Lambda^{(n)}$ on $K$, so it can be written in field space.  For $\phi\in\cC(\R^d)$ and continuous $D\ge0$ on $K$, set
\begin{equation}\label{eq:V-field}
    V(\phi;D) \;:=\; t\!\int\! g\,H_4(\phi;D)\,dx\;,
\end{equation}
so that $V^{(n)}_\Lambda=V\big(\,\cdot\,;D^{\omega^{(n)}}_\Lambda\big)\circ\iota^{(n)}_\Lambda$.

\begin{lemma}[Free-field convergence in $\cC(\R^d)$]\label{lem:CK-conv}
For each fixed $\Lambda$, the laws $P^{0,(n)}_\Lambda$ are tight in $\cC(\R^d)$, and almost surely in $\omega^\infty$ they converge weakly to $P^{0,(\infty)}_\Lambda$ as $n\to\infty$.
\end{lemma}

\begin{proof}
Under $\gamma_\Lambda^{(n)}$ the field $\phi_\Lambda^{(n)}$ is a centered Gaussian process on $\R^d$, continuous because it is a finite sum of cosines.  Its pointwise variance is $D^{\omega^{(n)}}_\Lambda(x)\le2D_\Lambda$, and by \Cref{lem:incr-var} its increments have variance at most $\Upsilon_\Lambda|x-y|^2$.  Neither bound involves $n$.  Since the increments are Gaussian, the second bound controls every moment:
\begin{equation}\label{eq:CK-moments}
    \sup_n\,\E_{\gamma^{(n)}_\Lambda}\big|\phi_\Lambda^{(n)}(x)-\phi_\Lambda^{(n)}(y)\big|^{2p}
    \;\le\;(2p-1)!!\,\Upsilon_\Lambda^{p}\,|x-y|^{2p},
    \qquad x,y\in\R^d,\quad p\ge1\;.
\end{equation}
Fix $p$ with $2p>d$.  The values at a single point are tight, because their variances are at most $2D_\Lambda$, and~\eqref{eq:CK-moments} is the moment condition of \cite[Thm.~23.7]{Kallenberg2021} with exponents $a=2p$ and $b=2p-d$.  By that theorem, the family $\{P^{0,(n)}_\Lambda\}_n$ is tight in $\cC(\R^d)$.

It remains to identify the limit, first for the finite-dimensional distributions.  Fix $x_1,\dots,x_r\in\R^d$.  The vector $(\phi_\Lambda^{(n)}(x_i))_{i\le r}$ is centered Gaussian with covariance $\sum_{j\le\Lambda}\Gamma_j^{\omega^{(n)}}(x_i,x_{i'})$, which converges to $Q_{\le\Lambda}(x_i-x_{i'})$ by \Cref{prop:kernel-conv}.  The characteristic function of a centered Gaussian vector depends continuously on its covariance, so the L\'evy continuity theorem \cite[Thm.~6.23]{Kallenberg2021} gives convergence of the finite-dimensional distributions of $P^{0,(n)}_\Lambda$ to those of $P^{0,(\infty)}_\Lambda$.

Together with tightness, convergence of the finite-dimensional distributions gives $P^{0,(n)}_\Lambda\to P^{0,(\infty)}_\Lambda$ weakly in $\cC(\R^d)$ \cite[Thm.~23.4 and Cor.~23.5]{Kallenberg2021}.
\end{proof}

\Cref{lem:CK-conv} gives the limit of the free laws, but the interacting measure is the free law reweighted by a density that itself depends on $n$ through the quenched Wick constant $D^{\omega^{(n)}}_\Lambda$. Passing to the limit therefore needs a form of weak convergence in which the integrand varies along with the measure.

\begin{lemma}[Continuous convergence]\label{lem:cont-conv}
Let $S$ be a metric space, let $P_n$ converge weakly to $P$ on $S$, and let $h_n,\,h:S\to\R$ satisfy $\sup_n\|h_n\|_\infty<\infty$, $h$ continuous, and $h_n(x_n)\to h(x)$ whenever $x_n\to x$ in $S$.  Then $\int h_n\,dP_n\to\int h\,dP$.
\end{lemma}

\begin{proof}
Let $X_n\sim P_n$ and $X\sim P$, so that $\int h_n\,dP_n=\E[h_n(X_n)]$ and $\int h\,dP=\E[h(X)]$.  The hypothesis on the $h_n$ is that of the continuous mapping theorem \cite[Thm.~5.27]{Kallenberg2021}, which gives $h_n(X_n)\to h(X)$ in distribution.  These variables are bounded by $\sup_n\|h_n\|_\infty$, which also bounds $h$, and for uniformly bounded variables convergence in distribution gives convergence of expectations.
\end{proof}

\begin{theorem}[Infinite-width limit at fixed cutoff]
\label{thm:infwidth}
Almost surely in $\omega^\infty$, for each fixed $\Lambda\in\N$ the laws $P^{(n)}_\Lambda$ converge weakly in $\cC(\R^d)$ as $n\to\infty$ to
\begin{equation}\label{eq:mu-std-def}
    dP^{(\infty)}_\Lambda
    \;:=\;\frac{e^{-V(\phi;\,D_\Lambda)}}
    {Z^{(\infty)}_\Lambda}\,dP^{0,(\infty)}_\Lambda\;,
    \qquad
    Z^{(\infty)}_\Lambda
    \;:=\;\int e^{-V(\phi;\,D_\Lambda)}\,
    dP^{0,(\infty)}_\Lambda\;,
\end{equation}
the Gibbs measure of the truncated free field.
\end{theorem}

\begin{proof}
Let $G:\cC(\R^d)\to\R$ be bounded and continuous, and set $\tilde h_n(\phi):=e^{-V(\phi;\,D^{\omega^{(n)}}_\Lambda)}$, so that the Gibbs density in~\eqref{eq:mu-def} is $e^{-V^{(n)}_\Lambda}/Z^{(n)}_\Lambda=\big(\tilde h_n\circ\iota^{(n)}_\Lambda\big)/Z^{(n)}_\Lambda$.  The substitution rule \cite[Lem.~1.24]{Kallenberg2021} gives the first and last equalities below,
\begin{align*}
    \E_{P^{(n)}_\Lambda}[G]
    &\;=\;\E_{\mu^{(n)}_\Lambda}\big[G\circ\iota^{(n)}_\Lambda\big]
    \;=\;\frac{1}{Z^{(n)}_\Lambda}\,\E_{\gamma^{(n)}_\Lambda}\big[(G\,\tilde h_n)\circ\iota^{(n)}_\Lambda\big]\\
    &\;=\;\frac{1}{Z^{(n)}_\Lambda}\int G\,\tilde h_n\,dP^{0,(n)}_\Lambda\;.
\end{align*}
That is, $P^{(n)}_\Lambda$ has density $\tilde h_n/Z^{(n)}_\Lambda$ against $P^{0,(n)}_\Lambda$, the finite-width form of~\eqref{eq:mu-std-def}.  At $G\equiv1$ this reads $Z^{(n)}_\Lambda=\int\tilde h_n\,dP^{0,(n)}_\Lambda$, so with $h_n:=G\,\tilde h_n$,
\begin{equation}\label{eq:ratio-form}
    \E_{P^{(n)}_\Lambda}[G]
    \;=\;\frac{\int h_n\,dP^{0,(n)}_\Lambda}
    {\int\tilde h_n\,dP^{0,(n)}_\Lambda}\;.
\end{equation}
We verify the hypotheses of \Cref{lem:cont-conv} for both integrands, with $S=\cC(\R^d)$, $P_n=P^{0,(n)}_\Lambda$, and $P=P^{0,(\infty)}_\Lambda$, which converge weakly by \Cref{lem:CK-conv}.

\textit{Uniform boundedness.}  By~\eqref{eq:H4-lower}, $V(\phi;D^{\omega^{(n)}}_\Lambda)\ge-6t\|g\|_1\|D^{\omega^{(n)}}_\Lambda\|_\infty^2$ for every $\phi\in\cC(\R^d)$, and $\|D^{\omega^{(n)}}_\Lambda\|_\infty\le2D_\Lambda$ because $q^\omega_j\le2q_j$.  Hence, for every $n$,
\begin{equation*}
    \|h_n\|_\infty \;\le\; \|G\|_\infty\,
    \exp\big(24t\|g\|_1D_\Lambda^2\big)\;.
\end{equation*}

\textit{Continuous convergence.}  Let $\phi_n\to\phi$ in $\cC(\R^d)$, so that $\phi_n\to\phi$ uniformly on $K$.  By \Cref{prop:kernel-conv}, the Wick constants $D^{\omega^{(n)}}_\Lambda$ also converge to $D_\Lambda$ uniformly on $K$.  For $x\in K$, the pairs $(\phi_n(x),D^{\omega^{(n)}}_\Lambda(x))$ stay in a fixed compact set, and $(u,D)\mapsto H_4(u;D)$ is uniformly continuous there.  Hence $H_4(\phi_n;D^{\omega^{(n)}}_\Lambda)\to H_4(\phi;D_\Lambda)$ uniformly on $K$, and integrating against the finite measure $g\,dx$ gives $V(\phi_n;D^{\omega^{(n)}}_\Lambda)\to V(\phi;D_\Lambda)$.  Since $G$ is continuous, it follows that $h_n(\phi_n)\to h(\phi)$ and $\tilde h_n(\phi_n)\to\tilde h(\phi)$, where
\begin{equation*}
    h(\phi):=G(\phi)\,e^{-V(\phi;\,D_\Lambda)}\;,
    \qquad
    \tilde h(\phi):=e^{-V(\phi;\,D_\Lambda)}
\end{equation*}
are continuous on $\cC(\R^d)$.

\Cref{lem:cont-conv} thus applies to the numerator and the denominator of~\eqref{eq:ratio-form}.  The denominator converges to $Z^{(\infty)}_\Lambda$, which is at least $1$ because every $Z^{(n)}_\Lambda$ is (\Cref{lem:V-lower}). The ratio therefore converges, and
\begin{equation}
    \E_{P^{(n)}_\Lambda}[G]\;\to\;\E_{P^{(\infty)}_\Lambda}[G]\;.
\end{equation}
Since $G$ was an arbitrary bounded continuous functional on $\cC(\R^d)$, the laws $P^{(n)}_\Lambda$ converge weakly to $P^{(\infty)}_\Lambda$.
\end{proof}

\subsection{Interchange of the UV and infinite-width limits}
\label{sec:interchange}

\Cref{sec:d2-setup,sec:d2-converge} work at a single feature realization, but the interchange argument of \Cref{thm:interchange} needs their estimates at every network of \Cref{def:nested}, with one set of constants.

\begin{lemma}[Width-uniformity of the constants]
\label{prop:uniform-N}
Work on the almost-sure event of \Cref{lem:unif-gg}.  Then \Cref{sec:d2-setup,sec:d2-converge} apply verbatim to every network of \Cref{def:nested}, with $j_0$ read as $j_0^*$, and their constants and thresholds do not depend on $n$.  In particular, the constant $C_{\mathcal{X}}$ of~\eqref{eq:char-rate} is free of $n$.
\end{lemma}

\begin{proof}
The features enter the estimates of \Cref{sec:one-shell,sec:d2-setup,sec:d2-converge} in only two ways.  The first is through the bounds of \Cref{thm:good-geom} above the threshold $j_0$, and \Cref{lem:unif-gg} gives these for every $n$ with the same threshold $j_0^*$ and the same constant $\kappa$.  The second is through bounds that hold for every realization, namely $q^\omega_j\le2q_j$ and the fact that every frequency of shell $j\ge1$ lies in the annulus $2^{j-1}\le|p|\le2^{j+1}$.  Every other constant and threshold in these sections is built from these bounds, from $j_0^*$, and from the fixed data, together with $\|\Delta f\|_{L^1}$ for the constant of~\eqref{eq:char-rate}.  None of them depends on $n$.
\end{proof}

At fixed cutoff, \Cref{thm:infwidth} gives weak convergence on $\cC(\R^2)$ of the laws of the network fields $\phi_\Lambda^{(n)}$, and hence convergence in law of every continuous function of those fields \cite[Thm.~5.27]{Kallenberg2021}.  The map $\phi\mapsto(\langle\phi,f_k\rangle)_{k\le r}$ is continuous on $\cC(\R^2)$ when every $f_k$ has compact support, but not for general Schwartz functions, because convergence in $\cC(\R^2)$ is only uniform on compact sets.

\begin{lemma}[Width limit of the smeared field]\label{lem:width-smeared}
Almost surely in $\omega^\infty$, for every $\Lambda\in\N$, every $r$, and all $f_1,\dots,f_r\in\cS(\R^2)$, the law of $\big(\langle\phi_\Lambda^{(n)},f_k\rangle\big)_{k\le r}$ under $\mu^{(n)}_\Lambda$ converges weakly on $\R^r$, as $n\to\infty$, to the law of $\big(\langle\phi,f_k\rangle\big)_{k\le r}$ under $P^{(\infty)}_\Lambda$.
\end{lemma}

\begin{proof}
Work on the probability-one event of \Cref{thm:infwidth}, fix $\Lambda$ and $s\in\R^r$, and set $h:=\sum_ks_kf_k$ and $h_R:=h\,\mathbf{1}_{\{|x|\le R\}}$.  Because convergence in $\cC(\R^2)$ is uniform on compact sets, the map $\phi\mapsto e^{i\langle\phi,h_R\rangle}$ is bounded and continuous.  \Cref{thm:infwidth} therefore gives convergence of the characteristic functions at $s$ with $h_R$ in place of $h$.  By \Cref{lem:V-lower}, Cauchy--Schwarz under the Gaussian prior, and $q^\omega_j\le2q_j$, the first moments $\E|\phi(x)|$ are bounded by a constant free of $x$ and $n$, under every $\mu^{(n)}_\Lambda$ and under $P^{(\infty)}_\Lambda$.  Since $|e^{ia}-e^{ib}|\le|a-b|$, replacing $h$ by $h_R$ therefore changes each characteristic function by at most a constant times $\int_{|x|>R}|h|\,dx$, uniformly in $n$.  Letting $n\to\infty$ and then $R\to\infty$ gives convergence of the characteristic functions at $s$, and the L\'evy continuity theorem \cite[Thm.~6.23]{Kallenberg2021} gives the claimed weak convergence.
\end{proof}

For Schwartz test functions, each limit now exists with the other parameter held fixed, the width limit by \Cref{lem:width-smeared} and the cutoff limit by \Cref{thm:d2-existence}(c).  Because the bound~\eqref{eq:char-rate} makes the cutoff limit uniform in the width (\Cref{prop:uniform-N}), the Moore--Osgood theorem gives both iterated limits and shows that they agree.

\begin{theorem}[Commutativity of limits]\label{thm:interchange}
Almost surely in $\omega^\infty$, for every $r$ and all $f_1,\dots,f_r\in\cS(\R^2)$, the double sequence
\begin{equation*}
    \mathcal{X}^{(n)}_\Lambda(s)
    \;:=\;\E_{\mu^{(n)}_\Lambda}
    \Big[\exp\Big(i\textstyle\sum_k s_k
    \langle\phi_\Lambda^{(n)},f_k\rangle\Big)\Big],
    \qquad s\in\R^r\;,
\end{equation*}
has both iterated limits, and they agree:
\begin{equation}\label{eq:interchange}
    \lim_{\Lambda\to\infty}\lim_{n\to\infty}\mathcal{X}^{(n)}_\Lambda(s)
    \;=\;\lim_{n\to\infty}\lim_{\Lambda\to\infty}\mathcal{X}^{(n)}_\Lambda(s)
    \;=:\;\mathcal{X}^{(\infty)}_\infty(s)\;.
\end{equation}
The common value is continuous on $\R^r$ with $\mathcal{X}^{(\infty)}_\infty(0)=1$.  Along either order of limits the law of the smeared vector $\big(\langle\phi_\Lambda^{(n)},f_k\rangle\big)_{k\le r}$ under $\mu^{(n)}_\Lambda$ therefore converges weakly on $\R^r$ to the law with characteristic function $\mathcal{X}^{(\infty)}_\infty$.
\end{theorem}

\begin{proof}
Work on the intersection of the probability-one events of \Cref{lem:unif-gg} and \Cref{thm:infwidth}, and fix $r$ and $f_1,\dots,f_r$. At each cutoff, \Cref{lem:width-smeared} gives the width limit $\mathcal{X}^{(n)}_\Lambda\to\mathcal{X}^{(\infty)}_\Lambda$, where $\mathcal{X}^{(\infty)}_\Lambda$ is the characteristic function of $(\langle\phi,f_k\rangle)_{k\le r}$ under $P^{(\infty)}_\Lambda$.  For the cutoff limit at each width, write $\mu^{(n)}_\infty$ for the measure $\mu_\infty$ of \Cref{thm:d2-existence} at the widths~\eqref{eq:width-schedule}, and $\langle\phi_\infty^{(n)},f\rangle$ for the corresponding series $\langle\phi_\infty,f\rangle$.  Part~(c) of that theorem gives $\mathcal{X}^{(n)}_\Lambda\to\mathcal{X}^{(n)}_\infty$ as $\Lambda\to\infty$, where $\mathcal{X}^{(n)}_\infty$ is the characteristic function of $(\langle\phi_\infty^{(n)},f_k\rangle)_{k\le r}$ under $\mu^{(n)}_\infty$.  The cutoff limit is uniform in the width.  Indeed, $\mathcal{X}^{(n)}_\Lambda(s)$ is the characteristic function of $\langle\phi_\Lambda^{(n)},h\rangle$ at $1$, with $h:=\sum_ks_kf_k$, so~\eqref{eq:char-rate} applied to $h$ gives, for $\Lambda\ge j_4$,
\begin{equation}\label{eq:unif-width}
    \big|\mathcal{X}^{(n)}_\Lambda(s)-\mathcal{X}^{(n)}_\infty(s)\big|
    \;\le\; 2C_{\mathcal{X}}\sum_{j>\Lambda}2^{-2j}j^3\;,
\end{equation}
in which $C_{\mathcal{X}}$, computed for $h$, does not depend on $n$ (\Cref{prop:uniform-N}) and is bounded for $s$ in compact subsets of $\R^r$.

The Moore--Osgood theorem \cite[Thm.~7.11]{Rudin1976} therefore applies, with $\Lambda$ as its sequence index and $1/n$ ranging over $\{1/n:n\ge1\}$, which has $0$ as a limit point.  Its uniformity hypothesis is~\eqref{eq:unif-width}, and its limit hypothesis is the width limit $\mathcal{X}^{(n)}_\Lambda\to\mathcal{X}^{(\infty)}_\Lambda$.  The theorem then yields~\eqref{eq:interchange}, with common value $\mathcal{X}^{(\infty)}_\infty(s)=\lim_\Lambda\mathcal{X}^{(\infty)}_\Lambda(s)=\lim_n\mathcal{X}^{(n)}_\infty(s)$.

Letting $n\to\infty$ in~\eqref{eq:unif-width} preserves the bound, so $\mathcal{X}^{(\infty)}_\Lambda\to\mathcal{X}^{(\infty)}_\infty$ uniformly on compact subsets of $\R^r$.  Each $\mathcal{X}^{(\infty)}_\Lambda$ is a characteristic function and hence continuous, so the limit is continuous. Furthermore, $\mathcal{X}^{(\infty)}_\infty(0)=1$ because $\mathcal{X}^{(n)}_\Lambda(0)=1$ for every $n$ and $\Lambda$.  The L\'evy continuity theorem \cite[Thm.~6.23]{Kallenberg2021}, applied to $\mathcal{X}^{(\infty)}_\Lambda\to\mathcal{X}^{(\infty)}_\infty$ and to $\mathcal{X}^{(n)}_\infty\to\mathcal{X}^{(\infty)}_\infty$, gives in each case weak convergence on $\R^r$ to the law with characteristic function $\mathcal{X}^{(\infty)}_\infty$.
\end{proof}

\subsection{Identification with the standard measure}
\label{sec:d2-identify}

\Cref{thm:interchange} shows that the two limits agree, but not which measure their common value describes.  The classical constructions regularize by mollification or by a lattice, and the Littlewood--Paley truncation used here is neither.  The identification below uses two properties of the truncation.  Its covariance converges to $Q$ in $L^4(\R^2)$ at rate $2^{-\Lambda/2}$, and the exponential of the interaction is bounded in $L^2$ uniformly in the cutoff, by the partition-function bound of \Cref{thm:d2-convergence} at coupling $2t$.  Glimm and Jaffe prove the corresponding two estimates, convergence of the Wick-ordered interaction and a uniform bound on its exponential, for the mollified field \cite[Prop.~8.5.1]{GlimmJaffe1987} and for the lattice \cite[\S9.6]{GlimmJaffe1987}.

To avoid confusion with the neural network fields, everything belonging to the free field of covariance $Q$ is written upright.  We write $\upphi_Q$ for the Gaussian measure on $\cS'(\R^2)$ with covariance $Q=(-\Delta+m^2)^{-1}$ and realize its field as
\begin{equation}\label{eq:phi-shells}
    \upphi \;=\; \sum_{j\ge0}\upphi^{(j)}\;,
    \qquad
    \upphi_{\le\Lambda} \;:=\; \sum_{j\le\Lambda}\upphi^{(j)}\;,
\end{equation}
the $\upphi^{(j)}$ independent centered Gaussian fields with covariances $Q_j$.  The sum has covariance $Q=\sum_{j\ge0}Q_j$ and hence law $\upphi_Q$.  \Cref{prop:gff}(a) constructs this measure, and \Cref{prop:gff}(b) gives the decomposition.  The partial sum $\upphi_{\le\Lambda}$ is a centered Gaussian process with covariance $Q_{\le\Lambda}$, so~\eqref{eq:trunc-incr} and the Kolmogorov--Chentsov criterion \cite[Thm.~4.23]{Kallenberg2021} give it a continuous version, whose law on $\cC(\R^2)$ is $P^{0,(\infty)}_\Lambda$.  We accordingly read $P^{(\infty)}_\Lambda$ as a law for $\upphi_{\le\Lambda}$.  Since $\upphi_{\le\Lambda}(x)$ has variance $D_\Lambda$ at every $x$, $\E\int|\upphi_{\le\Lambda}f|\,dx\le D_\Lambda^{1/2}\|f\|_{L^1}$, and the pairing $\langle\upphi_{\le\Lambda},f\rangle$ converges absolutely almost surely.

\begin{theorem}[Identification with the standard $\phi^4_2$ measure]
\label{thm:LP-standard}
Let $d=2$ and $t>0$, and write $V(\,\cdot\,;\,\cdot\,)$ for the field-space interaction~\eqref{eq:V-field}.  Then:
\begin{enumerate}[label=(\alph*)]
\item $Q\in L^4(\R^2)$ and $\|Q-Q_{\le\Lambda}\|_{L^4(\R^2)}\le c\,2^{-\Lambda/2}$.  Consequently the interaction $V(\upphi_{\le\Lambda};D_\Lambda)$ converges in $L^2(d\upphi_Q)$, and we write
\begin{equation}\label{eq:wick-limit}
    V\big(\upphi_{\le\Lambda};D_\Lambda\big)
    \;\xrightarrow[\Lambda\to\infty]{L^2(d\upphi_Q)}\;
    V\;=:\;t\int g(x):\!\upphi(x)^4\!:_Q dx\;.
\end{equation}
The right side is the smeared Wick power, defined by this limit, since $:\!\upphi(x)^4\!:_Q$ exists only after integration against a test function. The same limit arises from mollifying the field, as in \cite[Prop.~8.5.1]{GlimmJaffe1987}, and $V$ is a function of $\upphi$.
\item Let $\nu$ be the finite-volume Wick-ordered $\phi^4_2$ measure built from the full propagator,
\begin{equation}\label{eq:phi42-standard}
    d\nu\;=\;\frac{e^{-V}}{\int e^{-V}d\upphi_Q}\,d\upphi_Q\;.
\end{equation}
For all $f_1,\dots,f_r\in\cS(\R^2)$ the law of $\big(\langle\upphi_{\le\Lambda},f_k\rangle\big)_{k\le r}$ under $P^{(\infty)}_\Lambda$ converges weakly on $\R^r$, as $\Lambda\to\infty$, to the law of $\big(\langle\upphi,f_k\rangle\big)_{k\le r}$ under $\nu$.  Consequently, almost surely in $\omega^\infty$, the common value $\mathcal{X}^{(\infty)}_\infty$ of~\eqref{eq:interchange} is the characteristic function of $\big(\langle\upphi,f_k\rangle\big)_{k\le r}$ under $\nu$, and the law of $\big(\langle\phi_\infty^{(n)},f_k\rangle\big)_{k\le r}$ under $\mu^{(n)}_\infty$ converges weakly to the same law as $n\to\infty$.
\end{enumerate}
\end{theorem}

\begin{proof}
\textit{Part (a).}  By \Cref{lem:shell-dim}(a)--(c) at $d=2$, the shell kernels obey $\|Q_j\|_{L^\infty}= q_j\le C$ and $\|Q_j\|_{L^1}\le C2^{-2j}$, so
\begin{equation}
    \|Q_j\|_{L^4}^4 \;\le\; \|Q_j\|_{L^\infty}^3\,\|Q_j\|_{L^1}\leq C^42^{-2j}\;,
    \qquad\text{whence}\qquad
    \|Q_j\|_{L^4} \;\le\; C\,2^{-j/2}\;.
\end{equation}
Summing the tail gives $\|Q-Q_{\le\Lambda}\|_{L^4}\le\sum_{j>\Lambda}\|Q_j\|_{L^4}\le c\,2^{-\Lambda/2}$. At $\Lambda=0$ this bound covers $\sum_{j\ge1}Q_j$, and $Q_0$, the Fourier transform of the smooth, compactly supported $S_0$, is a Schwartz function, so $Q\in L^4$.

Because the $\upphi^{(j)}$ in~\eqref{eq:phi-shells} are independent, for $\Lambda\le\Lambda'$ those above $\Lambda$ contribute nothing to the cross-covariance,
\begin{equation}\label{eq:nesting}
    \Cov\big(\upphi_{\le\Lambda}(x),\upphi_{\le\Lambda'}(y)\big)
    \;=\; Q_{\le\Lambda}(x-y)\;,
\end{equation}
and the covariance at coincident points gives $\Var\upphi_{\le\Lambda}(x)=Q_{\le\Lambda}(0)=D_\Lambda$.  For centered jointly Gaussian $X$ and $Y$, Hermite orthogonality \cite[Thm.~3.9]{Janson1997} gives
\begin{equation*}
    \E\big[H_4(X;\Var X)\,H_4(Y;\Var Y)\big]\;=\;4!\,\Cov(X,Y)^4\;.
\end{equation*}
Expanding the square below gives two diagonal terms and a cross term, each a double integral against $g(x)g(y)$ to which this identity applies. The diagonal terms contribute $Q_{\le\Lambda'}^4(x-y)$ and $Q_{\le\Lambda}^4(x-y)$, while the cross term's covariance is set by the smaller truncation by \eqref{eq:nesting}.  Hence,
\begin{equation}\label{eq:wick-cauchy}
    \E_{\upphi_Q}\Big[\big(V(\upphi_{\le\Lambda'};D_{\Lambda'})
    -V(\upphi_{\le\Lambda};D_\Lambda)\big)^{2}\Big]
    \;=\; 4!\,t^2\!\iint g(x)g(y)
    \big[Q_{\le\Lambda'}^4-Q_{\le\Lambda}^4\big](x-y)\,dx\,dy\;.
\end{equation}
It remains to bound the right side of~\eqref{eq:wick-cauchy}.  Since $\iint g(x)g(y)|h(x-y)|\,dx\,dy\le\|g\|_{L^\infty}\|g\|_{L^1}\|h\|_{L^1}$, it is enough to bound the kernel difference in $L^1$.  From $|a^4-b^4|\leq|a-b|(|a|+|b|)^3$ and H\"older with exponents $4$ and $4/3$,
\begin{equation}\label{eq:fourth-power}
    \big\|Q_{\le\Lambda'}^4-Q_{\le\Lambda}^4\big\|_{L^1}
    \;\le\; \big\|Q_{\le\Lambda'}-Q_{\le\Lambda}\big\|_{L^4}
    \big(\|Q_{\le\Lambda'}\|_{L^4}+\|Q_{\le\Lambda}\|_{L^4}\big)^3\;.
\end{equation}
The first factor is at most $\sum_{\Lambda<j\le\Lambda'}\|Q_j\|_{L^4}\le c\,2^{-\Lambda/2}$, while $\|Q_{\le\Lambda}\|_{L^4}\le\|Q\|_{L^4}+c$ keeps the second bounded uniformly in $\Lambda$ and $\Lambda'$. The right side of~\eqref{eq:wick-cauchy} is therefore $O(2^{-\Lambda/2})$ uniformly in $\Lambda'>\Lambda$. Thus,
\begin{equation*}
    \big\|V(\upphi_{\le\Lambda'};D_{\Lambda'})
    -V(\upphi_{\le\Lambda};D_\Lambda)\big\|_{L^2(d\upphi_Q)}
    \;=\;O\big(2^{-\Lambda/4}\big)\;,
\end{equation*}
so the family is Cauchy in $L^2(d\upphi_Q)$ and therefore converges to a limit $V$, which is~\eqref{eq:wick-limit}.

The same computation identifies this limit with the one of \cite[Prop.~8.5.1]{GlimmJaffe1987}, which mollifies the field.  Fix an even $\varrho\in C_c^\infty(\R^2)$, nonnegative with integral one, and set $\varrho_\varepsilon(z):=\varepsilon^{-2}\varrho(z/\varepsilon)$.  The mollified field $\upphi_\varepsilon(x):=\langle\upphi,\varrho_\varepsilon(x-\cdot)\rangle$ is a function of $\upphi$ with covariance $Q_\varepsilon:=\varrho_\varepsilon*Q*\varrho_\varepsilon$, and the corresponding interaction is $V_\varepsilon:=t\int g\,H_4\big(\upphi_\varepsilon;Q_\varepsilon(0)\big)\,dx$.  Because the shells are independent, the truncation and the mollified field have cross-covariance $\Cov\big(\upphi_{\le\Lambda}(x),\upphi_\varepsilon(y)\big)=(Q_{\le\Lambda}*\varrho_\varepsilon)(x-y)$, and Hermite orthogonality gives, as in~\eqref{eq:wick-cauchy},
\begin{equation}\label{eq:molli}
    \E_{\upphi_Q}\Big[\big(V(\upphi_{\le\Lambda};D_\Lambda)-V_\varepsilon\big)^2\Big]
    \;=\;4!\,t^2\!\iint g(x)g(y)\Big[Q_{\le\Lambda}^4-2\,(Q_{\le\Lambda}*\varrho_\varepsilon)^4+Q_\varepsilon^4\Big](x-y)\,dx\,dy\;.
\end{equation}
Convolution with $\varrho_\varepsilon$ is a contraction on $L^4(\R^2)$ that converges to the identity as $\varepsilon\to0$, so together with the tail bound above, all three kernels converge to $Q$ in $L^4(\R^2)$ as $\Lambda\to\infty$ and $\varepsilon\to0$.  Write the bracket as $\big[Q_{\le\Lambda}^4-(Q_{\le\Lambda}*\varrho_\varepsilon)^4\big]+\big[Q_\varepsilon^4-(Q_{\le\Lambda}*\varrho_\varepsilon)^4\big]$.  As in~\eqref{eq:fourth-power}, the $L^1$ norm of each term is at most a constant times the $L^4$ distance between its two kernels, and because these distances tend to zero, the right side of~\eqref{eq:molli} tends to zero. Thus $V$ is also the $L^2$ limit of the $V_\varepsilon$, which is $t$ times the smeared Wick power of \cite[Prop.~8.5.1]{GlimmJaffe1987}.

Each $V_\varepsilon$ is a function of $\upphi$, and therefore so is their $L^2$ limit.

\medskip\noindent\textit{Part (b).}  Fix $f_1,\dots,f_r\in\cS(\R^2)$ and a bounded continuous $F:\R^r\to\R$.  By~\eqref{eq:mu-std-def} the measure $P^{(\infty)}_\Lambda$ has density $e^{-V(\upphi_{\le\Lambda};D_\Lambda)}/Z^{(\infty)}_\Lambda$ against the law of $\upphi_{\le\Lambda}$, so the Gibbs ratio may be read on the space carrying $\upphi$,
\begin{equation}\label{eq:LP-ratio}
    \E_{P^{(\infty)}_\Lambda}\big[F\big]
    \;=\;\frac{\E_{\upphi_Q}\big[F\big(\langle\upphi_{\le\Lambda},f_k\rangle\big)\,
    e^{-V(\upphi_{\le\Lambda};D_\Lambda)}\big]}
    {\E_{\upphi_Q}\big[e^{-V(\upphi_{\le\Lambda};D_\Lambda)}\big]}\;.
\end{equation}
Equation~\eqref{eq:phi42-standard} presents $\E_{\nu}[F]$ as the same ratio with $\upphi$ in place of $\upphi_{\le\Lambda}$ and $V$ in place of $V(\upphi_{\le\Lambda};D_\Lambda)$.  With
\begin{equation*}
    \xi_\Lambda\;:=\;F\big(\langle\upphi_{\le\Lambda},f_k\rangle\big)\,
    e^{-V(\upphi_{\le\Lambda};D_\Lambda)}\;,
    \qquad
    \xi\;:=\;F\big(\langle\upphi,f_k\rangle\big)\,e^{-V}\;,
\end{equation*}
it suffices to show that $\xi_\Lambda\to\xi$ in $L^1(d\upphi_Q)$, and likewise at $F\equiv1$ for the denominator, since $|\E_{\upphi_Q}[\xi_\Lambda]-\E_{\upphi_Q}[\xi]|\le\|\xi_\Lambda-\xi\|_{L^1(d\upphi_Q)}$.  By~\cite[Thm.~5.12]{Kallenberg2021} at $p=1$, convergence in $L^1$ follows from convergence in probability together with uniform integrability of the $|\xi_\Lambda|$.

Convergence in probability holds factor by factor.  The difference $\langle\upphi-\upphi_{\le\Lambda},f_k\rangle$ is centered, with second moment
\begin{equation}
    \E_{\upphi_Q}\big|\langle\upphi-\upphi_{\le\Lambda},f_k\rangle\big|^2
    \;=\;\int\frac{d^2p}{(2\pi)^2}
    \Big(\sum_{j>\Lambda}S_j(p)\Big)\big|\hat f_k(p)\big|^2
    \;\le\;\int_{|p|\ge2^{\Lambda}}\frac{d^2p}{(2\pi)^2}\,
    \frac{\big|\hat f_k(p)\big|^2}{p^2+m^2}\;,
\end{equation}
because the shells above $\Lambda$ are supported in $|p|\ge2^{\Lambda}$.  The right side is the tail of a convergent integral and tends to zero as $\Lambda\to\infty$.  By Chebyshev's inequality, $\langle\upphi_{\le\Lambda},f_k\rangle\to\langle\upphi,f_k\rangle$ in probability for each $k$, and since $F$ is continuous, $F(\langle\upphi_{\le\Lambda},f_k\rangle)\to F(\langle\upphi,f_k\rangle)$ in probability.  Part~(a) gives $V(\upphi_{\le\Lambda};D_\Lambda)\to V$ in $L^2(d\upphi_Q)$, hence in probability, so $e^{-V(\upphi_{\le\Lambda};D_\Lambda)}\to e^{-V}$ in probability, and the product $\xi_\Lambda$ converges to $\xi$ in probability.

A family bounded in $L^2$ is uniformly integrable, so it suffices to bound the second moments of $\xi_\Lambda$ and of its $F\equiv1$ counterpart uniformly in $\Lambda$.

Since $|\xi_\Lambda|\le\|F\|_\infty\,e^{-V(\upphi_{\le\Lambda};D_\Lambda)}$ and squaring $e^{-V}$ doubles the coupling,
\begin{equation}\label{eq:LP-UI}
    \E_{\upphi_Q}|\xi_\Lambda|^2
    \;\le\;\|F\|_\infty^2\,\E_{\upphi_Q}\big[e^{-2V(\upphi_{\le\Lambda};D_\Lambda)}\big]
    \;=\;\|F\|_\infty^2\,Z^{(\infty)}_\Lambda(2t)\;.
\end{equation}
It is therefore enough to bound $Z^{(\infty)}_\Lambda(2t)$ uniformly in $\Lambda$.  Since $Z^{(\infty)}_\Lambda(2t)$ does not involve the features, it suffices to bound it along one realization of $\omega^\infty$ in the probability-one events of \Cref{lem:unif-gg} and \Cref{thm:infwidth}. At each width, \Cref{thm:d2-convergence} at coupling $2t$ and the base bound $Z^{(n)}_0(2t)\le e^{2c_V}$ from \Cref{lem:V-lower} and~\eqref{eq:d2-C1} give
\begin{equation*}
    \log Z^{(n)}_\Lambda(2t)
    \;\le\;\log Z^{(n)}_0(2t)+C_T(2t)\sum_{j\ge1}2^{-2j}j^3
    \;\le\;2c_V+C_T(2t)\sum_{j\ge1}2^{-2j}j^3\;,
\end{equation*}
where $C_T(2t)$ is the constant of \Cref{prop:d2-exp-int} at coupling $2t$, free of $n$ by \Cref{prop:uniform-N}.  The right side depends on neither $\Lambda$ nor $n$.  Because $Z^{(\infty)}_\Lambda(2t)$ is the width limit of $Z^{(n)}_\Lambda(2t)$ (\Cref{thm:infwidth}), the same bound holds for $\log Z^{(\infty)}_\Lambda(2t)$.  Both families are therefore uniformly integrable.  Together with the convergence in probability, we have $\xi_\Lambda\to\xi$ in $L^1(d\upphi_Q)$, and likewise for the denominators.

The denominators converge to $\int e^{-V}d\upphi_Q$.  Because each $V(\upphi_{\le\Lambda};D_\Lambda)$ is Wick ordered at the variance of $\upphi_{\le\Lambda}$ and has mean zero, its $L^2$ limit $V$ has mean zero too, and Jensen's inequality gives $\int e^{-V}d\upphi_Q\ge 1$.  Hence $\E_{P^{(\infty)}_\Lambda}[F]\to\E_{\nu}[F]$, which is the stated weak convergence on $\R^r$.

Applied to the real and imaginary parts of $u\mapsto\exp(i\sum_ks_ku_k)$, this weak convergence shows that $\mathcal{X}^{(\infty)}_\Lambda$ converges pointwise to the characteristic function of $\big(\langle\upphi,f_k\rangle\big)_{k\le r}$ under $\nu$, and~\eqref{eq:interchange} identifies that limit as $\mathcal{X}^{(\infty)}_\infty$.  The second weak limit is then the one given by \Cref{thm:interchange} for the other order of limits.
\end{proof}

\Cref{thm:LP-standard} writes $\nu$ in~\eqref{eq:phi42-standard} without showing that its density defines a probability measure.  The following corollary proves that it does and states the convergence for the neural-network measures $\mu^{(n)}_\Lambda$ themselves.

\begin{corollary}[The finite-volume $\phi^4_2$ measure]
\label{cor:d2-construction}
The measure $\nu$ of~\eqref{eq:phi42-standard} is a probability measure on $\cS'(\R^2)$.  It is absolutely continuous with respect to the free-field measure $\upphi_Q$, with density $e^{-V}/\int e^{-V}d\upphi_Q$ and normalization $\int e^{-V}d\upphi_Q\in[1,\infty)$.  Almost surely in $\omega^\infty$, for every $r$, all $f_1,\dots,f_r\in\cS(\R^2)$, and either order of the limits $\Lambda\to\infty$ and $n\to\infty$, the law of the smeared vector $\big(\langle\phi_\Lambda^{(n)},f_k\rangle\big)_{k\le r}$ under the neural-network measure $\mu^{(n)}_\Lambda$ of \Cref{def:nested} converges weakly on $\R^r$ to the law of $\big(\langle\upphi,f_k\rangle\big)_{k\le r}$ under $\nu$.
\end{corollary}

\begin{proof}
The proof of \Cref{thm:LP-standard}(b) gives $e^{-V(\upphi_{\le\Lambda};D_\Lambda)}\to e^{-V}$ in $L^1(d\upphi_Q)$, and shows that
\begin{equation*}
    1\;\le\;\int e^{-V}\,d\upphi_Q\;=\;\lim_{\Lambda\to\infty}Z^{(\infty)}_\Lambda\;<\;\infty\;.
\end{equation*}
By \Cref{thm:LP-standard}(a), $V$ is a function of the field $\upphi$, so $e^{-V}/\int e^{-V}d\upphi_Q$ is a nonnegative measurable function on $\cS'(\R^2)$ with integral one against $\upphi_Q$, and~\eqref{eq:phi42-standard} defines a probability measure absolutely continuous with respect to $\upphi_Q$.

For the convergence, \Cref{thm:interchange} gives weak convergence along either order of limits to the law with characteristic function $\mathcal{X}^{(\infty)}_\infty$, and \Cref{thm:LP-standard}(b) identifies that law as the law of $\big(\langle\upphi,f_k\rangle\big)_{k\le r}$ under $\nu$.
\end{proof}

\section{Open problems}\label{sec:open}

The construction reaches finite volume in two dimensions.  We consider in turn what higher dimensions, infinite volume, the Osterwalder--Schrader axioms, and other models would require.

\subsection{Three and four dimensions}\label{sec:higher-dim}

Above two dimensions Wick ordering no longer removes every divergence, and the interaction~\eqref{eq:V-def} must be supplemented.  In three dimensions the one-shell perturbation of the quartic still decays, like $t\,2^{-j}$~\eqref{eq:quartic-pert-size}, but the interaction needs a mass counterterm (\Cref{rem:no-counterterms}).  The one-shell analysis of \Cref{sec:one-shell}, carried out for the Wick-ordered quartic alone, would have to be extended to include it.  One path forward is a counterterm that runs with the shell index, chosen at each shell to cancel the mass shift from integrating out that shell.  What the cancellation leaves involves the kernels of several shells at once, and it could be bounded by keeping track of the scale of each.

Scalar $\phi^4$ is known to be trivial in four dimensions, in the sense that the continuum limits of its lattice regularizations are Gaussian \cite{Aizenman2021}, as they are in higher dimensions \cite{Aizenman1981,Frohlich1982}.  An NN-FT proof of triviality may be within reach at small coupling, where convexity is available.  Because the one-shell perturbation is $O(t)$ at every shell in four dimensions, the Hessian bound makes every tilted measure uniformly log-concave once the coupling is small (\Cref{prop:hessian}).  The argument of \Cref{cor:tilted-moments} then holds the fluctuations of every shell field to the free-field scale uniformly in the background.  Unlike in two dimensions, the shell integrations do not become negligible.  At one loop, each lowers the effective quartic coupling by an amount proportional to its square, and summed over the shells these amounts drive the coupling at any fixed scale to zero as the cutoff is removed. Whether convexity of the tilted measure gives enough control to turn the one-loop picture into a proof of triviality is open.

\subsection{Infinite volume}

Every estimate of \Cref{sec:d2-setup,sec:d2-converge} holds at a fixed spatial cutoff $g$, and several of their constants grow with the volume.  The worst-case bound $e^{-V_J}\le e^{c_V(J+1)^2}$ of~\eqref{eq:d2-C1} has $c_V$ proportional to $\|g\|_1$, and through the returns to the free measure it enters the prefactor $\mathcal{K}$ of \Cref{thm:d2-moments} exponentially.  Because $\nu$ is the standard measure, infinite volume can be approached by any method, and the question here is whether NN-FT can reach it itself.  One route would be to make the estimates local in space, but each neuron of the present architecture is a plane wave spread over all of $\R^2$.  Features localized in position as well as in frequency, such as wavelets, would keep the shell structure while making each neuron local.  Whether the one-shell analysis survives that substitution is open.

\subsection{The Osterwalder--Schrader axioms}

Most of the axioms, Euclidean invariance and cluster decomposition among them, require the thermodynamic limit, but reflection positivity can be tested on the regularized measures themselves.  At infinite width and every finite cutoff $\Lambda$, the free field already fails it.  Its law $P^{0,(\infty)}_\Lambda$ is Gaussian with a continuous, Euclidean-invariant covariance kernel $Q_{\le\Lambda}$ that is finite at coincident points.  By a theorem of Harvey \cite[Thm.~2.1]{Harvey2026}, a field with these properties is reflection positive only if its two-point function is constant.  Ferko and Mutchler \cite{FerkoMutchler2026} trace the failure to Lehmann's spectral formula \cite[Thm.~6.2.4]{GlimmJaffe1987}, by which the two-point function of a Euclidean-invariant, reflection-positive theory is a superposition of free propagators, each of which diverges at coincident points in two or more dimensions. Both results say that a Euclidean-invariant, reflection-positive field in two or more dimensions has either a constant two-point function or infinite variance at a point.

However, Harvey's theorem does not apply to the quenched measures at finite width because they are not Euclidean invariant (\Cref{sec:shell-fields}).  The lattice regularization gives up rotation invariance but preserves reflection positivity \cite[\S9.6]{GlimmJaffe1987}, and whether an NN-FT construction can similarly retain reflection positivity at finite width or at finite cutoff is open.

\subsection{Other models}

The framework is available whenever the field can be organized into dyadic shells, each parametrized linearly by Gaussian coefficients.  The architecture must also give a sampled covariance that concentrates around the shell kernel, as in \Cref{thm:good-geom}, and a shell kernel whose row integral is $O(2^{-2j})$.  The nearest case is $P(\phi)_2$ \cite{Simon1974}, whose interaction is the Wick-ordered integral of a polynomial of degree $2n$ with positive leading coefficient.  Its second derivative in the shell field is at least $-C(1+\mathcal{D}_j)^{n-1}$, with $\mathcal{D}_j$ as in \Cref{prop:hessian} and $C$ depending only on the polynomial.  Rescaling the field by $\mathcal{D}_j^{1/2}$ turns that derivative into $\mathcal{D}_j^{\,n-1}$ times a polynomial of even degree with a fixed positive leading coefficient and bounded lower coefficients.  Such a polynomial is bounded below uniformly in $\mathcal{D}_j$.  The one-shell perturbation is then $O(j^{n-1}2^{-2j})$, and uniform log-concavity above a threshold follows as in \Cref{cor:BL}.  We expect the rest of the argument of \Cref{sec:d2-setup,sec:d2-converge} to follow the same lines, with higher moments of the background.

Whether the construction extends to gauge theories or sigma models is open.

\section*{Acknowledgments}
The author thanks Jim Halverson and Christian Ferko for helpful discussions, and especially Thomas Harvey for feedback on a previous version. This work is supported in part by the National Science Foundation under Cooperative Agreement PHY-2525568 (The NSF AI Institute for Artificial Intelligence and Fundamental Interactions). The author used Claude Opus 5.5 to draft and revise text, to check proofs, and to compare the paper's statements about the cited literature with the sources. The author reviewed and verified all such output and takes full responsibility for the content of this paper.

\bibliographystyle{unsrt}
\bibliography{main_d2_refs}

\clearpage
%=====================================================

\appendix
\crefalias{section}{appendix}
%=====================================================

\section{The free field and its recovery from the architecture}
\label{app:free-field}

This appendix constructs the Gaussian free field together with its decomposition into independent shells, which \Cref{sec:d2-identify} uses (\Cref{app:gff}).  It then shows that the architecture of \Cref{sec:arch} reproduces the free field once the widths are sent to infinity, in every dimension $d\ge2$ and with no renormalization (\Cref{app:recovery}).

Throughout, $d\ge2$, the test functions $f$ are real-valued Schwartz functions, and $\langle h,f\rangle:=\int h(x)f(x)\,d^dx$ is the pairing of \Cref{sec:d2-existence} in $d$ dimensions.  We write $\hat{f}(p):=\int f(x)\,e^{-ip\cdot x}\,d^dx$ for the Fourier transform, and
\begin{equation}\label{eq:fQ-norm}
    \|f\|_Q^2 \;:=\; \int_{\R^d}\frac{d^dp}{(2\pi)^d}\,
    \frac{|\hat f(p)|^2}{p^2+m^2}
    \;=\;\iint f(x)\,Q(x-y)\,f(y)\,d^dx\,d^dy
\end{equation}
for the variance that a Gaussian field of covariance $Q$ assigns to $\langle\,\cdot\,,f\rangle$.  It is finite because $\hat f$ is a Schwartz function and $m>0$.  Half the corresponding variance for the shell kernel $Q_j$ is
\begin{equation}\label{eq:theta-j}
    \theta_j \;:=\; \frac12\int_{\R^d}\frac{d^dp}{(2\pi)^d}\,
    S_j(p)\,\big|\hat f(p)\big|^2 \;\ge\;0\;,
    \qquad
    \sum_{j\ge0}\theta_j \;=\; \tfrac12\|f\|_Q^2\;,
\end{equation}
where the second identity follows from~\eqref{eq:dyadic} by monotone convergence, since the $S_j$ are nonnegative.

\subsection{The free field and its shells}\label{app:gff}

\begin{proposition}[The free field and its shell decomposition]\label{prop:gff}
\leavevmode
\begin{enumerate}[label=(\alph*)]
\item There is a unique probability measure $\upphi_Q$ on $\cS'(\R^d)$ with
\begin{equation}\label{eq:minlos}
    \int_{\cS'(\R^d)}e^{i\langle\upphi,f\rangle}\,d\upphi_Q
    \;=\; e^{-\|f\|_Q^2/2}
    \qquad\text{for every }f\in\cS(\R^d)\;,
\end{equation}
namely the Gaussian free field of covariance $Q=(-\Delta+m^2)^{-1}$.
\item The shell fields can be realized on one probability space as independent centered Gaussian fields $\upphi^{(j)}$ with covariances $Q_j$, together with a field $\upphi$ of law $\upphi_Q$, such that for every $f\in\cS(\R^d)$ the series $\sum_{j\ge0}\langle\upphi^{(j)},f\rangle$ converges to $\langle\upphi,f\rangle$ almost surely and in $L^2$.  The truncations $\upphi_{\le\Lambda}=\sum_{j\le\Lambda}\upphi^{(j)}$ have covariance $Q_{\le\Lambda}$.  This is the decomposition~\eqref{eq:phi-shells}.
\end{enumerate}
\end{proposition}

\begin{proof}
\textit{Part (a).}  By~\eqref{eq:fQ-norm}, $f\mapsto\|f\|_Q^2$ is a positive semidefinite quadratic form on $\cS(\R^d)$, so $f\mapsto e^{-\|f\|_Q^2/2}$ is positive definite and equals $1$ at $f=0$.  The Bochner--Minlos theorem \cite[Thm.~3.4.2]{GlimmJaffe1987} also asks that this functional be continuous on $\cS(\R^d)$.  Because $\|\cdot\|_Q$ is a norm and $(p^2+m^2)^{-1}\le m^{-2}$, for $f,f'\in\cS(\R^d)$
\begin{equation*}
    \big|\,\|f\|_Q-\|f'\|_Q\,\big|\;\le\;\|f-f'\|_Q
    \;\le\;m^{-1}\,\|f-f'\|_{L^2}\;,
\end{equation*}
and $\|\cdot\|_{L^2}$ is continuous on $\cS(\R^d)$.  The Bochner--Minlos theorem therefore gives a unique probability measure $\upphi_Q$ satisfying~\eqref{eq:minlos}, and by~\eqref{eq:fQ-norm} it is the centered Gaussian measure of covariance $Q$.

\medskip\noindent\textit{Part (b).}  Part~(a) with $S_j$ in place of $(p^2+m^2)^{-1}$ gives, for each $j$, a centered Gaussian measure on $\cS'(\R^d)$ with covariance $Q_j$.  Take the $\upphi^{(j)}$ independent with these laws.  A finite sum of independent centered Gaussian fields is centered Gaussian with the sum of the covariances, which for $\upphi_{\le\Lambda}$ is $Q_{\le\Lambda}$.

To construct $\upphi$, fix $s>d/2$, set $w(x):=(1+|x|^2)^{-d}$, and let $H$ be the Hilbert space of distributions $u$ with $\|u\|_H:=\|(1-\Delta)^{-s/2}(wu)\|_{L^2}<\infty$, which embeds continuously in $\cS'(\R^d)$.  Since the $S_j$ sum to $(p^2+m^2)^{-1}\le m^{-2}$ by~\eqref{eq:dyadic}, the covariance operators satisfy $\sum_jQ_j=Q\le m^{-2}$ on $L^2(\R^d)$, and therefore
\begin{equation*}
    \sum_{j\ge0}\E\big\|\upphi^{(j)}\big\|_H^2
    \;\le\;m^{-2}\,\|w\|_{L^2}^2\int_{\R^d}\frac{d^dp}{(2\pi)^d}\,(1+p^2)^{-s}
    \;<\;\infty\;.
\end{equation*}
Because the shells are independent and centered, for $\Lambda<\Lambda'$
\begin{equation*}
    \E\big\|\upphi_{\le\Lambda'}-\upphi_{\le\Lambda}\big\|_H^2
    \;=\;\sum_{\Lambda<j\le\Lambda'}\E\big\|\upphi^{(j)}\big\|_H^2\;,
\end{equation*}
a tail of the convergent series above.  The truncations are therefore Cauchy in mean square in $H$, and they converge in mean square to an $H$-valued random field $\upphi$.  Pairing with $f$ is continuous on $H$, so $\sum_{j\le\Lambda}\langle\upphi^{(j)},f\rangle\to\langle\upphi,f\rangle$ in $L^2$.  The summands are independent and centered, with variances $2\theta_j$ summing to $\|f\|_Q^2$ by~\eqref{eq:theta-j}, so the series also converges almost surely, and $\langle\upphi,f\rangle$ is centered Gaussian with variance $\|f\|_Q^2$.  By the uniqueness in part~(a), the law of $\upphi$ is $\upphi_Q$.
\end{proof}

\subsection{Recovery from the architecture}\label{app:recovery}

The architecture is that of \Cref{sec:shell-fields}, with the widths $N^{(n)}_j$ of the nested family~\eqref{eq:width-schedule}.  In this subsection $\E$ denotes the \emph{annealed} expectation, which averages over the coefficients and the features together, in contrast to the quenched $\E_{a_j}$ of \Cref{not:quench}.  For $f\in\cS(\R^d)$, the annealed characteristic functional of the field at cutoff $\Lambda$ is
\begin{equation}\label{eq:char-factorize}
    \upchi^{(n)}_\Lambda(f)
    \;:=\; \E\big[e^{i\langle\phi_\Lambda,f\rangle}\big]
    \;=\; \prod_{j=0}^{\Lambda}
    \E\big[e^{i\langle\varphi_j,f\rangle}\big]\;,
\end{equation}
where the factorization holds because the shells are independent.  It is written upright, like the free-field objects of \Cref{sec:d2-identify}, to distinguish it from the quenched $\mathcal{X}$ of \Cref{sec:interchange}.  At finite width the smeared field is not Gaussian, since it is a mixture over the features.

\begin{lemma}[Shell characteristic function and its width error]
\label{lem:shell-char}
For $f\in\cS(\R^d)$, write $I(b,c):=\int f(x)\cos(b\cdot x+c)\,d^dx$, and let $\E_{b,c}$ average over $(b,c)\sim\pi_j\otimes\mathrm{Unif}([-\pi,\pi])$.  Then
\begin{equation}\label{eq:shell-char}
    \E\big[e^{i\langle\varphi_j,f\rangle}\big]
    \;=\;\Big(\E_{b,c}\big[
    e^{-q_jI(b,c)^2/N^{(n)}_j}\big]\Big)^{N^{(n)}_j}
    \;>\;0\;,
\end{equation}
and its logarithm differs from $-\theta_j$, the logarithm of the free-field value, by at most a constant times $1/N^{(n)}_j$,
\begin{equation}\label{eq:shell-rate}
    0 \;\le\; \log\E\big[e^{i\langle\varphi_j,f\rangle}\big]+\theta_j
    \;\le\; \frac{\|f\|_{L^1}^2\,\theta_j}{2}\cdot
    \frac{q_j}{N^{(n)}_j}\;.
\end{equation}
\end{lemma}

\begin{proof}
Write $N:=N^{(n)}_j$ and $X:=q_j\,I(b,c)^2\ge0$, so that the $N$ features of shell $j$ give independent copies $X_1,\dots,X_N$ of $X$.  Conditionally on the features, the smeared shell field
\begin{equation*}
    \langle\varphi_j,f\rangle
    \;=\;\sum_{\ell=1}^{N}a_{j,\ell}\,
    \sqrt{\frac{2q_j}{N\sigma^2}}\;I(b_{j,\ell},c_{j,\ell})
\end{equation*}
is a centered Gaussian of variance $\tfrac{2}{N}\sum_\ell X_\ell$, so its conditional characteristic function is $\prod_\ell e^{-X_\ell/N}$.  Averaging over the independent features gives $\E[e^{i\langle\varphi_j,f\rangle}]=\big(\E_{b,c}[e^{-X/N}]\big)^N$, which is~\eqref{eq:shell-char}.

Since $I(b,c)=\tfrac12\big(e^{ic}\hat f(-b)+e^{-ic}\hat f(b)\big)$, $\E_c[e^{\pm2ic}]=0$, and $\hat f(-b)=\overline{\hat f(b)}$ for real $f$, only the cross term of $I(b,c)^2$ survives the average over $c$, and $\E_c[I(b,c)^2]=\tfrac12|\hat f(b)|^2$.  Averaging over $b\sim\pi_j$ and multiplying by $q_j$ gives $\E_{b,c}[X]=\theta_j$, so the quantity in~\eqref{eq:shell-rate} is $N\log\E_{b,c}[e^{-X/N}]+\E_{b,c}[X]$.  Jensen's inequality for $u\mapsto e^{-u}$ shows that it is nonnegative.  For the upper bound, $e^{-u}\le1-u+\tfrac12u^2$ on $u\ge0$ gives
\begin{equation*}
    \E_{b,c}\big[e^{-X/N}\big]
    \;\le\; 1-\frac{\E_{b,c}[X]}{N}+\frac{\E_{b,c}[X^2]}{2N^2}\;,
\end{equation*}
and the right side is positive because it dominates the left.  Taking logarithms with $\log(1+w)\le w$ and multiplying by $N$ bounds the quantity by $\E_{b,c}[X^2]/(2N)$.  Finally, $|I(b,c)|\le\|f\|_{L^1}$ gives $X\le q_j\|f\|_{L^1}^2$, hence $\E_{b,c}[X^2]\le q_j\|f\|_{L^1}^2\,\theta_j$, which is~\eqref{eq:shell-rate}.
\end{proof}

\begin{theorem}[Free field recovery]\label{thm:free-field}
For every $f\in\cS(\R^d)$, both iterated limits of $\upchi^{(n)}_\Lambda(f)$ exist and agree,
\begin{equation}\label{eq:char-limit}
    \lim_{\Lambda\to\infty}\lim_{n\to\infty}
    \upchi^{(n)}_\Lambda(f)
    \;=\;\lim_{n\to\infty}\lim_{\Lambda\to\infty}
    \upchi^{(n)}_\Lambda(f)
    \;=\; e^{-\|f\|_Q^2/2}\;,
\end{equation}
which by~\eqref{eq:minlos} is the characteristic functional of the free field $\upphi_Q$.
\end{theorem}

\begin{proof}
By~\eqref{eq:shell-char}, every factor in~\eqref{eq:char-factorize} is strictly positive, so $\upchi^{(n)}_\Lambda(f)>0$, and~\eqref{eq:char-factorize} gives
\begin{equation*}
    r_{n,\Lambda}\;:=\;\log\upchi^{(n)}_\Lambda(f)
    +\sum_{j\le\Lambda}\theta_j
    \;=\;\sum_{j\le\Lambda}\Big(
    \log\E\big[e^{i\langle\varphi_j,f\rangle}\big]+\theta_j\Big)\;.
\end{equation*}
By~\eqref{eq:shell-rate} each summand is nonnegative, so $r_{n,\Lambda}$ is nondecreasing in $\Lambda$, and summing the upper bounds with $\sum_j\theta_j=\tfrac12\|f\|_Q^2$ gives
\begin{equation}\label{eq:ann-rate}
    0 \;\le\; r_{n,\Lambda}
    \;\le\; \frac{\|f\|_{L^1}^2\,\|f\|_Q^2}{4}\,
    \sup_{j\ge0}\frac{q_j}{N^{(n)}_j}\;.
\end{equation}
By \Cref{lem:shell-dim}(a), $q_j\le C\,2^{(d-2)j}$, and the schedule~\eqref{eq:width-schedule} gives $N^{(n)}_j\ge nc_0(j{+}1)2^{2dj}$, so
\begin{equation*}
    \frac{q_j}{N^{(n)}_j}
    \;\le\;\frac{C}{n\,c_0}\,2^{-(d+2)j}
    \;\le\;\frac{C}{n\,c_0}\;.
\end{equation*}
Hence $0\le r_{n,\Lambda}\le C'/n$ with $C'$ free of $\Lambda$.  Taking $n\to\infty$ first, $r_{n,\Lambda}\to0$ at each fixed $\Lambda$, so the inner limit of $\log\upchi^{(n)}_\Lambda(f)$ is $-\sum_{j\le\Lambda}\theta_j$, and the outer one is $-\tfrac12\|f\|_Q^2$.  Taking $\Lambda\to\infty$ first, $r_{n,\Lambda}$ is nondecreasing and bounded, so it converges at each fixed $n$ to some $r_n\in[0,C'/n]$.  The inner limit is then $-\tfrac12\|f\|_Q^2+r_n$, and letting $n\to\infty$ gives $-\tfrac12\|f\|_Q^2$.  Exponentiating gives~\eqref{eq:char-limit}.
\end{proof}

\section{Derivation of the interpolation derivative}
\label{app:interpolation}

We derive~\eqref{eq:psi-prime}, with $\langle\,\cdot\,\rangle_\alpha$ the tilted expectation of \Cref{sec:d2-special} and the conditioning suppressed as in \Cref{not:quench}.

Recall that $\Psi_j(\alpha)=\log\E_{a_j}[e^{-X_{j,\alpha}}]$, with $X_{j,\alpha}$ as in~\eqref{eq:Xj-alpha}.  The integrand has $\alpha$-derivative $-(\partial_\alpha X_{j,\alpha})\,e^{-X_{j,\alpha}}$.  On an interval $[\epsilon,1]$ with $\epsilon>0$, the first factor is a polynomial in $a_j$ whose coefficients are bounded uniformly in $\alpha$, and the second is at most $e^{C_B}$ by \Cref{lem:X-lower}.  The derivative of the integrand is therefore bounded on all of $[\epsilon,1]$ by $e^{C_B}$ times a single polynomial in $|a_j|$, and that bound is integrable under the Gaussian prior.  By dominated convergence, differentiation under the expectation is justified on $[\epsilon,1]$, hence on $(0,1]$. The endpoint at $\alpha=0$ is recovered at the end of this appendix.  Differentiating,
\begin{equation}\label{eq:psi-deriv}
    \Psi_j'(\alpha)
    \;=\; -\big\langle\partial_\alpha X_{j,\alpha}
    \big\rangle_\alpha\;.
\end{equation}
The interpolation enters $X_{j,\alpha}$ only through the Wick constant $\alpha\,q^\omega_j$ and the field $\sqrt{\alpha}\,\varphi_j$, with
\begin{equation}\label{eq:alpha-derivs}
    \partial_\alpha\big(\alpha\,q^\omega_j\big) = q^\omega_j,
    \qquad
    \partial_\alpha\big(\sqrt\alpha\,\varphi_j\big)
    = \frac{\varphi_j}{2\sqrt\alpha}\;,
\end{equation}
so the chain rule, with $\partial_z T=4H_3^{(\alpha)}$ and $\partial_q T=-6H_2^{(\alpha)}$ from~\eqref{eq:H4-derivs} (both evaluated at $B=B_{j-1}(x)$, $D=D_{j-1}^\omega(x)$, $q=\alpha q^\omega_j(x)$, and $z=\sqrt{\alpha}\,\varphi_j(x)$), gives
\begin{equation}\label{eq:dalpha-Y4}
    \partial_\alpha X_{j,\alpha}
    \;=\; t\int g(x)\,\bigg[
    \frac{2\,\varphi_j(x)}{\sqrt{\alpha}}\,
    H_3^{(\alpha)}(x)
    - 6\,q^\omega_j(x)\,H_2^{(\alpha)}(x)\bigg]\,dx\;,
\end{equation}
where $\Phi_\alpha(x):=B_{j-1}(x)+\varphi_{j,\alpha}(x)$ and $H_n^{(\alpha)}(x):=H_n\big(\Phi_\alpha(x);\, D_{j-1}^\omega(x)+\alpha\,q^\omega_j(x)\big)$.  The first term is the field contribution (carrying the $\alpha^{-1/2}$ from~\eqref{eq:alpha-derivs}), and the second is the Wick-constant-shift contribution.

To simplify the field term we use Gaussian integration by parts.  For the prior $\gamma_\sigma=\cN(0,\sigma^2 I)$ and the tilt $e^{-X_{j,\alpha}}$, every polynomial $F(a_j)$ satisfies
\begin{equation}\label{eq:IBP}
    \big\langle a_{j,\ell}\,F\big\rangle_\alpha
    \;=\; \sigma^2\,\big\langle\partial_\ell F
    \big\rangle_\alpha
    - \sigma^2\,\big\langle F\,
    \partial_\ell X_{j,\alpha}
    \big\rangle_\alpha\;,
    \qquad
    \partial_\ell:=\partial/\partial a_{j,\ell}\;.
\end{equation}
With $\varphi_j(x)=\sum_\ell a_{j,\ell}\,U_{j,\ell}(x)$, the field term is of the form $\langle\varphi_j(x)\,F(x)\rangle_\alpha =\sum_\ell U_{j,\ell}(x)\,\langle a_{j,\ell}\,F(x)\rangle_\alpha$, and~\eqref{eq:IBP} splits $\langle a_{j,\ell}\,F(x)\rangle_\alpha$ into a first term $\sigma^2\langle\partial_\ell F\rangle$ and a second term $-\sigma^2\langle F\,\partial_\ell X\rangle$.  We now show that the first term cancels the Wick-shift contribution and that the second gives~\eqref{eq:psi-prime}.

In \eqref{eq:dalpha-Y4}, take $F=H_3^{(\alpha)}(x)$. Using $\partial_u H_3(u;C)=3H_2(u;C)$ and $\partial_\ell\Phi_\alpha(x)=\sqrt{\alpha}\,U_{j,\ell}(x)$, we have $\partial_\ell H_3^{(\alpha)}(x) =3\sqrt{\alpha}\,H_2^{(\alpha)}(x)\,U_{j,\ell}(x)$, so multiplying by $\sigma^2 U_{j,\ell}(x)$, summing over $\ell$, and using $\sigma^2\sum_\ell U_{j,\ell}(x)^2=q^\omega_j(x)$,
\begin{equation}\label{eq:first-IBP-quartic}
    \frac{2}{\sqrt{\alpha}}\sum_\ell U_{j,\ell}(x)\,\sigma^2\,
    \big\langle\partial_\ell H_3^{(\alpha)}(x)
    \big\rangle_\alpha
    \;=\; 6\,q^\omega_j(x)\,
    \big\langle H_2^{(\alpha)}(x)\big\rangle_\alpha\;,
\end{equation}
which exactly cancels the $-6\,q^\omega_j\,\langle H_2^{(\alpha)}\rangle_\alpha$ Wick-shift term.

Only the second term of~\eqref{eq:IBP} now survives:
\begin{equation}\label{eq:pre-assembly}
    \Psi_j'(\alpha)
    \;=\; \frac{2t\,\sigma^2}{\sqrt{\alpha}}\int g(x)
    \sum_\ell U_{j,\ell}(x)\,
    \big\langle H_3^{(\alpha)}(x)\,
    \partial_\ell X_{j,\alpha}\big\rangle_\alpha\,dx\;.
\end{equation}
Since $X_{j,\alpha}$ depends on $a_{j,\ell}$ only through $z=\sqrt\alpha\,\varphi_j$, with $\partial_z T=4H_3^{(\alpha)}$ and $\partial_\ell z=\sqrt\alpha\,U_{j,\ell}$,
\begin{equation}
    \partial_\ell X_{j,\alpha}
    \;=\; 4t\sqrt{\alpha}\int g(y)\,
    H_3^{(\alpha)}(y)\,
    U_{j,\ell}(y)\,dy\;.
\end{equation}
In \eqref{eq:pre-assembly}, the factors $2t\sigma^2/\sqrt\alpha$ and $4t\sqrt\alpha$ combine to $8t^2\sigma^2$, and $\sigma^2\sum_\ell U_{j,\ell}(x)\,U_{j,\ell}(y) =\Gamma_j^\omega(x,y)$ gives
\begin{equation}
    \Psi_j'(\alpha)
    \;=\; 8t^2\iint g(x)\,g(y)\,\Gamma_j^\omega(x,y)\,
    \big\langle H_3^{(\alpha)}(x)\,
    H_3^{(\alpha)}(y)\big\rangle_\alpha\,dx\,dy,
\end{equation}
which is~\eqref{eq:psi-prime}.

The factor $\alpha^{-1/2}$ in~\eqref{eq:dalpha-Y4} cancels against the $\sqrt\alpha$ from $\partial_\ell\Phi_\alpha=\sqrt\alpha\,U_{j,\ell}$, once in the first integration-by-parts term ($2/\sqrt\alpha\times\sqrt\alpha$) and once in the assembly ($2t\sigma^2/\sqrt\alpha\times 4t\sqrt\alpha$).  The right side of~\eqref{eq:psi-prime} is therefore continuous on the closed interval $[0,1]$, and together with $\Psi_j(0)=0$ and the continuity of $\Psi_j$ on $[0,1]$, this justifies the identity $\Psi_j(1)=\int_0^1\Psi_j'(\alpha)\,d\alpha$ used in \Cref{thm:d2-one-shell}.\qed

\section{The supremum of the background field}
\label{app:sup-tail}

\begin{proof}[Proof of \Cref{lem:d2-sup-tail}]
Under $\gamma_{J-1}$, for fixed $x,y\in K$ the background $B_{J-1}(x)$ and the increment $B_{J-1}(x)-B_{J-1}(y)$ are centered Gaussians with variances
\begin{subequations}\label{eq:d2-sup-var}
\begin{align}
    \Var\big(B_{J-1}(x)\big)
    &\;=\;D^\omega_{J-1}(x)\;\le\; c_qJ\;,
    \label{eq:d2-sup-var-a}\\
    \Var\big(B_{J-1}(x)-B_{J-1}(y)\big)
    &\;\le\; \Upsilon_{J-1}\,|x-y|^2\;.
    \label{eq:d2-sup-var-b}
\end{align}
\end{subequations}
The first is~\eqref{eq:sigmaJ} together with $D^\omega_{J-1}\le c_qJ$, and the second is \Cref{lem:incr-var} at index $J-1$.

For $k\ge J$ let $\Delta_k\subset K$ be a $2^{-k}$-net over $K$ with $|\Delta_k|\le c_K4^k$. We emphasize that the index $k$ labels spatial resolution, not a shell. Fix $x\in K$ and let $x_k\in\Delta_k$ be a nearest net point, so that $|x-x_k|\le2^{-k}$ and $|x_k-x_{k-1}|\le3\cdot2^{-k}$.  Since $B_{J-1}$ is a finite sum of cosines, it is continuous, so $B_{J-1}(x_k)\to B_{J-1}(x)$ and
\begin{equation*}
    B_{J-1}(x) \;=\; B_{J-1}(x_J)
    +\sum_{k>J}\Big[B_{J-1}(x_k)
    -B_{J-1}(x_{k-1})\Big]\;.
\end{equation*}
Split the threshold $s$ into $s/2$ for $B_{J-1}(x_J)$ and $u_k:=\tfrac{s}{8}\,2^{-(k-J)/2}$ for the $k$-th increment, so that $\sum_{k>J}u_k\le s/2$.  If $|B_{J-1}(x)|>s$, then either $\max_{\Delta_J}|B_{J-1}|>s/2$, or, for some $k>J$, a pair of points from $\Delta_{k-1}$ and $\Delta_k$ at distance at most $3\cdot2^{-k}$ has increment exceeding $u_k$ in absolute value. Neither event depends on $x$, so one of them occurs whenever $\sup_K|B_{J-1}|>s$.

By~\eqref{eq:d2-sup-var-a} and the Gaussian tail $\gamma(|Z|>u)\le2e^{-u^2/(2v)}$ for a centered Gaussian $Z$ with variance~$v$,
\begin{equation*}
    \gamma_{J-1}\Big(\{\max_{y\,\in\Delta_J}|B_{J-1}(y)|>s/2\}\Big)\leq2c_K4^Je^{-s^2/(8c_qJ)}\;.
\end{equation*}
At $k=J+m$ ($m\ge1$), there are at most $|\Delta_k|\cdot|\Delta_{k-1}|\leq c_K^2\,4^{2k}$ pairs, each with increment variance at most $9\Upsilon_{J-1}4^{-k}$ by~\eqref{eq:d2-sup-var-b}.  Then the probability that one of these pairs has increment exceeding $u_k$ in absolute value is at most
\begin{equation*}
    2c_K^2\,16^{J}\,16^{m}\,e^{-\xi\,2^{m}},
    \qquad
    \xi:=\frac{s^2}{3072\,q_{\max}}\;,
\end{equation*}
using $u_k^2\,4^{k}=\tfrac{s^2}{64}\,4^{J}2^{m}$ and $18\Upsilon_{J-1}\le48\,q_{\max}4^J$.  Set
\begin{equation*}
    c_2:=\frac{1}{3072}\;,
    \qquad
    C_K:=\max\{16,\,2c_K+4c_K^2\}\;,
\end{equation*}
and distinguish two cases according to the size of $\xi$.

If $\xi\ge2\ln16$, then $m\ln16\le\tfrac{\xi}{2}2^m$ gives $16^me^{-\xi2^m}\le e^{-\xi2^{m-1}}\le e^{-\xi m}$, and the levels sum to at most $\sum_{m\geq1}e^{-\xi m}\leq2e^{-\xi}$.  Collecting, and using $q_{\max}\le c_qJ$ in the second exponent,
\begin{equation*}
    \gamma_{J-1}\Big(\sup_K|B_{J-1}|>s\Big)
    \;\le\; 2c_K\,4^Je^{-s^2/(8c_qJ)}
    +4c_K^2\,16^{J}e^{-s^2/(3072\,q_{\max})}
    \;\le\; C_K\,16^J
    \exp\!\Big(-\frac{c_2\,s^2}{c_qJ}\Big)\;.
\end{equation*}
If instead $\xi<2\ln16$, then $c_2s^2/(c_qJ)\le\xi<2\ln16$. The right side of \eqref{eq:d2-sup-tail} then already exceeds $1$, so the bound holds trivially.
\end{proof}

\end{document}